\documentclass[11pt,a4paper,colorlinks=true,urlcolor=blue,citecolor=red]{article}

\usepackage[T1]{fontenc}

\usepackage[boldsans]{concmath} %
\usepackage{euler} %

\usepackage[boxruled,vlined]{algorithm2e}
\usepackage{amsmath,amsfonts,amssymb,amsthm,dsfont,bm}
\usepackage{enumitem,marvosym,graphicx}
\usepackage[justification=centering]{subfig}
\usepackage[margin=1in]{geometry}
\usepackage{thmtools}\usepackage{thm-restate}
\usepackage{hyperref}
\graphicspath{{figures/}}

\usepackage{color,soul} %
\newcommand{\xxx}[1]{}

\usepackage{refcount}

\newcommand{\tr}{\ensuremath{\mathrm t}} %
\newcommand{\rootof}{\ensuremath{\mathop{\uparrow}}} %
\newcommand{\p}{\ensuremath{\mathrm p}} %
\newcommand{\treeat}[1]{B\langle{#1}\rangle} %
\newcommand{\treeatt}[2]{B{#1}\langle{#2}\rangle} %
\usepackage[nofancy]{svninfo} %

\svnInfo $Id: treepack.tex 6088 2016-03-24 19:10:48Z hoffmann $

\theoremstyle{plain}
\newtheorem{theorem}{Theorem}
\newtheorem{lemma}[theorem]{Lemma}
\newtheorem{proposition}[theorem]{Proposition}

\newtheorem{observation}[theorem]{Observation}
\newcommand{\HRule}{\rule{\linewidth}{0.5mm}}
\newenvironment{proofof}[1]{%
  \par\medskip\noindent\textbf{\sffamily Proof of #1.}~}{\qed\par\medskip}

\newcommand{\subsubparagraph}[1]{\paragraph{#1}}
\newcommand{\case}[1]{\par\vspace{.5\baselineskip}\noindent\textbf{\sffamily Case~#1}}
\newcommand{\EB}{\mathrm{E}(B)}

\title{The Planar Tree Packing Theorem ({R\svnInfoRevision})}

\begin{document}

\begin{titlepage}
  \begin{center}
    \HRule \\[0.4cm]
    {\huge The Planar Tree Packing Theorem}\\[0cm]
    \HRule \\[1cm]

    \begin{minipage}{0.45\textwidth}
      \begin{center} \large
        \textsc{Markus Geyer}\\
        \small
        Universit\"{a}t T\"{u}bingen, Germany\\
        \verb|geyer@informatik.uni-tuebingen.de|\\[0.4cm]
        \large
        \textsc{Michael Kaufmann}\\
        \small
        Universit\"{a}t T\"{u}bingen, Germany\\
        \verb|mk@informatik.uni-tuebingen.de|\\[0.4cm]
      \end{center}
    \end{minipage}
    \begin{minipage}{0.45\textwidth}
      \begin{center} \large
        \textsc{Michael Hoffmann$^*$}\\
        \small 
        ETH Z\"{u}rich, Switzerland\\
        \verb|hoffmann@inf.ethz.ch|\\[0.4cm]
        \large
        \textsc{Vincent Kusters$^*$}\\
        \small 
        ETH Z\"{u}rich, Switzerland\\
        \verb|vincent.kusters@inf.ethz.ch|\\[0.4cm]
      \end{center}
    \end{minipage}
    \rule{\linewidth}{0mm} \\[0.6cm]
    \begin{minipage}{0.45\textwidth}
      \begin{center} \large
        \textsc{Csaba D. T\'oth$^{\dag}$}\\
        \small \small California State University Northridge\\ Los
        Angeles, CA, USA\\
        \verb|cdtoth@acm.org|
      \end{center}
    \end{minipage}
  \end{center}

  \vspace{\baselineskip}

  \begin{center}
    {\large 
      \svnToday}
  \end{center}

  \vspace{\baselineskip}

  \begin{abstract}
    Packing graphs is a combinatorial problem where several given graphs
    are being mapped into a common host graph such that every edge is
    used at most once. In the planar tree packing problem we are given
    two trees $T_1$ and $T_2$ on $n$ vertices and have to find a planar
    graph on $n$ vertices that is the edge-disjoint union of $T_1$ and
    $T_2$. A clear exception that must be made is the star which cannot
    be packed together with any other tree. But according to a
    conjecture of Garc\'ia et al.\ from 1997 this is the only exception,
    and all other pairs of trees admit a planar packing. Previous
    results addressed various special cases, such as a tree and a spider
    tree, a tree and a caterpillar, two trees of diameter four, two
    isomorphic trees, and trees of maximum degree three. Here we settle
    the conjecture in the affirmative and prove its general form, thus
    making it the planar tree packing theorem. The proof is constructive
    and provides a polynomial time algorithm to obtain a packing for two
    given nonstar trees.
  \end{abstract}

  \vfill

  \begin{center}
    \includegraphics{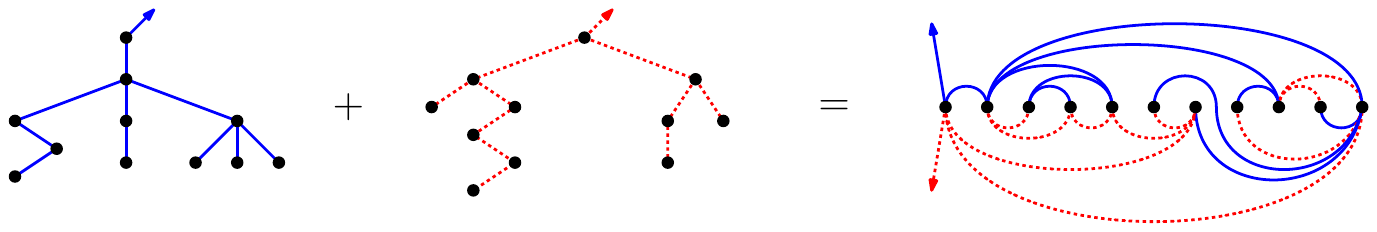}
  \end{center}

  \vfill

 \noindent
 \scriptsize$^*$ Supported by the ESF EUROCORES programme EuroGIGA, CRP
 GraDR and the Swiss National Science Foundation, SNF
 Project 20GG21-134306.\\
 \scriptsize$^{\dag}$ Supported by the NSF awards CCF-1422311 and
 CCF-1423615.

\end{titlepage}

\section{Introduction}\label{sec:introduction}
The \emph{packing problem} is to find a graph $G$ on $n$ vertices that
contains a given collection $G_1,\ldots, G_k$ of graphs on $n$ vertices
each as edge-disjoint subgraphs. This problem has been studied in a wide
variety of scenarios (see, e.g., \cite{AkiyamaC90,CaroY97,FrSz}). Much
attention has been devoted to the packing of trees (e.g., tree packing
conjectures by Gy\'arfas~\cite{gl-ptdok-78} and by Erd\H{o}s and
S\'os~\cite{e-epgt-65}). Hedetniemi~\cite{MR629868} proved that any two
nonstar trees can be packed into $K_n$. Teo and Yap~\cite{ty-ptgo-90}
showed, extending an earlier result by Bollob\'as and
Eldridge~\cite{be-pgacc-78}, that \emph{any} two graphs of maximum
degree at most $n-1$ with a total of at most $2n-2$ edges pack into
$K_n$ unless they are one of thirteen specified pairs of graphs. Maheo
et al{.}~\cite{msw-1996} characterized triples of trees that can be
packed into $K_n$.

In the \emph{planar packing} problem the graph $G$ is required to be
planar. Garc\'ia et al.~\cite{ghhnt-2002} conjectured in~1997 that there
exists a planar packing for any two nonstar trees, that is, for any two
trees with diameter greater than two. The assumption that none of the
trees is a star is necessary, since a star uses all edges incident to
one vertex and so there is no edge left to connect that vertex in the
other tree. Garc\'ia et al{.} proved their conjecture when one
of the trees is a path and when the two trees are isomorphic. Oda and
Ota~\cite{oo-2006} addressed the case that one of the trees is a
caterpillar or that one of the trees is a spider of diameter at most
four. A \emph{caterpillar} is a tree that becomes a path when all leaves
are deleted and a \emph{spider} is a tree with at most one vertex of
degree greater than two. Frati~et~al.~\cite{j-fgk-pptst-08} gave an
algorithm to construct a planar packing of any spider with any
tree. Frati~\cite{f-ppdft-09} proved the conjecture for the case that
both trees have diameter at most four. Finally,
Geyer~et~al.~\cite{gkh-ppbt-13} proved the conjecture for binary trees
(maximum degree three). In this paper we settle the general conjecture
in the affirmative:
\begin{theorem}\label{thm:planar_packing}
  Every two nonstar trees of the same size admit a planar packing.
\end{theorem}

\subsubparagraph{Related work.} Finding subgraphs with specific
properties within a given graph or more generally determining
relationships between a graph and its subgraphs is one of the most
studied topics in graph theory. The \emph{subgraph isomorphism}
problem~\cite{Epp-JGAA-99,GareyJ79,Ullmann76} asks to find a subgraph
$H$ in a graph $G$. The \emph{graph thickness} problem~\cite{mutzel}
asks for the minimum number of planar subgraphs which the edges of a
graph can be partitioned into. The \emph{arboricity}
problem~\cite{Epp-IPL-94} asks to determine the minimum number of
forests which a graph can be partitioned into. Another related classical
combinatorial problem is the $k$ edge-disjoint spanning trees problem
which dates back at least to Tutte~\cite{t-pdgncf-61} and
Nash-Williams~\cite{nw-edstfg-61}, who gave necessary and sufficient
conditions for the existence of $k$ edge-disjoint spanning trees in a
graph. The interior edges of every maximal planar graph can be
partitioned into three edge-disjoint trees, known as a \emph{Schnyder
  wood}~\cite{s-pgpd-89}. Gon\c{c}alves~\cite{1060666} proved that every
planar graph can be partitioned in two edge-disjoint outerplanar graphs.

The study of relationships between a graph and its subgraphs can also be
done the other way round. Instead of decomposing a 
graph, one can ask for a graph $G$ that encompasses a given set of
graphs $G_1,\ldots,G_k$ and satisfies some additional properties. This
topic occurs with different flavors in the computational geometry and
graph drawing literature. It is motivated by applications in
visualization, such as the display of networks evolving over time and
the simultaneous visualization of relationships involving the same
entities. In the \emph{simultaneous embedding}
problem~\cite{BrassCDEEIKLM07}
the graph $G=\bigcup G_i$ is given and the goal is to draw it so that
the drawing of each $G_i$ is plane. The \emph{simultaneous embedding
  without mapping} problem~\cite{BrassCDEEIKLM07} is to find a graph $G$
on $n$ vertices such that: (i) $G$ contains all $G_i$'s as subgraphs,
and (ii) $G$ can be drawn with straight-line edges so that the drawing
of each $G_i$ is plane.

\section{Notation and Overview}\label{sec:def_overview}

A \emph{rooted tree} is a directed tree $T$ with exactly one vertex of
outdegree zero: its root, denoted $\rootof(T)$. Every vertex $v\ne\rootof(T)$ has
exactly one outgoing edge $(v,\p_T(v))$. The target $\p_T(v)$ is the
\emph{parent} of $v$ in $T$, and conversely $v$ is a \emph{child} of
$\p_T(v)$. In figures we denote the root of a tree by an outgoing
vertical arrow. For a vertex $v$ of a rooted tree $T$, denote by
$\tr_T(v)$ the \emph{subtree rooted at $v$}, that is, the subtree of $T$
induced by the vertices from which $v$ can be reached on a directed
path. The subscript is sometimes omitted if $T$ is clear from the
context. A \emph{subtree of (or below) $v$} is a tree $\tr_T(c)$, for a
child $c$ of $v$ in $T$. For a tree $T$, denote by $|T|$ the \emph{size}
(number of vertices) of $T$. We denote by $\deg_T(v)$ the degree
(indegree plus outdegree) of $v$ in $T$. For a graph $G$ we denote by
$\mathrm{E}(G)$ the edge set of $G$. A \emph{star} is a tree on $n$
vertices that contains at least one vertex of degree $n-1$. Such a
vertex is a \emph{center} of the star. A star on $n\ne 2$ vertices has a
unique center. For a star on two vertices, both vertices act as a
center. When considered as a rooted tree, there are two different rooted
stars on $n\geq 3$ vertices. A star rooted at a center is called
\emph{central-star}, whereas a star rooted at a leaf
that is not a center is called a \emph{dangling star}. In particular,
every star on one or two vertices is a central-star. A \emph{nonstar} is
a graph that is not a star. A \emph{substar} of a graph is a subgraph
that is a star.
A \emph{one-page book embedding} of a graph $G$ is an embedding of $G$
into a closed halfplane such that all vertices are placed on the
bounding line. This line is called the \emph{spine} of the book
embedding.

We embed vertices equidistantly along the positive $x$-axis and refer to
them by their $x$-coordinate, that is, $P=\{1,\ldots,n\}$. An
\emph{interval} $[i,j]$ in $P$ is a sequence of the form
$i,i+1,\ldots,j$, for $1\le i\le j\le n$, or $i,i-1,\ldots,j$, for
$1\le j\le i\le n$. Observe that we consider an interval $[i,j]$ as
oriented and so we can have $i>j$. Denote the \emph{length} of an
interval $[i,j]$ by $|[i,j]|=|i-j|+1$.  A \emph{suffix} of an interval
$[i,j]$ is an interval $[k,j]$, for some $k\in[i,j]$. To avoid
notational clutter we often identify points from $P$ with vertices
embedded at them.


\subsubparagraph{Overview.} We 
construct a plane drawing of two $n$-vertex trees $T_1$ and $T_2$ on the point set
$P=[1,n]$. We call $T_1$ the \emph{blue tree}; its edges are shown as
solid blue arcs in figures. The tree $T_2$ is called the \emph{red
  tree}; its edges are shown as dotted red arcs. The algorithm first
computes a preliminary one-page book embedding of $T_1$ onto $P$ (the
\emph{blue embedding}) in Section~\ref{sec:emb_t1}. In the second step
we recursively construct an embedding for the red tree to pair up with
the blue embedding. In principle we follow a similar strategy as in the
first step, but we take the constraints imposed by the blue embedding
into account. During this process we may reconsider and change the blue
embedding locally. For instance, we may \emph{flip} the embedding of
some subtree of $T_1$ on an interval $[i,j]$, that is, reflect the
embedded tree at the vertical line $x=\frac{i+j}{2}$ through the
midpoint of $[i,j]$. In some cases we also perform more drastic changes
to the blue embedding. In particular, the blue embedding may not be a
one-page book embedding in the final packing. Although neither of the
two trees $T_1$ and $T_2$ we start with is a star, it is possible---in
fact, unavoidable---that stars appear as subtrees during the recursion.
We have to deal with stars explicitly whenever they arise, because the
general recursive step works for nonstars only. We introduce the
necessary concepts and techniques in Section~\ref{sec:preT2} and give
the actual proof in Section~\ref{sec:embedding_the_red_tree}.


\section{A preliminary blue embedding}\label{sec:emb_t1}
We begin by defining a preliminary one-page book embedding
$\pi: V_1 \to [1,n]$ for a tree $T_1=(V_1,E_1)$ rooted at $r_1\in V$.
In every recursive step, we are given a tree $T$ rooted at a vertex $r$
and an interval $[i,j]$ of length $|T|$. Recall that we may have $i<j$
or $i>j$. We place $r$ at position $i$ and recursively embed the
subtrees of $r$ on pairwise disjoint subintervals of
$[i,j]\setminus\{i\}$. The embedding is guided by two rules illustrated
in \figurename~\ref{fig:emb_T_1}.
\begin{itemize}
\item The \emph{larger-subtree-first rule} (LSFR) dictates that for any
  two subtrees of $r$, the larger of the subtrees must be embedded on an
  interval closer to $r$. Ties are broken arbitrarily.
\item The \emph{one-side rule} (1SR) dictates that for every vertex
  all neighbors 
  are mapped to the same side. 
  That is, if $N_T(v)$ denotes the set of neighbors of $v$ in $T$
  (including its parent), then either $\pi(u)<\pi(v)$ for all
  $u\in N_T(v)$ or $\pi(u)>\pi(v)$ for all $u\in N_T(v)$.
\end{itemize}
These rules imply that every subtree $T\subseteq T_1$ is embedded onto
an interval $[i,j]\subseteq[1,n]$ so that $\{i,j\}$ is an edge of $T$
and either $i$ or $j$ is the root of $T$. Together with $\pi(r_1)=1$,
these rules define the embedding (up to tiebreaking). An explicit
formulation of the algorithm can be found as
Algorithm~\ref{alg:embed_t1} below and an example is depicted in
\figurename~\ref{fig:emb_T_1:3}.
\begin{figure}[htbp]
  \centering%
  \subfloat[LSFR]{\includegraphics{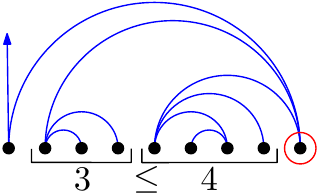}\label{fig:emb_T_1:1}}\hfil%
  \subfloat[1SR]{\includegraphics{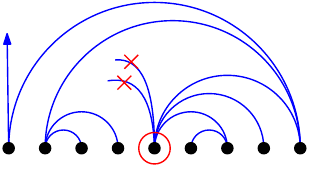}\label{fig:emb_T_1:2}}\hfil%
  \subfloat[A preliminary blue embedding.]{\includegraphics{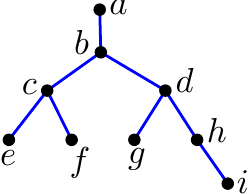}\hspace{2em}\includegraphics{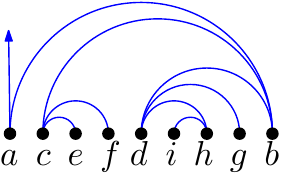}\label{fig:emb_T_1:3}}%
  \caption{Illustrations for the two rules and an example embedding.\label{fig:emb_T_1}}
\end{figure}

\begin{algorithm}[H] \label{alg:embed_t1}%
  \newcommand{\alet}{\leftarrow}%
  \newcommand{\id}[1]{\mathit{#1}}%
  \DontPrintSemicolon%
  \KwIn{A rooted tree $T=(V,E)$ and a directed interval $I\subseteq
    [1,n]$ with $|T|=|I|$.}%
  \KwOut{A map $\pi : V \rightarrow I$.}%

  Let $r$ be the root of $T$ and let $[i,j] = I$.\;%
  $\pi(r)\alet i$\;%
  \If{$|T|>1$}{%
    Let $r_1,\ldots,r_k$ be the children of $r$ in $T$ such that
    $|\tr_T(r_1)|\ge\ldots\ge|\tr_T(r_k)|$.\;%
    $\Sigma_0\alet 0$\;%
    \For{$h=1,\ldots,k$}{%
      $\Sigma_h=\sum_{b=1}^h|\tr_T(r_b)|$.\;%
    }
    \eIf{$i<j$}{%
      \For{$h=1,\ldots,k$}{%
        $\id{Embed}(\tr_T(r_h),[i+\Sigma_h,i+\Sigma_{h-1}+1])$\;%
      }
    }{%
      \For{$h=1,\ldots,k$}{%
        $\id{Embed}(\tr_T(r_h),[i-\Sigma_h,i-\Sigma_{h-1}-1])$\;%
      }
    }
  }
  \caption{$\id{Embed}(T,I)$.}
\end{algorithm}

\section{A red tree and a blue forest}\label{sec:preT2}

As common with inductive proofs, we prove a stronger statement than
necessary. 
This stronger statement does not hold
unconditionally but we need to impose some restrictions on the input.
The goal of this section is to derive this more general
statement---formulated as Theorem~\ref{thm:main}---from which
Theorem~\ref{thm:planar_packing} follows easily.

Our algorithm receives as input a nonstar subtree $R$ of the red tree
and an interval $I=[i,j]$ of size $|R|$ along with a blue graph $B$
embedded on $I$. Without loss of generality we assume $i<j$. In the
initial call $B$ is a tree, but in a general recursive call $B$ is a
\emph{blue forest} that may consist of several components. For
$k\in[i,j]$ let $\treeat{k}$ denote the component of $B$ that contains
$k$. For $[x,y]\subseteq[i,j]$ let $B[x,y]$ denote the subgraph of $B$
induced by the vertices in 
$[x,y]$, and for $k\in[x,y]$ let $\treeatt{[x,y]}{k}$ denote the
component of $B[x,y]$ that contains $k$.


In general the algorithm sees only a small part of the overall picture
because it has access to the vertices in $I$ only. However, blue
vertices in $I$ may have edges to vertices outside of $I$ and also
vertices of $R$ may have neighbors outside of $I$. We have to ensure
that such \emph{outside edges} are used by one tree only and can be
routed without crossings.
In order to control the effect of outside edges, we allow only one
vertex in each component---that is, the root of $R$ and the root of each
component of $B$---to have neighbors outside of $I$. Whenever we change
the blue embedding we need to maintain the relative order of these roots
so as to avoid crossings among outside edges.

\subsubparagraph{Conflicts.} Typically $r:=\rootof(R)$ has at least one
neighbor outside of $I$: its parent $\p_{T_2}(r)$. But $r$ may also have
children in $T_2\setminus R$. We assume that all neighbors---parent and
children---of $r$ in $T_2\setminus R$ are already embedded outside of
$I$ when the algorithm is called for $R$. There are two principal
obstructions for mapping $r$ to a point $v\in I$:
\begin{itemize}
\item 
  A vertex $v\in I$ is in \emph{edge-conflict} with $r$, if
  $\{v,r'\}\in\mathrm{E}(T_1)$ for some neighbor $r'$ of $r$ in
  $T_2\setminus R$. Mapping $r$ to $v$ would make $\{v,r'\}$ an edge of
  both $T_1$ and $T_2$
  (\figurename~\ref{fig:conflicts_1}--\ref{fig:conflicts_2}). In figures
  we mark vertices in edge-conflict with $r$ by a lightning symbol
  \Lightning.
\item 
  A vertex $v\in I$ is in \emph{degree-conflict} with $r$ on $I$ if
  $\deg_R(r)+\deg_B(v)\ge|I|$. If we map $r$ to $v$, then no child of
  $r$ in $R$ can be mapped to the same vertex as a child of $v$ in
  $B$. With only $|I|-1$ vertices available there is not enough room for
  both groups (\figurename~\ref{fig:conflicts_3}).
\end{itemize}
\begin{figure}[htbp]
  \centering%
  \subfloat[]{\includegraphics{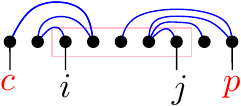}\label{fig:conflicts_1}}\hfil
  \subfloat[]{\includegraphics{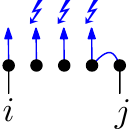}\label{fig:conflicts_2}}\hfil
  \subfloat[]{\includegraphics{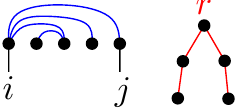}\label{fig:conflicts_3}}\hfil
  \caption{An interval $[i,j]$ on which a tree $R=\tr(r)$ is to be
    embedded. Two neighbors $p$ and $c$ of $r$ in $T_2\setminus R$ are
    already embedded (a). Then the situation on $[i,j]$ presents itself
    as in (b), where the three central vertices are in edge-conflict
    with $r$ due to blue outside edges to $p$ or $c$.  In (c) the vertex
    $i$ is in degree-conflict with $r$ because
    $\deg_R(r)+\deg_B(i)=2+3=5\ge|[i,j]|$. We cannot map $r$ to the blue
    vertex at $i$ because there is not enough room for the neighbors of
    both in $[i,j]$.\label{fig:conflicts}}
\end{figure}

We cannot hope to avoid conflicts entirely and we do not need to. It
turns out that is sufficient to avoid a very specific type of conflict
involving stars.
\begin{itemize}
\item An interval $[i,j]$ is in \emph{edge-conflict}
  (\emph{degree-conflict}) with $R=\tr(r)$ if $B^*:=\treeat{i}$ is a
  central-star and the root of $B^*$ is in edge-conflict
  (degree-conflict) with $r$ (\figurename~\ref{fig:starconflicts}).
\item An interval $I$ is in \emph{conflict} with $R$ if $I$ is in
  edge-conflict or degree-conflict with $R$ (or both).
\end{itemize}

\begin{figure}[htbp]
  \centering%
  \subfloat[]{\includegraphics{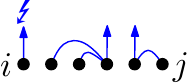}\label{fig:starconflicts_1}}\hfil
  \subfloat[]{\includegraphics{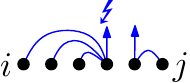}\label{fig:starconflicts_1a}}\hfil
  \subfloat[]{\includegraphics{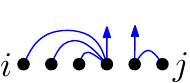}\label{fig:starconflicts_2}}\hfil
  \subfloat[]{\includegraphics{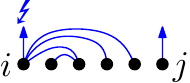}\label{fig:starconflicts_3}}\hfil
  \subfloat[]{\includegraphics{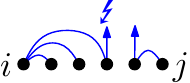}\label{fig:starconflicts_4}}\hfil
  \caption{An interval $[i,j]$ in edge-conflict (a)--(b), and examples
    where $[i,j]$ is not in edge-conflict (c)--(e). In (c) the center of
    $B^*$ is not in edge-conflict; it may be in degree-conflict, though,
    if $\deg_R(r)\ge 3$. In both (d) and (e) the tree $\treeat{i}$ is
    not a central-star.\label{fig:starconflicts}}
\end{figure}


\noindent
The following lemma shows that a degree-conflict cannot be caused by a
very small star.
\begin{restatable}{lemma}{degconthree}\label{lem:degcon3}
  If an interval $[i,j]$ is in degree-conflict with a nonstar subtree
  $R$ of $T_2$, then $\treeat{i}$ is a central-star on at least three
  vertices.
\end{restatable}
\begin{proof}
  By the definition of degree-conflict for $[i,j]$, $B^*:=\treeat{i}$ is
  a central-star. Let $c$ denote its root.  Then a degree-conflict
  implies $\deg_R(r)+\deg_{B^*}(c)\ge|I|$. As $R$ is not a star, we have
  $\deg_R(r)\le|R|-2=|I|-2$. Therefore $\deg_{B^*}(c)\ge 2$, that is,
  $|B^*|\ge 3$.
\end{proof}

We claim that $R$ can be packed with $B$ onto $I$ unless $I$ is in
conflict with $R$. The following theorem presents a precise formulation
of this claim. Only $R$ and the graph $\treeat{i}$ determine whether or
not an interval $[i,j]$ is in conflict with $R$. Therefore we can phrase
the statement without referring to an embedding of $B$ but just
regarding it as a sequence of trees. The set $C$ represents the set of
roots from $B$ that are in edge-conflict with $r$.
\begin{theorem}\label{thm:main}
  Let $R$ be a nonstar tree with $r=\rootof(R)$ and let $B$ be a nonstar
  forest with $|R|=|B|=n$, together with an ordering $b_1,\ldots,b_k$ of
  the $k\in\{1,\ldots,n\}$ roots of $B$ and a set
  $C\subseteq\{b_1,\ldots,b_k\}$. Suppose 
  (i) $\tr_B(b_1)$ is not a central-star or (ii) $b_1\notin C$ and
  $\deg_R(r)+\deg_B(b_1)<n$. Then there is a plane packing $\pi$ of $B$
  and $R$ onto any interval $I$ with $|I|=n$ such that
  \begin{itemize}
  \item $\pi(r)\notin\pi(C)$ and
  \item we can access $b_1,\ldots,b_k,r$ in this order from the outer
    face of $\pi$, that is, we can add a new vertex $v$ in the outer
    face of $\pi$ and route an edge to each of $b_1,\ldots,b_k,r$ such
    that the resulting multigraph is plane and the circular order of
    neighbors around $v$ is $b_1,\ldots,b_k,r$. (If $r=b_i$, for some
    $i\in\{1,\ldots,k\}$, then two distinct edges must be routed from
    $v$ to $r$ so that the result is a non-simple plane multigraph.)
  \end{itemize}
  Such a packing $\pi$ we call an \emph{ordered plane packing} of $B$
  and $R$ onto $I$.
\end{theorem}

\noindent
Theorem~\ref{thm:main} is a strengthening of
Theorem~\ref{thm:planar_packing} and so we obtain
Theorem~\ref{thm:planar_packing} as an easy corollary.
\begin{proofof}{Theorem~\ref{thm:planar_packing} from
    Theorem~\ref{thm:main}}
  Select roots arbitrarily so that $T_1=\tr(r_1)$ and $T_2=\tr(r_2)$.
  Then use Theorem~\ref{thm:main} with $R=T_2$, $B=T_1$, $k=1$,
  $b_1=r_1$, and $C=\emptyset$. By assumption $T_1$ is not a star and so
  (i) holds. Therefore we can apply Theorem~\ref{thm:main} and obtain
  the desired plane packing of $T_1$ and $T_2$.
\end{proofof}

It is not hard to see that forbidding conflicts in
Theorem~\ref{thm:main} is necessary: The example families depicted in
\figurename~\ref{fig:confness} do not admit an ordered plane packing.
\begin{figure}[htbp]
  \centering%
  \subfloat[$b_1\in C$]{\includegraphics{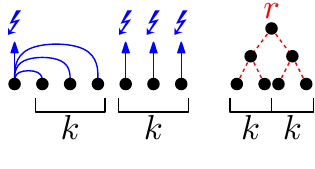}\label{fig:confness:1}}\hfil%
  \subfloat[$\deg_R(r)+\deg_{\tr(b_1)}(b_1)\ge n$]{%
    \hspace{1cm}\includegraphics{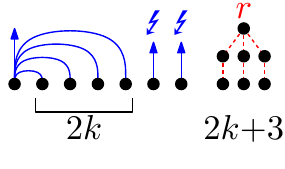}\hspace{1cm}\label{fig:confness:2}}\hfil%
  \caption{The statement of Theorem~\ref{thm:main} does not hold without
    (i) or (ii). In the examples the trees of $B$ are ordered from left
    to right so that $\tr(b_1)$ is a central-star. Vertices in $C$ are
    labeled with \Lightning.\label{fig:confness}}
\end{figure}

\paragraph{Runtime analysis.} The algorithm is parameterized with a
subtree $R$ of $T_2$ and an interval $I\subseteq[1,n]$, which $R$ is to
be packed onto together with an already embedded subforest of $T_1$. If
we represent $T_1$ as an adjacency matrix and the embeddings as arrays,
then after an $O(n^2)$ time initialization we can test in constant time
for the presence of an edge between $i,j\in I$. To represent $T_2$ we
use an adjacency list where the children are sorted by the size of their
subtrees, which can be precomputed in $O(n\log n)$ time. Then at each
step, the algorithm spends $O(|I|)$ time and makes at most two recursive
calls with disjoint sub-intervals of $I$, which yields $O(n^2)$ time
overall.

\section{Embedding the red tree: fundamentals}\label{sec:embedding_the_red_tree}

In this section we discuss some fundamental tools for our recursive
embedding algorithm to prove Theorem~\ref{thm:main}. First we formulate
four invariants that hold for every recursive call of the
algorithm. Next we present three tools that are specific types of
embeddings to handle a ``large'' substar of $B$ or $R$. All of these
embeddings rearrange the given embedding of $B$ to make room for the
center of the star. Finally, we conclude with an outline of the
algorithm.

\subsection{Invariants}

In the algorithm we are given a red tree $R=\tr(r)$, a blue forest $B$
with roots $b_1,\ldots,b_k$, an interval $I=[i,j]\subseteq[1,n]$ with
$|I|=|R|=|B|$, and a set $C$ that we consider to be the vertices from
$B$ in edge-conflict with $r$. As a first step, we embed $B$ onto $I$ by
embedding $\tr(b_1),\ldots,\tr(b_k)$ in this order from left to right,
each time using the algorithm from Section~\ref{sec:emb_t1}.
\begin{observation}\label{obs:invariants}
  We may assume that $R$, $B$ and $I=[i,j]$ satisfy the following
  invariants:\normalfont
  \begin{enumerate}[leftmargin=*,label={(I\arabic*)}]\setlength{\itemindent}{\labelsep}
  \item\label{inv:starconflict} $I$ is not in conflict with $R$.
    \emph{(peace invariant)}
  \item\label{inv:bluelocal} Every component of $B$ satisfies LSFR and
    1SR. All edges of $B$ are drawn in the upper halfplane (above the
    $x$-axis).
    All roots of $B$ are visible from above (that is, a vertical ray
    going up from $b_x$ does not intersect any edge of
    $B$). \emph{(blue-local invariant)}
  \item\label{inv:placement} $i$ is not in edge-conflict with
    $r$. \emph{(placement invariant)}
  \end{enumerate}
\end{observation}
\begin{proof}
  \ref{inv:starconflict} follows from the assumption (i) or (ii) in
  Theorem~\ref{thm:main}. \ref{inv:bluelocal} is achieved by using
  the embedding from Section~\ref{sec:emb_t1}. If $i$ is in conflict
  with $r$, then \ref{inv:starconflict} implies that $\treeat{i}$ is not
  a singleton (which would be a central-star). Therefore flipping
  $\treeat{i}$ establishes \ref{inv:placement} without affecting
  \ref{inv:starconflict} or \ref{inv:bluelocal}.
\end{proof}

Theorem~\ref{thm:main} ensures that all roots of $B$ along with $r$
appear on the outer face in the specified order. We cannot assume that
we can draw an edge to any other vertex of $B$ or $R$ without crossing
edges of the embedding given by Theorem~\ref{thm:main}. Therefore it is
important that whenever the algorithm is called recursively,
\begin{enumerate}[leftmargin=*,label={(I\arabic*)},start=4]%
  \setlength{\itemindent}{\labelsep}
\item\label{inv:rootsonly} only the roots $b_1,\ldots,b_k$ and $r$ have
  edges to the outside of $I$.
\end{enumerate}
Assuming 1SR for $B$ helps when splitting intervals for recursive
treatment.
\begin{observation}\label{obs:bluelocal}
  If $B$ satisfies \ref{inv:bluelocal} and \ref{inv:rootsonly} on an
  interval $I$, then both invariants also hold for $B[x,y]$ on $[x,y]$,
  for every subinterval $[x,y]\subseteq I$.
\end{observation}
In the remainder of the proof we will ensure and assume that invariants
\ref{inv:starconflict}--\ref{inv:rootsonly} hold for every call of the
algorithm. For the initial instance of packing $T_1$ and $T_2$, we know
that \ref{inv:starconflict}--\ref{inv:placement} hold by
Observation~\ref{obs:invariants} and \ref{inv:rootsonly} holds trivially
because there are no vertices outside of $I$.

\subsection{Blue-star embedding}
\label{sec:greedy_grab_embedding}
The blue-star embedding is useful to handle the center $\sigma$ of a
substar $B^*$ of $B$. It explicitly embeds a subtree $A$ of $R$ onto a
part of $B$ that includes $\sigma$. It may use some of the leaves of
$B^*$. After taking care of $\sigma$, any unused leaf of $B^*$ appears
as a locally isolated vertex in the remaining interval of vertices.

The blue-star embedding consists of several steps:
It rearranges some vertices of $B$, moves some edges of $B$ below the
$x$-axis, and introduces edges that straddle both halfplanes above and
below the $x$-axis (\figurename~\ref{fig:greedygrabex}).
\begin{figure}[htbp]
  \centering\hfil%
  \subfloat[$A$]{\includegraphics{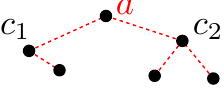}\label{fig:greedygrabex:1}}\hfil
  \subfloat[$B$]{\includegraphics{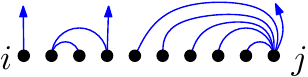}\label{fig:greedygrabex:2}}\hfil
  \subfloat[Result]{\includegraphics{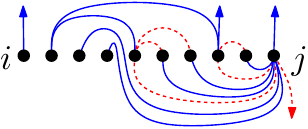}\label{fig:greedygrabex:3}}\hfil
  \caption{blue-star embedding $A$ onto a part of
    $B$.\label{fig:greedygrabex}}
\end{figure}

Suppose that $A$ is a subtree of $R$ with $a:=\rootof(A)$ (possibly
$a=r$) and $\sigma\in[i,j]$ is the center of a star $B^*=\tr_B(\sigma)$.
Either $\sigma$ is the root of $\treeat{\sigma}$ or
$\tau:=\p_B(\sigma)\in[i,j]$. Denote by $B^+$ the subgraph of $B$
induced by $\sigma$ and all its neighbors (parent and children). Note
that either $B^+=B^*$ or $B^+=B^*\cup\{\tau\}$. Put $d=\deg_A(a)$ and
let $\varphi=(v_1,\ldots,v_d)$ be a sequence of elements from
$B\setminus B^+$. Furthermore, suppose the following four conditions
hold:
  \begin{enumerate}[label={(BS\arabic*)}]\setlength{\itemindent}{4\labelsep}
  \item\label{gg:ec} $a$ is not in edge-conflict with $\sigma$,
  \item\label{gg:dc} $|A|\le|B^*|+\deg_A(a)$ and
    $|B^+|+\deg_A(a)\le|R|-1$,
  \item\label{gg:int} at least one of $B\setminus(B^*\cup\varphi)$ or
    $B\setminus(B^+\cup\varphi)$ forms an interval, and
  \item\label{gg:cs} if $B\setminus(B^*\cup\varphi)$ does not form an
    interval, then $A$ is not a central-star, $v_1=\tau\pm 1$ and
    $\{v_1,\tau\}\notin\EB$.
  \end{enumerate}
Note that \ref{gg:cs} is a trivial consequence of \ref{gg:int} in case
$B^*=B^+$. Furthermore, \ref{gg:ec} is trivially satisfied if no
neighbor of $a$ in $T_2$ has been embedded yet.

Let $c_1,\ldots,c_d$ denote the children of $a$ in $A$ such that
$|\tr_R(c_1)|\ge\ldots\ge|\tr_R(c_d)|$. Partition the leaves of $B^*$
into $d+1$ groups $G_1,\ldots,G_{d+1}$ such that $|G_k|=|\tr_R(c_k)|-1$,
for $k\in\{1,\ldots,d\}$, and $|G_{d+1}|=|B^*|-1-\sum_{k=1}^d|G_k|$. We
intend to embed the vertices of $\tr_R(c_k)\setminus\{c_k\}$ on the
leaves in $G_k$. Note that some (possibly all) of the sets $G_k$ may be
empty. Also note that
$\sum_{k=1}^d|G_k|=\sum_{k=1}^d(|\tr_R(c_k)|-1)=|A|-(d+1)$, where the
$+1$ accounts for $a$. Therefore
$|G_{d+1}|=(|B^*|-1)-(|A|-d-1)=|B^*|+d-|A|$ is nonnegative by
\ref{gg:dc} and so our assignment is well-defined.

If $B\setminus(B^*\cup\varphi)$ does not form an interval, then by
\ref{gg:cs} $A$ is not a central-star and so $|G_1|\ge 1$. In this case,
we move one leaf from $G_1$ to $G_{d+1}$ and add $\tau$ to $G_1$
instead.

The \emph{blue-star embedding of $A$ from $\sigma$ with $\varphi$}
proceeds in four steps, as detailed below. The first two steps rearrange
the embedding of $B$ to make room for the embedding of $A$ in the third
step. The fourth step ensures that the remaining unused vertices appear
in a form that allows to further process them.

\subsubparagraph{Step~1 (Flip)} We draw all edges of $B^*$ below the
spine. All edges of $B$ not inside $B^*$ remain above the spine
(\figurename~\ref{fig:greedygrab_1}).

\subsubparagraph{Step~2 (Mix)} Leaving $\sigma$ where it is, we
distribute the leaves of $B^*$ among the vertices in $\varphi$ as
follows: for $k\in\{1,\ldots,d\}$, move the vertices of $G_k$ so that
they appear as a contiguous subsequence immediately to the right of
$v_k$ (\figurename~\ref{fig:greedygrab_3}). If
$B\setminus(B^*\cup\varphi)$ does not form an interval, then we have
$\tau$ in $G_1$. As $\tau$ is not a leaf of $B^*$, we cannot move it
around so easily. Fortunately, no relocation is necessary because by
\ref{gg:cs} $\tau$ appears right next to $v_1$ in $I$. Any
remaining vertices in $G_1$ are placed between $\tau$ and $v_1$.
\begin{figure}[htbp]
  \centering\hfil%
  \subfloat[flip]{\includegraphics{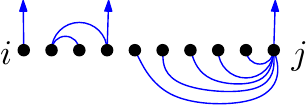}\label{fig:greedygrab_1}}\hfil
  \subfloat[mix]{\includegraphics{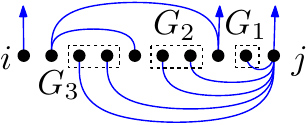}\label{fig:greedygrab_3}}\hfil
  \subfloat[complete]{\includegraphics{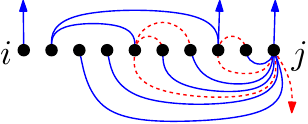}\label{fig:greedygrab_4}}\hfil
  \subfloat[cleanup]{\includegraphics{greedygrab_5}\label{fig:greedygrab_5}}\hfil
  \caption{The example from \figurename~\ref{fig:greedygrabex} in
    detail. We blue-star embed $A$ from $\sigma=j$ where $\varphi$
    takes the vertices of $B\setminus B^+$ from right to
    left.\label{fig:greedygrab}}
\end{figure}

\subsubparagraph{Step~3 (Complete)} Embed $A$ by first mapping $a$ to
$\sigma$, which is possible by \ref{gg:ec}. Next map $c_i$ to $v_i$, for
$i\in\{1,\ldots,d\}$, drawing the edge to $\sigma$ below the spine. Then
embed each subtree $\tr_R(c_i)$ explicitly (using
Algorithm~\ref{alg:embed_t1} and drawing all edges above the spine) on
the interval of $|\tr_R(c_i)|$ locally isolated vertices immediately to
the right of $c_i$ (\figurename~\ref{fig:greedygrab_4}). Note that
$G_1\cup\{v_1\}$ is locally isolated even if
$B\setminus(B^*\cup\varphi)$ does not form an interval because by
\ref{gg:cs} we have $\{v_1,\tau\}\notin\EB$.

It remains to describe the embedding for $G_{d+1}$. Before we do this,
let us consider the properties that we want the embedding to fulfill.
Note that the blue-star embedding---as far as described---does not use
any of the invariants \ref{inv:starconflict}--\ref{inv:bluelocal} other
than that we start from a one-page book embedding. However, if
\ref{inv:starconflict}--\ref{inv:bluelocal} hold for $B$, then we would
like to maintain these invariants also for the part
$B':=B\setminus(\{\sigma\}\cup\varphi\cup\bigcup_{x=1}^d{G_x})$ of $B$
that is not yet used by $R$ after the blue-star embedding. A necessary
prerequisite is that $B'$ forms an interval, that is, the vertices of
$B'$ appear as a contiguous subsequence of $[i,j]$. Given that we are
still free to place the vertices in $G_{d+1}$, it is enough that the
vertices in $B'\setminus G_{d+1}$ form a subinterval of $[i,j]$ that is
reachable from $\sigma$ (without crossing edges).

\subsubparagraph{Step~4 (Cleanup)} Suppose without loss of generality
that $\sigma$ is to the right of $B'\setminus G_{d+1}$. (If $\sigma$ is
to the left of $B'\setminus G_{d+1}$, replace all occurrences of
``right'' by ``left'' in the following paragraph.)

Move the vertices of $G_{d+1}$ so that they appear as a contiguous
subsequence immediately to the right of the rightmost vertex $z$ of
$B'\setminus G_{d+1}$. In order to establish that all edges are drawn
above the spine, we cannot draw the edges between $\sigma$ and $G_{d+1}$
in the same way as we did for $G_1,\ldots,G_d$ above. Instead we route
all edges between $\sigma$ and $G_{d+1}$ as parallel biarcs (curves that
cross the spine once) that leave $\sigma$ below the spine, then cross
the spine just to the right of the rightmost vertex of $G_{d+1}$, and
finally enter their destination from above
(\figurename~\ref{fig:greedygrab_5}). As a result, for the purpose of
embedding some part of $R$ onto $[i,j-|A|]$, the vertices of $G_{d+1}$
become isolated roots; each is connected with a single edge to the
outside that is (locally) routed in the upper halfplane.

This completes the description of the blue-star embedding. Below is a
formal statement summarizing the pre- and postconditions.
\begin{proposition}\label{p:greedygrab}
  Let $A=\tr_R(a)$ be a subtree of $R$, let $\sigma\in[i,j]$ be the
  center of a star $B^*=\tr_B(\sigma)$, and let $\varphi$ be a sequence
  of $\deg_A(a)$ pairwise distinct vertices from $B\setminus B^+$, where
  $B^+$ denotes the subgraph of $B$ induced by $\sigma$ and all its
  neighbors. If $A$ and $\sigma$ fulfill \ref{gg:ec}--\ref{gg:cs}, then
  the blue-star embedding of $A$ from $\sigma$ with $\varphi$ provides
  an ordered plane packing of $A$ onto $[i,j]\setminus[i',j']$, for some
  subinterval $[i',j']\subset[i,j]$.

  Furthermore, $\{x,\sigma\}\notin\EB$ after the blue-star
  embedding, where $x=i'$, if $\sigma>j'$, and $x=j'$, if
  $\sigma<i'$. Put $X=B\setminus(B^*\cup\varphi)$, if
  $B\setminus(B^*\cup\varphi)$ is an interval, and
  $X=B\setminus(B^+\cup\varphi)$, otherwise. Then $[i',j']$ is the union
  of $X$ with some (possibly empty) sequence of isolated vertices on the
  side of $[i',j']$ opposite from $x$.

  Finally, if the embedding of $B$ on $[i,j]$ initially satisfies
  \ref{inv:bluelocal}, then after the blue-star embedding the modified
  embedding of $B$ on $[i',j']$ satisfies \ref{inv:bluelocal}.
\end{proposition}
\begin{proof}
  The packing for $A$ is immediate by construction. Let us first argue
  that after the blue-star embedding an interval $[i',j']\subset[i,j]$
  remains. We distinguish two cases.

  If $B\setminus(B^*\cup\varphi)$ is an interval, then the embedding
  uses exactly the vertices of $(B^*\cup\varphi)\setminus G_{d+1}$, and
  the vertices of $G_{d+1}$ are placed so that they extend the interval
  $B\setminus(B^*\cup\varphi)$.

  Otherwise, $B\setminus(B^*\cup\varphi)$ does not form an
  interval. Then the embedding uses exactly the vertices of
  $(B^+\cup\varphi)\setminus G_{d+1}$ (where one vertex originally in
  $G_1$ is moved to $G_{d+1}$). By \ref{gg:int} we know that
  $B\setminus(B^+\cup\varphi)$ forms an interval and the vertices of
  $G_{d+1}$ are placed so that they extend this interval.

  Next we argue that $\{x,\sigma\}\notin\EB$. By \ref{gg:dc}
  we have $|B\setminus B^+|=|R|-|B^+|\ge d+1$. As $\varphi$ consists of
  $d$ vertices, at least one vertex in $B\setminus B^+$ is not in
  $\varphi$. Due to the way we run the cleanup step, it follows that the
  vertex of $[i',j']$ furthest from $\sigma$ is in $B\setminus B^+$
  (whereas the closest vertex may be in $G_{d+1}$, which is adjacent to
  $\sigma$). By construction no vertex of $B\setminus B^+$ is adjacent
  to $\sigma$ in $B$. The description of $[i',j']$ holds by
  construction.

  It remains to argue that if $B$ satisfies \ref{inv:bluelocal}, then so
  does $B[i',j']$. The blue-star embedding does not change the order
  of the vertices in $B\setminus B^*$ and the vertices of $G_{d+1}$
  become isolated roots. Given the way the edges incident to $G_{d+1}$
  have been drawn, they do not affect the visibility of the roots in
  $B'\setminus G_{d+1}$.  Therefore, \ref{inv:bluelocal} holds for
  $B'\setminus G_{d+1}$. The validity of \ref{inv:bluelocal} for the
  vertices in $G_{d+1}$ follows from the discussion in Step~4 above.
\end{proof}

\subsection{Red-star embedding}
There is a natural counterpart to the blue-star embedding that we call
\emph{red-star embedding}. It embeds a red central-star onto a blue
tree.

Consider an interval $I=[i,j]$ on which we wish to embed a subtree $A^*$
of $R$ that is a 
central-star with $a:=\rootof(A^*)$. Consider some $\sigma\in\{i,j\}$
such that $\sigma$ is the root of $\treeat{\sigma}$.  Let
$k:=\deg_{B}(\sigma)$ and let $v_1,\dots,v_k$ denote the children of
$\sigma$ in $B$, such that $\tr_B(v_1)$ is the subtree closest to
$\sigma$.  Choose any interval $I'\subseteq I\setminus\{\sigma\}$ such
that $\tr_B(v_i)$ is either completely inside or completely outside
$I'$, for every $i\in\{1,\ldots,k\}$. See
\figurename~\ref{fig:sgg_setup}. We require that
\begin{enumerate}[label={(RS\arabic*)}]\setlength{\itemindent}{3em}
\item\label{sgg:ec} $a$ is not in edge-conflict with $\sigma$ and
\item\label{sgg:dc} $\deg_{A^*}(a)+\deg_{B[I'\cup\{\sigma\}]}(\sigma)\le|I'|$.
\end{enumerate}
Note that~\ref{sgg:ec} and~\ref{sgg:dc} are analogous to~\ref{gg:ec} and
\ref{gg:dc}, but only one inequality is needed in~\ref{sgg:dc}. In the
blue-star embedding, we need~\ref{gg:int} and~\ref{gg:cs} to handle
central-stars whose parent is also present in the interval under
consideration. In the red-star embedding, we have no requirements on $B$
other than \ref{sgg:ec} and \ref{sgg:dc}.

\subsubparagraph{Step 1 (Embed)} First embed $a$ at $\sigma$. This works
by~\ref{sgg:ec}. Let $d:=\deg_{A^*}(a)$ and let $c_1,\dots,c_d$ denote
the children of $a$ in $A$. By~\ref{sgg:dc} the interval $I'$ contains
enough vertices not adjacent to $\sigma$ in order to embed
$c_1,\dots,c_d$. Let $N$ be the set of the $d$ closest non-neighbors of
$\sigma$ in $I'$. Embed $c_1,\dots,c_d$ onto $N$. We next describe how
to draw the red edges from $c_1,\dots,c_d$ to $a$. Consider a vertex
$c_i$ and let $v$ be the vertex of the blue forest we embedded $c_i$
onto.  Refer to \figurename~\ref{fig:sgg_embed}. If $v\in\tr_B(v_1)$,
then draw $\{c_i,a\}$ as a semi-circle in the lower halfplane. If
$v\in\tr_B(v_t)$ with $1<t\leq k$ then draw $\{c_i,a\}$ as a biarc that
is in the upper halfplane near $a$, in the lower halfplane near $c_i$,
and crosses the spine between $v_{t-1}$ and $\tr_B(v_t)$.  Finally, if
$v\not\in\tr_B(\sigma)$, then draw $\{c_i,a\}$ as a biarc that is in the
upper halfplane near $a$, in the lower halfplane near $c_i$, and crosses
the spine right after $\tr_B(\sigma)$. Afterwards, the vertices of
$B[I']\setminus N$, i.e. the blue vertices that are not mapped to any
$c_i$, are visible from below.

\begin{figure}
  \centering\hfil%
  \subfloat[Setup.]{\includegraphics{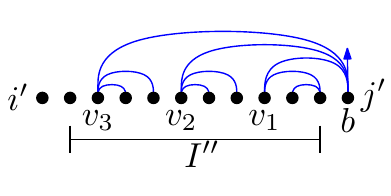}\label{fig:sgg_setup}}\hfil%
  \subfloat[Embed.]{\includegraphics{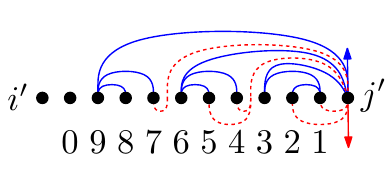}\label{fig:sgg_embed}}\hfil%
  \subfloat[Cleanup.]{\includegraphics{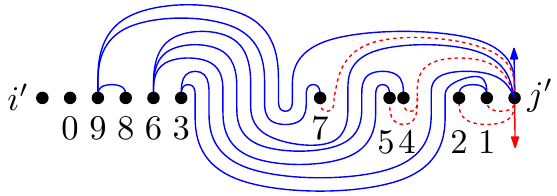}\label{fig:sgg_cleanup}}\hfil%
  \caption{Using the red-star embedding to embed $A^*$ with
    $\deg_{A^*}(a)=5$.}
  \label{fig:sgg}
\end{figure}

\subsubparagraph{Step 2 (Cleanup)} In general, the vertices of
$B[I']\setminus N$ do not form an interval. Assume without loss of
generality that $\sigma$ is the rightmost vertex of $I'$. Let
$N^+=N\cup\{\sigma\}$. We rearrange the vertices on $I'$: from left to
right, we first place all vertices of $B[I']\setminus N^+$ (maintaining
their relative order) and then all vertices of $N^+$ (maintaining their
relative order). Refer to \figurename~\ref{fig:sgg_cleanup}. In
particular, $\sigma$ is still at the rightmost position after this
rearrangement. The edges of $B[I']\setminus N^+$ are drawn as before, as
are the edges of $N^+$. We must redraw the edges that have one end
vertex in $N^+$ and one in $B[I']\setminus N^+$. The edges
$\{v_i,\sigma\}$ are drawn as triarcs: the edge is in the upper
halfplane near $v_i$ and $\sigma$. Its first spine intersection is to
the right of the rightmost vertex of $B[I']\setminus N^+$. Its second
spine intersection is such that it maintains the cyclic order of edges
leaving $\sigma$ (as before the rearrangement). The other edges are
drawn similarly.

The pre- and postconditions of the red-star embedding are
summarized by the following proposition.

\begin{proposition}\label{prop:stargreedygrab}
  Let $I$ be an interval for which $B$ satisfies
  \ref{inv:bluelocal}. Let $A^*=\tr_R(a)$ be a subtree of $R$ that is a
  central-star. Consider some $\sigma\in I$ such that $\sigma$ is the
  root of $\treeat{\sigma}$. Let $k:=\deg_{B}(\sigma)$ and denote the
  children of $\sigma$ in $B$ by $v_1,\dots,v_k$. Let
  $I'\subseteq I\setminus\{\sigma\}$ be any interval such that
  $\tr_B(v_i)$ is either completely inside or completely outside $I'$,
  for every $i\in\{1,\ldots,k\}$.

  If $A^*$ and $\sigma$ and $B[I']$ fulfill~\ref{sgg:ec}
  and~\ref{sgg:dc}, then the red-star embedding of $A^*$ from $\sigma$
  on $I'$ provides an ordered plane packing of $A^*$ onto a subinterval
  $X$ of $I'$. The set $I'\setminus X$ forms an interval that satisfies
  \ref{inv:bluelocal} and consists of $I'\setminus I$ followed by some
  vertices (possibly zero) originally in $B[I']$.
\end{proposition}
\begin{proof}
  As argued above, Step~1 produces a plane packing of $A^*$ and $B[I']$
  by~\ref{sgg:ec} and~\ref{sgg:dc}. Any remaining vertices of
  $\treeatt{[I']}{\sigma}$ remain visible from below. Furthermore, if a
  subtree of $\treeatt{[I']}{\sigma}$ is embedded onto a (directed)
  interval $[x,y]$ with the root at $x$, then Step~1 embeds children of
  $a$ on a (possibly empty) suffix of $[x+1,y]$. Since Step~2 does not
  change the relative position of the remaining vertices of
  $\treeatt{[I']}{\sigma}$ nor the relative position of the other
  vertices in $I'$, the set $I'\setminus X$ satisfies
  \ref{inv:bluelocal} after Step~2.
\end{proof}

\subsection{Leaf-isolation shuffle}
While we are at discussing how to deal with red stars, let us introduce
another basic operation that will turn out useful in this context.

Suppose we need to embed a substar $A^*\subset R$ onto a subinterval
$[a,b]\subset[i,j]$. Then we need to pair the center of $A^*$ with an
isolated vertex in $B[a,b]$. If there is no such vertex, we occasionally
embed $A^*$ onto $[a+1,b+1]$ after a rearrangement of $B$ that ensures
that $B[a+1,b+1]$ has a suitable isolated vertex. The goal of such a
\emph{leaf-isolation shuffle} is to modify $B$ so that a leaf of
$B[a,b]$ is at $a+1$ and its parent is at
$a$. \figurename~\ref{fig:leafshuffle_2} shows the result of performing
a left-isolation shuffle on \figurename~\ref{fig:leafshuffle_0} with
$[a,b]=[1,9]$. The idea is then to take the parent out of the interval
by embedding $A^*$ onto $[a+1,b+1]$ instead and mapping the center of
$A^*$ to $a+1$, which is locally isolated on $[a+1,b+1]$. The
proposition below guarantees that such a leaf-isolation shuffle is
always possible. Note that we do not care about the invariant
\ref{inv:bluelocal} in this scenario because we cannot use a recursive
embedding for a star anyway. There is one part of the invariant that we
need to maintain, though, which is the visibility of the blue root from
above.
\begin{figure}[htbp]
  \centering\hfil%
  \subfloat[]{\includegraphics{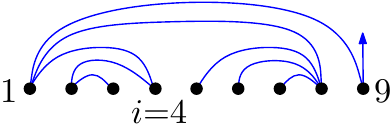}\label{fig:leafshuffle_0}}\hfil
  \subfloat[]{\includegraphics{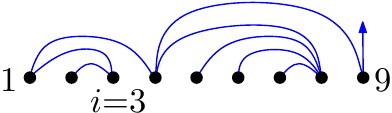}\label{fig:leafshuffle_1}}\hfil
  \subfloat[]{\includegraphics{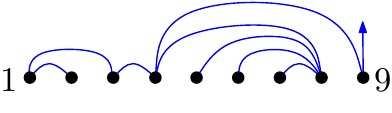}\label{fig:leafshuffle_2}}\hfil
  \caption{A leaf is shuffled into position $2$, with its parent at
    $1$.\label{fig:leafshuffle}}
\end{figure}

Below is a formal statement summarizing the conditions and properties of
the leaf-isolation shuffle.
\begin{proposition}\label{prop:leafshuffle}
  Every rooted tree $T$ on $|T|\ge 2$ vertices admits a one-page book
  embedding onto $[1,|T|]$ such that $q:=\rootof(T)$ is visible from
  above, $2$ is a leaf $\ell$ of $T$, and $1$ is the parent of
  $\ell$. Moreover, if $T$ is a central-star, then $q=1$; otherwise,
  $q=|T|$.
\end{proposition}
\begin{proof}
  We use induction on $n=|T|$. Clearly the statement holds for $n=2$.
  For $n\ge 3$ we start by constructing a one-page book embedding for
  $T$ with a modified version of Algorithm~\ref{alg:embed_t1} where we
  invert the order of subtrees, that is, we use a ``smaller subtree
  first rule'' (SSFR). By starting from $q$ and placing it at $|T|$ we
  ensure that it is visible from above. As $T$ is a tree, the embedding
  uses the edge $\{1,|T|\}$. If $1$ is a leaf of $T$, then $q$ is its
  parent and by SSFR $T$ is a central-star. Therefore, flipping $T$
  yields the desired embedding. Otherwise, let $i\in\{2,\ldots,|T|-1\}$
  denote the smallest (index) neighbor of $1$ and obtain the desired
  embedding inductively for $B[1,i]$ (whose root is $1$). The root of
  this subtree $B[1,i]$ ends up at either $1$ or $i$, both of which are
  visible from above. Therefore, we can complete the embedding by
  routing all edges from $1$ or $i$ to the existing forest on
  $[i+1,|T|]$. \figurename~\ref{fig:leafshuffle} illustrates
  the execution of the leaf-isolation shuffle on an example. The root
  $q$ is at $1$ if and only if $T$ is a central-star; otherwise, it
  remains at $|T|$.
\end{proof}

\subsection{Algorithm outline}\label{proofstart}

Recall that we are given a red tree $R=\tr(r)$, a blue forest $B$ with
roots $b_1,\ldots,b_k$, an interval $I=[i,j]\subseteq[1,n]$ with
$|I|=|R|=|B|$, and a set $C\subset B$ of vertices in edge-conflict with
$r$.

Let $s$ denote a child of $r$ that minimizes $|\tr_R(c)|$ among all
children $c$ of $r$ in $R$. Denote $S=\tr_R(s)$ and $R^-=R\setminus S$.
If $|R^-|\ge 2$, then $R^-$ cannot be a central-star: if it were, then
$|S|=1$ and $R$ would be a star. Another easy consequence of the choice
of $s$ is the following.
\begin{restatable}{lemma}{degr}\label{lem:degr}
  If $\deg_R(r)\ge 2$, then $|R^-|\ge|S|+\deg_{R^-}(r)$.
\end{restatable}
\begin{proof}
  Set $d:=\deg_R(r)=\deg_{R^-}(r)+1$ and suppose to the contrary that
  $|R^-|-|S|\le\deg_{R^-}(r)-1=d-2$. Adding $2|S|$ on both sides of the
  inequality yields $|R|\le d+2|S|-2$. By the minimality of $S$ we have
  $|S|\le(|R|-1)/d$. Solving for $|R|$ and combining with the previous
  inequality yields
  \[
  d|S|+1\le |R|\le d+2|S|-2 \Longrightarrow (d-2)|S|\le d-3,
  \]
  which is impossible, given that $|S|\ge 1$.
\end{proof}

Ideally, we can recursively embed $S$ onto $[j,j-|S|+1]$ and $R^-$ onto
$[i,j-|S|]$ (\figurename~\ref{fig:ideal:1}).  But in general the
invariants may not hold for the recursive subproblems. For instance,
some of the subgraphs could be stars, or if $\{i,j\}\in\EB$, then
placing $r$ at $i$ may put $[j,j-|S|+1]$ in edge-conflict with
$S$. Therefore, we explore a number of alternative strategies, depending
on which---if any---of the four forests $R^-$, $S$, $B[i,j-|S|]$ and
$B[j-|S|+1,j]$ in our decomposition is a star.
\begin{figure}[htbp]
  \centering
  \subfloat[]{\includegraphics{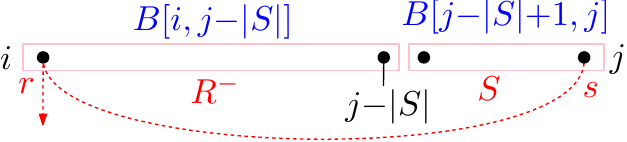}\label{fig:ideal:1}}\hfil
  \subfloat[]{\includegraphics{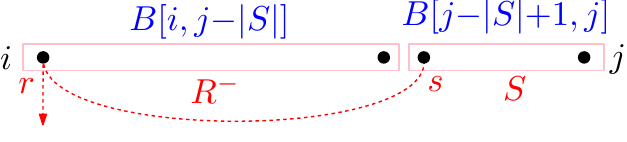}\label{fig:ideal:2}}\hfil
  \caption{Our recursive strategy in an ideal world.\label{fig:ideal}}
\end{figure}

To complete the proof of Theorem~\ref{thm:main} we distinguish seven
cases. In each of these seven cases, we follow the notation of and
assume the preconditions discussed above. First, in
Section~\ref{subsec:rec_general} we discuss the general case, where none
of the four forests is a star. Then, in Section~\ref{subsec:rec_unary}
and Section~\ref{subsec:rec_singleton} we handle the special cases
$\deg_R(r)=1$ and $|S|=1$, respectively. The final four sections each
correspond to one of the four forests being a star. Capturing the
general intuition we refer to $R^-$ as ``large'' and to $S$ as
``small'', although they may have almost the same size and---in special
cases, like $\deg_R(r)=1$---$S$ may actually be larger than $R^-$.

\section{Embedding the red tree: the general case}
\label{subsec:rec_general}
In the general case, we suppose that none of the subtrees in our current
decomposition is a star.
\begin{lemma}\label{lem:rec_general}
  If none of $S$, $R^-$, $B[i,j-|S|]$, and $B[j-|S|+1,j]$ is a star,
  then there is an ordered plane packing of $B$ and $R$ onto $I$.
\end{lemma}
\begin{proof}
  As $S$ is a minimum size subtree of $r$ in $R$, and neither $S$ nor
  $R^-$ is a star, we know that $r$ has at least one more subtree other
  than $S$ and every subtree of $r$ in $R$ has size at least four. (All
  trees on three or less vertices are stars.) It follows that
  \begin{equation}\label{eq:degr}
    \deg_{R^-}(r)\le(|R^-|-1)/4.
  \end{equation}

  The general plan is to use one of the following two options. In both
  cases we first embed $R^-$ recursively onto $[i,j-|S|]$. Then we
  conclude as follows.
  \begin{description}
  \item[Option 1:] Embed $S$ recursively onto $[j,j-|S|+1]$
    (\figurename~\ref{fig:ideal:1}).
  \item[Option 2:] Embed $S$ recursively onto $[j-|S|+1,j]$
    (\figurename~\ref{fig:ideal:2}).
  \end{description}
  In some cases neither of these two options works and so we have to use
  a different embedding.

  As we embed $S$ after $R^-$, the (final) mapping for $s$ is not known
  when embedding $R^-$. However, we need to know the position of $s$ in
  order to determine the conflicts for embedding $R^-$. Therefore,
  before embedding $R^-$ we \emph{provisionally} embed $s$ at
  $\alpha:=\rootof(\treeatt{[j-|S|+1,j]}{j})$ (Option 1) or
  $\alpha:=\rootof(\treeatt{[j-|S|+1,j]}{j-|S|+1})$ (Option 2).  That
  is, for the recursive embedding of $R^-$ we pretend that some neighbor
  of $r$ is embedded at $\alpha$. In this way we ensure that $S$ is not
  in edge-conflict with the interval in its recursive embedding.
  The final placement for $s$ is then determined by the recursive
  embedding of $S$, knowing the definite position of its parent $r$.

  For the recursive embeddings to work, we need to show that the
  invariants \ref{inv:starconflict}, \ref{inv:bluelocal} and
  \ref{inv:rootsonly} hold (\ref{inv:placement} then follows as in
  Observation~\ref{obs:invariants}). For \ref{inv:bluelocal} and
  \ref{inv:rootsonly} this is obvious by construction and
  Observation~\ref{obs:bluelocal}, as long as we do not change the
  embedding of $B$. As we do not change the embedding in Option~1 and 2,
  it remains to ensure \ref{inv:starconflict}. So suppose that for both
  options, \ref{inv:starconflict} does not hold for at least one of the
  two recursive embeddings. There are two possible obstructions for
  \ref{inv:starconflict}: edge-conflicts and degree-conflicts. We
  discuss both types of conflicts, starting with edge-conflicts.

  \case{1} $[i,j-|S|]$ is not in degree-conflict with $R^-$ and
  $[j,j-|S|+1]$ is not in degree-conflict with $S$. Then Option~1 works,
  unless $[i,j-|S|]$ is in edge-conflict with $R^-$. Recall that
  $[j,j-|S|+1]$ is not in edge-conflict with $S$ after embedding $R^-$
  onto $[i,j-|S|]$ due to the provisional placement of $s$.

  We claim that an edge-conflict between $R^-$ and $[i,j-|S|]$ implies
  $\{i,j\}\in\EB$. To prove this claim, suppose that $[i,j-|S|]$ is in
  edge-conflict with $R^-$. Then $\treeatt{[i,j-|S|]}{i}$ is a
  central-star whose root $c$ is in edge-conflict with $r$. If $c=i$,
  then by \ref{inv:placement} there was no such conflict initially (for
  $R$ and $[i,j]$). So, as claimed, the conflict can only come from a
  blue edge to $s$ (provisionally placed) at $j$. Otherwise, $c>i$ and
  by 1SR there is no edge in $B$ from $c$ to any point in $[c+1,j]$. It
  follows that $\treeatt{[i,j-|S|]}{i}=\treeat{i}$. The conflict between
  $c$ and $r$ does not come from the edge to $s$ but from an edge to a
  vertex outside of $[i,j]$. This contradicts \ref{inv:starconflict} for
  $R$ and $[i,j]$, which proves the claim.

  The presence of the edge $\{i,j\}$ implies that $B$ is a tree and by
  \ref{inv:rootsonly} only (the root) $i$ or $j$ may have edges out of
  $[i,j]$. Consider Option~2, which embeds $S$ onto $[j-|S|+1,j]$,
  provisionally placing $s$ at
  $\rootof(\treeatt{[j-|S|+1,j]}{j-|S|+1})$.  There are two possible
  obstructions: an edge-conflict for $R^-$ or a degree-conflict for
  $S$. In both cases we face a central-star $B^*=B[j-|S|+1,b]$ with
  center $b\in[j-|S|+1,j-1]$. Due to 1SR and $\{i,j\}\in\EB$, we know
  that $b=\rootof(\treeatt{[j-|S|+1,j]}{j-|S|+1})$. We distinguish three
  cases.

  \case{1.1} $\{i,b\}\in\EB$. Then we consider a third option:
  provisionally place $s$ at $j$, 
  embed $R^-$ recursively onto $[j-|S|,i]$ and then $S$ 
  onto $[j,j-|S|+1]$ (\figurename~\ref{fig:general2_1}). The edge
  $\{i,b\}$ of $B$ prevents any edge-conflict between $[j-|S|,i]$ and
  $R^-$ (and, as before, for $S$). Given that we assume in Case~1 that
  $[j,j-|S|+1]$ is not in degree-conflict with $S$, we are left with
  $[j-|S|,i]$ being in degree-conflict with $R^-$ as a last possible
  obstruction.
  \begin{figure}[htbp]
    \centering\hfil%
    \subfloat[]{\includegraphics{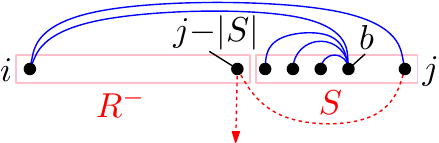}\label{fig:general2_1}}\hfil
    \subfloat[]{\includegraphics{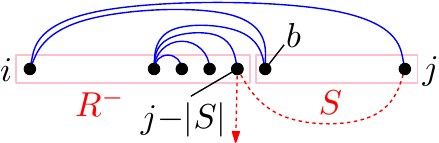}\label{fig:general2_2}}\hfil
    \caption{A third embedding when the first two options
      fail.\label{fig:general2}}
  \end{figure}

  \noindent
  Then the tree $\treeatt{[i,j-|S|]}{j-|S|}$ is a central-star $A^*$
  with root $a$ such that
  \begin{equation}\label{eq:degcon5}
    \deg_{A^*}(a)+\deg_{R^-}(r)\ge|R^-|.
  \end{equation}

  \noindent
  Combining Lemma~\ref{lem:degr} with \eqref{eq:degcon5} we get
  $|A^*|=\deg_{A^*}(a)+1\ge|S|+1\ge 5$. Note that $A^*$ can be huge, but
  we know that it does not include $i$ (because $B[i,j-|S|]$ is not a
  star). 
  We also know that $a\ne j-|S|$:
  If $a= j-|S|$, then by 1SR we have
  $\p_B(a)\in [i,j-|S|-1]$, in contradiction to
  $a=\rootof(\treeatt{[i,j-|S|]}{j-|S|}$. Therefore $a=j-|S|-|A^*|+1$
  and by 1SR its parent is to the right. Due to $\{i,b\}\in\EB$ and
  since $B[j-|S|+1,b]$ is a tree rooted at $b$, we have $\p_B(a)=b$. As
  $A^*$ is a subtree of $b$ in $B$ on at least five vertices, by LSFR
  $b$ cannot have a leaf at $b-1$. Therefore, the star
  $\treeatt{[j-|S|+1,j]}{j-|S|+1}$ consists of a single vertex only,
  that is, $b=j-|S|+1$ (\figurename~\ref{fig:general2_2}). We consider
  two subcases.
  In both 
  the packing is eventually completed by recursively embedding $S$ onto
  $[j,j-|S|+1]$.

  \case{1.1.1} $\{x,b\}\in\EB$, for some $x\in[i+1,a-1]$
  (\figurename~\ref{fig:general3_1}). Select $x$ to be maximal with this
  property. Then we exchange the order of the two subtrees $\tr(x)$ and
  $A^*$ of $b$ (\figurename~\ref{fig:general3_2}). This may violate LSFR
  for $B$ at $b$, but \ref{inv:bluelocal} holds for both $B[i,j-|S|]$
  and $B[j-|S|+1,j]$.
  Clearly there is still no edge-conflict for $[j-|S|,i]$ with $R^-$
  after this change. We claim that there is no degree-conflict anymore,
  either.

  \begin{figure}[htbp]
    \subfloat[]{\includegraphics{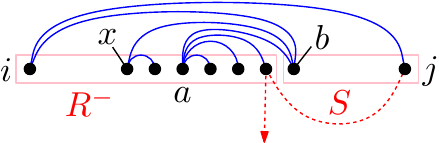}\label{fig:general3_1}}\hfill
    \subfloat[]{\includegraphics{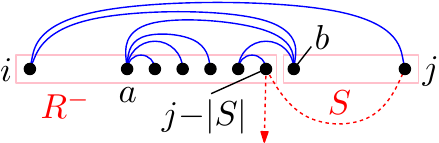}\label{fig:general3_2}}\hfill
    \subfloat[]{\includegraphics{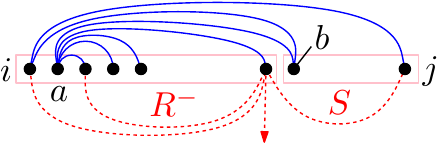}\label{fig:general4}}
    \caption{Swapping two subtrees of $b$ in Case~1.1.1 and an explicit
      embedding for Case~1.1.2.\label{fig:general3}}
  \end{figure}

  To prove the claim, note that by LSFR at $b$ we have
  $|\tr(x)|\le|A^*|$. As the size of both subtrees combined is at most
  $|R^-|$, we have $|\tr(x)|\le|R^-|/2$.  Then, using~\eqref{eq:degr},
  $|\tr(x)|-1+\deg_{R^-}(r)< |R^-|/2+\deg_{R^-}(r)< 3|R^-|/4<|R^-|$.
  Therefore after the exchange $[j-|S|,i]$ is not in degree-conflict
  with $R^-$, which proves the claim and concludes this case.

  \case{1.1.2} $i$ and $a=j-|S|-|A^*|+1$ are the only neighbors of $b$
  in $B$. We claim that in this case $A^*$ extends all the way up to
  $i+1$, that is, $A^*= 
  B[i+1,j-|S|]$.
  To prove this claim, suppose to the contrary that $a\ge i+2$. Then
  there is another subtree of $i$ to the left of $a$ and, in particular,
  $\{i,a-1\}\in\EB$. By LSFR this closer subtree 
  is at least as large as $A^*$. Using \eqref{eq:degr} and
  \eqref{eq:degcon5} we get
  $|[i+1,a-1]|+|A^*|\ge 2|A^*|>2(|R^-|-\deg_{R^-}(r))>3|R^-|/2>|R^-|$,
  in contradiction to $|[i+1,a-1]|+|A^*|<|R^-|$. Therefore $a=i+1$, as
  claimed (\figurename~\ref{fig:general4}).

  The vertex $a$ has high degree in $B$ but it is not adjacent to
  $i$. Therefore, we can embed $R^-$ as follows: put $r$ at $j-|S|$ and
  embed an arbitrary subtree $Y$ of $r$ onto $[i,i+|Y|-1]$ recursively
  or, if it is a star, explicitly, using the locally isolated vertex at
  $i$ for the center (and $i+|Y|-1$ for the root in case of a dangling
  star). As $i$ is 
  isolated on $[i,i+|Y|-1]$ there is no conflict between $[i,i+|Y|-1]$
  and $Y$. As $|Y|\ge|S|\ge 4$, the remaining graph $B[i+|Y|,j-|S|-1]$
  consists of isolated vertices only, on which we can explicitly embed
  any remaining subtrees of $r$ using the algorithm from
  Section~\ref{sec:emb_t1}.

  \case{1.2} $\{i,b\}\notin\EB$ and $b=\p_B(j-|S|)$. Then $j-|S|$ is a
  locally isolated vertex in $B[i,j-|S|]$, whose only neighbor in $B$ is
  at $b\notin\treeatt{[j-|S|+1,j]}{j}$. Therefore, we can provisionally
  place $s$ at $j$ 
  so that $[j-|S|,i]$ is not in conflict with $R^-$. By the assumption
  of Case~1 
  $[j,j-|S|+1]$ is not in degree-conflict with $S$. Therefore, we obtain
  the claimed packing by first embedding $R^-$ onto $[j-|S|,i]$
  recursively and then $S$ onto $[j,j-|S|+1]$.

  \case{1.3} $\{i,b\}\notin\EB$ and $b\ne\p_B(j-|S|)$. As
  $\{i,b\}\notin\EB$ and $s$ is provisionally placed at
  $b$, 
  the interval $[i,j-|S|]$
  is not in edge-conflict with $R^-$.
  Thus, Option~2 (\figurename~\ref{fig:ideal:2}) succeeds unless
  $[j-|S|+1,j]$ is in degree-conflict with $S$. Hence suppose
  \begin{equation}\label{eq:cconf}
    \deg_S(s)+\deg_{B^*}(b)\ge|S|.
  \end{equation}

  By Lemma~\ref{lem:degcon3} we have $|B^*|\ge 3$. %
  As $b\ne\p_B(j-|S|)$, by LSFR $b$ has exactly one neighbor in $B$
  outside of $B^*$: its parent $\p_B(b)\in[i+1,j-|S|]$
  (\figurename~\ref{fig:general_inline2}). Let
  $B^+=B^*\cup\{\p_B(b)\}$. We blue-star embed $S$ starting from $b$
  with $\varphi=(v_1,\ldots,v_d)=(j,\ldots)$ so that $\varphi$ takes the
  vertices of $I\setminus B^+$ from right to left. Let us argue that the
  conditions for the blue-star embedding hold.

  \ref{gg:ec} holds due to $\{i,b\}\notin\EB$ and
  $i=\rootof(\treeatt{[i,j-|S|]}{i})$. For the first inequality of
  \ref{gg:dc} we have to show $|S|\le|B^*|+\deg_S(s)$, which is
  immediate from \eqref{eq:cconf}. For the second inequality of
  \ref{gg:dc} we have to show $|B^+|+\deg_S(s)\le|I|-1$. This follows
  from
  $|B^*|+1+\deg_S(s)\le(|S|-1)+1+(|S|-1)\le|S|+(|R^-|-1)-1=|I|-2$. Regarding
  \ref{gg:int} note that in $\varphi$ we take the vertices of
  $I\setminus B^+$ from right to left. If $\varphi$ reaches beyond
  $\p_B(b)$, then $B\setminus(B^+\cup\varphi)$ forms an interval
  (\figurename~\ref{fig:general_inline2_gg5}); otherwise,
  $B\setminus(B^*\cup\varphi)$ forms an interval
  (\figurename~\ref{fig:general_inline2_gg2}).  Conversely, if
  $B\setminus(B^*\cup\varphi)$ does not form an interval, then $\varphi$
  reaches beyond $\p_B(b)$. In particular, in that case $\varphi$
  includes $\p_B(b)-1$ and we may simply move $\p_B(b)-1$ to the front
  of $\varphi$, establishing the second condition in
  \ref{gg:cs}. Regarding the remaining two conditions it suffices to
  note that $S$ is not a star by assumption and that $\p_B(b)-1$ is not
  a neighbor of $b$ in $B$ because $\p_B(b)$ is the only neighbor of $b$
  outside of $B^*$.

  Therefore, we can blue-star embed $S$ as claimed. By construction and
  Proposition~\ref{p:greedygrab} that leaves us with an interval
  $[i',j']$, where $i=i'$. This ``new'' interval is obtained from the
  interval $[i,j-|S|]$ before the blue-star embedding 
  by replacing some suffix of vertices by a corresponding number of
  locally isolated vertices. In particular, $\treeatt{[i',j']}{i'}$ is a
  subtree of $\treeatt{[i,j-|S|]}{i}$ and
  $i'=\rootof(\treeatt{[i',j']}{i'})$.

  We complete the packing by recursively embedding $R^-$ onto
  $[i',j']$. This interval is not in edge-conflict with $R^-$ by
  \ref{inv:placement}, $\{i,b\}\notin\EB$ and
  $i'=\rootof(\treeatt{[i',j']}{i'})$. We claim that it is not in
  degree-conflict with $R^-$, either. Suppose towards a contradiction
  that $[i',j']$ is in degree-conflict with $R^-$. Then
  $\treeatt{[i',j']}{i'}$ is a central-star and so by LSFR also
  $\treeatt{[i,j-|S|]}{i}$ is a central-star on at least this many
  vertices before the blue-star embedding. This contradicts the
  assumption of Case~1 that $[i,j-|S|]$ is not in degree-conflict with
  $R^-$. Therefore, $[i',j']$ is not in degree-conflict with $R^-$ and
  we can complete the packing as described. This completes the proof for
  Case~1.

  \begin{figure}[thbp]
    \centering%
    \subfloat[]{\includegraphics{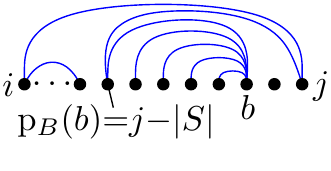}\label{fig:general_inline2_gg4}}\hfil%
    \subfloat[$\p_B(b)>v_d$]{\includegraphics{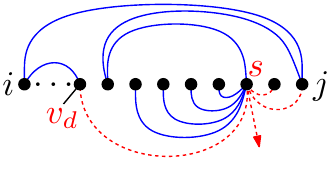}\label{fig:general_inline2_gg5}}\hfil%
    \subfloat[]{\includegraphics{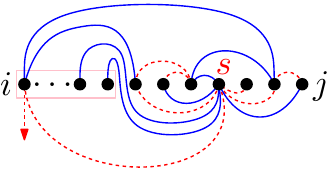}\label{fig:general_inline2_gg6}}\\
    \subfloat[]{\includegraphics{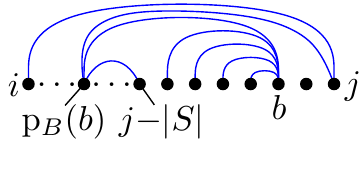}\label{fig:general_inline2_gg1}}\hfil%
    \subfloat[$\p_B(b)<v_d$]{\includegraphics{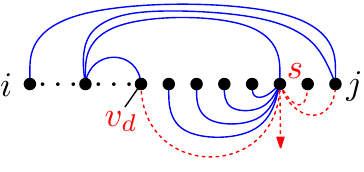}\label{fig:general_inline2_gg2}}\hfil%
    \subfloat[]{\includegraphics{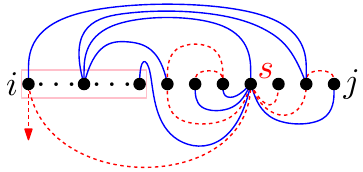}\label{fig:general_inline2_gg3}}
    \caption{Explicit embedding of $S$ in Case~1.3. The edge
      $\{\p_B(b),j\}$ need not be present in
      $B$.\label{fig:general_inline2}}
  \end{figure}

  \case{2} $[i,j-|S|]$ is in degree-conflict with $R^-$. Then
  $\treeatt{[i,j-|S|]}{i}$ is a central-star $B[i,x]$
  \begin{equation}
    \label{eq:degconf}
    \mbox{with}\;\;\deg_{R^-}(r)+(x-i)\ge|R^-|
  \end{equation}
  and $|B[i,x]|=x-i+1\ge 3$ by Lemma~\ref{lem:degcon3}. We distinguish
  two cases.

  \case{2.1} $\treeat{i}=B[i,x]$. Then $\treeat{i}\neq\treeat{j}$. If
  necessary, flip $\treeat{i}$ to put its center at $i$. If $\treeat{j}$
  is a central-star on $\ge 3$ vertices, then---if necessary---flip
  $\treeat{j}$ to put its root at $j$. We use a blue-star embedding
  for $R^-$ starting from $\sigma=i$ with $\varphi=(x+1,\ldots)$. As
  $\varphi$ consists of $d:=\deg_{R^-}(r)$ vertices, we have
  $[i,j]\setminus(B[i,x]\cup\varphi)=[x+d+1,j]$. If
  $\treeatt{[x+d+1,j]}{j}$ is a central-star on $\ge 3$ vertices, then
  use $\varphi=(j,x+1,\ldots)$ instead (and note that
  $\rootof(\treeatt{[x+d+1,j]}{j})=j$).

  In the notation of the blue-star embedding we have
  $B^*=B^+=B[i,x]$. We need to show that the conditions for this
  embedding hold. \ref{gg:ec} holds by \ref{inv:starconflict} (for
  embedding $R$ onto $[i,j]$). For \ref{gg:dc} we have to show
  $|R^-|\le|B^*|+\deg_{R^-}(r)\le|R|-1$. The first inequality holds by
  \eqref{eq:degconf} and $|B^*|\ge x-i$. The second inequality holds due
  to \ref{inv:starconflict} (for embedding $R$ onto $[i,j]$), which
  implies $\deg_{R}(r)+(x-i)\le|R|-1$. As $|B[i,x]|=x-i+1$ and
  $\deg_{R}(r)=\deg_{R^-}(r)+1$, \ref{gg:dc} follows. \ref{gg:int} is
  obvious by the choice of $\varphi$ and \ref{gg:cs} is trivial for
  $B^*=B^+$ due to \ref{gg:int}. 
  That leaves us with an interval $[i',j']$, where $j'\in\{j,j-1\}$.  We
  claim that $[j',i']$ is not in conflict with $S$.

  \begin{figure}[htbp]
    \centering\hfil%
    \subfloat[$R$]{\includegraphics{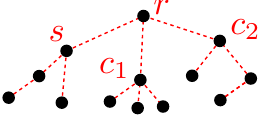}\label{fig:greedygrab_m2}}\hfil
    \subfloat[$B$]{\includegraphics{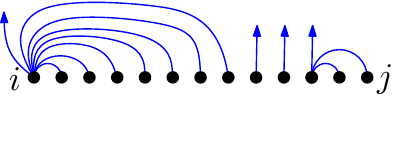}\label{fig:greedygrab_m0}}\hfil
    \subfloat[blue-star]{\includegraphics{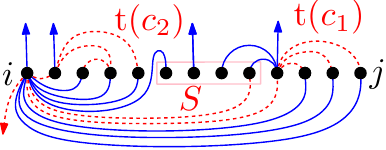}\label{fig:greedygrab_m1}}\hfil
    \caption{Handling a degree-conflict for $R^-$ in
      Case~2.1.\label{fig:modgreedygrab}}
  \end{figure}

  To prove the claim we consider two cases. If $j'=j-1$, then initally
  $\treeatt{[x+d+1,j]}{j}$ was a central-star on $\ge 3$ vertices rooted
  at $j$. By the choice of $\varphi$ a leaf of this star is at $j'$
  whose only neighbor in $B$ is at $j\ne i=r$. Therefore $[j',i']$ is
  not in edge-conflict with $S$.  As $\treeatt{[j',i']}{j'}$ is an
  isolated vertex, by Lemma~\ref{lem:degcon3} there is no
  degree-conflict between $[j',i']$ and $S$, either, which proves the
  claim.

  Otherwise, $j'=j$ and $\treeatt{[j',i']}{j'}=\treeatt{[x+d+1,j]}{j}$
  is not a central-star on $\ge 3$ vertices. Therefore by
  Lemma~\ref{lem:degcon3} there is no degree-conflict between $[j',i']$
  and $S$. In order to show that there is no edge-conflict, either, it
  is enough to show that $\rootof(\treeatt{[j',i']}{j'})$ is not
  adjacent to $i=r$ in $B$. If $\rootof(\treeatt{[j',i']}{j'})\ne j'$
  this follows from 1SR. Otherwise
  $\rootof(\treeatt{[j',i']}{j'})=j'=j$, and $\{i,j\}\notin\EB$ because
  $\treeat{i}$ is a star but $B$ is not. Therefore the claim holds and
  we can complete the packing by recursively embedding $S$ onto
  $[j',i']$.

  \case{2.2} $\treeat{i}\ne B[i,x]$. By 1SR this means that $i=\p_B(x)$
  and $i$ has at least one more neighbor in $[i,j]\setminus
  B[i,x]$. Since by assumption $B[i,j-|S|]$ is not a star, we have $x\le
  j-|S|-1$. Since $B[i,x]$ is a central-star and $x\le j-|S|-1$, by LSFR
  for $i$ the only neighbor of $i$ in $B$ outside of $B[i,x]$ is its
  parent $\p_B(i)\in[j-|S|+1,j]$.
  We claim that such a configuration is impossible. To prove the claim,
  note that $\p_B(i)$ has at least two children in $B[i,j-|S|]$ because
  $x\le j-|S|-1$ and $\p_B(i)\ge j-|S|+1$. By LSFR, the corresponding
  subtrees have size at least $|B[i,x]|=x-i+1$, and so
  $|R|\geq |B[i,\p_B(i)]|\ge 2|B[i,x]|+1\ge
  2(|R^-|-\deg_{R^-}(r)+1)+1$,
  %
  where the last inequality uses \eqref{eq:degconf}. Rewriting and using
  \eqref{eq:degr} yields
  \[|R^-|\leq
  \frac{|R|-1}{2}+\deg_{R^-}(r)-1<\frac{|R|-1}{2}+\frac{|R^-|}{4}.\]
  It follows that $|R^-|<\frac23(|R|-1)$ and hence that
  $|S|>\frac13(|R|-1)$. Since $S$ is a smallest subtree of $r$ in $R$,
  this means that $r$ is binary in $R$ and thus unary in $R^-$. This,
  finally, contradicts the degree-conflict for $[i,j-|S|]$ with $R^-$
  because $x<j-|S|$ and hence $\deg_{R^-}(r)+\deg_{B[i,x]}(i)=1+(x-i)<
  1+(j-|S|)-i=|R^-|$.

  \case{3} $[j,j-|S|+1]$ is in degree-conflict with $S$ and $[i,j-|S|]$
  is not in degree-conflict with $R^-$. Then $\treeatt{[j-|S|+1,j]}{j}$
  is a central-star $Z=\tr_B(z)$ with $|Z|\ge 3$ by
  Lemma~\ref{lem:degcon3} and
  \begin{equation}\label{eq:degcz}
    \deg_S(s)+\deg_Z(z)\ge|S|.
  \end{equation}

  \case{3.1} $\{i,j\}\notin\EB$. Then we claim that we may
  assume $z=j$ and $Z=\treeat{j}$.

  Let us prove this claim. If $z=j-|Z|+1$, then by 1SR it does not have
  any neighbor in $B\setminus Z$. Flipping $Z=\treeat{j}$ establishes
  the claim. Otherwise, $z=j$. Suppose that $z$ has a neighbor
  $y\in B\setminus Z$. As $z$ is the root of
  $Z=\treeatt{[j-|S|+1,j]}{j}$, it does not have a neighbor in
  $[j-|S|+1,j-|Z|]$ and therefore $y\in[i+1,j-|S|]$. By LSFR and because
  $B[j-|S|+1,j]$ is not a star, $y=\p_B(z)$. In particular, since
  $|Z|\ge 3$, LSFR for $y$ implies $\{y,y+1\}\notin\EB$. It follows that
  after flipping $\treeat{j}$ the resulting subtree
  $\treeatt{[j-|S|+1,j]}{j}$ is not a central-star anymore and so there
  is no conflict for embedding $S$ onto $[j,j-|S|+1]$ anymore. Therefore
  we can proceed as above in Case~1 (the conflict situation for $R^-$
  did not change because $\treeat{i}$ remains unchanged). Hence we may
  suppose that there is no such neighbor $y$ of $z$, which establishes
  the claim.

  We blue-star embed $S$ starting from $\sigma=j=z$ with
  $\varphi=(j-|Z|,j-|Z|-1,\ldots)$. In the terminology of the
  blue-star embedding we have $B^*=B^+=Z$. Let us argue that the
  conditions for the embedding hold. \ref{gg:ec} is trivial because no
  neighbor of $s$ is embedded yet. For \ref{gg:dc} we have to show
  $|S|\le|Z|+\deg_S(s)\le|R|-1$. The first inequality holds by
  \eqref{eq:degcz} and the second by
  $|S|\le(|R|-1)/\deg_R(r)\le(|R|-1)/2$, which implies
  $|Z|+\deg_S(s)\le 2(|S|-1)\le|R|-3$. \ref{gg:int} is obvious by the
  choice of $\varphi$ and given $B^*=B^+$, \ref{gg:cs} is trivial.
  That leaves us with an interval $[i',j']$, where $i'=i$.

  The plan is to recursively embed $R^-$ onto $[i,j']$. This works fine,
  unless $[i,j']$ and $R^-$ are in conflict. So suppose that they are in
  conflict. Then there is a central-star
  $Y=\treeatt{[i,j']}{i}$. Considering how $\varphi$ consumes the
  vertices in $I$ from right to left, $Y$ appears as a part of some
  component of $B$, that is, $Y=B[i,y]$, for some $y\in[i,j-|S|]$.

  We claim that $\rootof(Y)=i$. To prove the claim, suppose to the
  contrary that $\rootof{Y}=y=\p_B(i)$. Then by 1SR $Y$ \emph{is} a
  component of $B$. Thus, a degree-conflict contradicts the assumption
  of Case~3 that $[i,j-|S|]$ is not in degree-conflict with $R^-$, and
  an edge-conflict contradicts \ref{inv:starconflict} for embedding $R$
  onto $[i,j]$ together with the fact that by 1SR $y$ is not adjacent to
  any vertex outside of $Y$ in $B$---in particular not to $j$, where $s$
  was placed. This proves the claim and, furthermore, that
  $\p_B(i)\in[y+1,j-|S|]$ and $\p_B(i)$ appears in $\varphi$.

  By \ref{inv:placement} for embedding $R$ onto $[i,j]$ and
  $\{i,j\}\notin\EB$ we know that $[i,j']$ and $R^-$ are not in
  edge-conflict and so they are in degree-conflict.  In particular,
  $\deg_Y(i)+\deg_{R^-}(r)\ge|R^-|$.

  Undo the blue-star embedding. We claim 
  $\{i,y+1\}\in\EB$. To prove the claim, suppose to the contrary that
  $\{i,y+1\}\notin\EB$. Then $\{y+1,\p_B(i)\}\in\EB$ because in $B$ the
  vertex $y+1$ lies below the edge $\{i,\p_B(i)\}$. By LSFR the subtree
  of $\p_B(i)$ rooted at $y+1$ is at least as large as $Y$. Therefore,
  \[
  |\tr_B(p_B(i))|\ge 2|Y|+1=2\deg_Y(i)+3\ge 2(|R^-|-\deg_{R^-}(r))\ge
  \frac{3}{2}|R^-|,
  \]
  where the last inequality uses \eqref{eq:degr}. This is in
  contradiction to $p_B(i)\le j-|S|$, which implies
  $|\tr_B(p_B(i))|\le|R^-|$. Therefore, the claim holds and
  $\{i,y+1\}\in\EB$.

  Flip $B[i,y+1]$ and perform the blue-star embedding again. Although
  1SR may be violated at $y+1$, this is of no consequence for the
  blue-star embedding. As $Y=\treeatt{[i,j']}{i}=B[i,y]$, we know that
  $y+1$ appears in $\varphi$ and so the offending vertex is not part of
  $[i,j']$ after the blue-star embedding. Furthermore, in this way we
  also get rid of the high-degree vertex of $B$ that was at $i$
  initially so that the vertices in $[i,y]\subset[i,j']$ are isolated.
  In particular, $i$ is isolated in $[i,j']$ and its only neighbor in
  $B$ is at $y+1\ne j$. Therefore, $[i,j']$ and $R^-$ are not in
  conflict, unless $y+1=\p_B(i)$ initially and $\p_B(i)$ is in
  edge-conflict with $r$.

  In other words, it remains to consider the case 
  $\treeat{i}=B[i,y+1]=\tr_B(y+1)$ is a dangling star whose root $i$ (at
  $y+1$ before flipping) is in edge-conflict with $r$
  (\figurename~\ref{fig:31:1}).
  Then $[i,j']$ is an independent set in $B$ that consists of leaves of
  the two stars $Y$ and $Z$ plus the isolated vertex at $i$. Yet we
  cannot simply embed $R^-$ using the algorithm from
  Section~\ref{sec:emb_t1} because $i$ is and $j'$ may be in
  edge-conflict with $r$. Given that $|Y|\ge 3$ and $\varphi$ gets to
  $y+1$ only, at least two leaves of $Y$ remain in $[i,j']$ and so, in
  particular, $i+1$ is not in conflict with $r$. We explicitly embed
  $R^-$ as follows (\figurename~\ref{fig:31:2}): place $r$ at $i+1$ and
  a child $c$ of $r$ in $R^-$ at $i$. Then collect $|\tr_R(c)|$ leaves
  from $Z$ and/or $Y$ and put them right in between $i$ and $i+1$.
  First---from left to right---the leaves of $Z$ whose blue edges leave
  them upwards to bend down and cross the spine immediately to the right
  of the vertices of 
  the red subtree rooted at $y+1$ (the leftmost subtree of $S$) and then
  reach $z$ from below. Next come the leaves of $Y$ whose blue edges to
  $y+1$ are drawn as arcs in the upper halfplane. In order to make room
  for those leaves, the blue edge $\{i,y+1\}$ is re-routed to leave $i$
  downwards to bend up and cross the spine immediately to the left of
  $i+1$ in order to reach $y+1$ from above. Using the algorithm from
  Section~\ref{sec:emb_t1} we can now embed $\tr_R(c)$ onto these leaves
  and any remaining subtrees of $r$ can be embedded explicitly
  on the vertices $i+2,\ldots$ (ignoring the change of numbering caused
  by the just discussed repositioning of leaves).
  \begin{figure}[htbp]
    \centering%
    \subfloat[]{\includegraphics{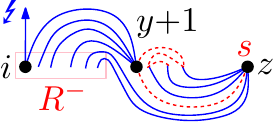}\label{fig:31:1}}\hfil
    \subfloat[]{\includegraphics{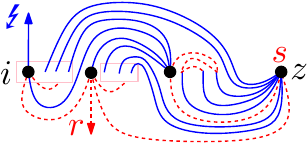}\label{fig:31:2}}\hfil
    \subfloat[]{\includegraphics{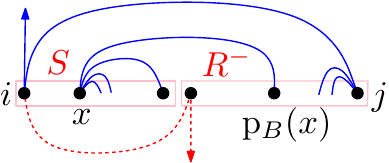}\label{fig:gen_322}}\hfil
    \caption{(a)--(b): Relocating some leaves of the stars
      $Y=\tr_B(y+1)$ and $Z=\tr_B(z)$ in Case~3.1. One subtree of
      $\tr_R(c)$ of $R^-$ is embedded at $i$ and the leaves to the right
      of $i$; all other subtrees of $R^-$ are embedded to the right of
      $r$. Both from $y+1$ and from $z$ we can route as many blue edges
      as desired to either of these ``pockets''. (c): Evading a
      degree-conflict for $S$ in Case~3.2.1.\label{fig:31}}
  \end{figure}

  \case{3.2} $\{i,j\}\in\EB$. Then $z=j$ because $j-|Z|+1$ is enclosed
  by $\{i,j\}$ and therefore cannot be the root of $Z$. Moreover,
  $\rootof(B)=i$ by LSFR and since $B$ is not a star. By LSFR $j$ does
  not have any child in $B\setminus Z$ and as $B[j-|S|+1,j]$ is not a
  star, $Z\subseteq B[j-|S|+2,j]$. In particular, $j$ is not adjacent to
  any vertex in $B[i+1,j-|S|+1]$. We provisionally place $s$ at any
  vertex in $[i+1,j-|R^-|]$, say, at $j-|R^-|$. Then $[j,j-|R^-|+1]$ is
  not in edge-conflict with $R^-$. We claim that it is not in
  degree-conflict, either. As $Z$ is a star on $|Z|\le|S|-1$ vertices,
  by Lemma~\ref{lem:degr} we have
  $\deg_{R^-}(r)+\deg_Z(z)\le\deg_{R^-}(r)+|S|-2\le|R^-|-2$ and the
  claim follows. We recursively embed $R^-$ onto $[j,j-|R^-|+1]$,
  treating all local roots of $B$ other than $j$ as in conflict with
  $r$. It remains to recursively embed $S$ onto $[j-|R^-|,i]$.

  Suppose towards a contradiction that $[j-|R^-|,i]$ is in conflict with
  $S$. Then there is a central-star $X=B[x,j-|R^-|]=\treeatt{[i,j-|R^-|]}{j-|R^-|}$.
  Due to $\{i,j\}\in\EB$ and 1SR we have $\rootof(X)=x$ and $\p_B(x)>j-|R^-|$.
  Together with LSFR for $j$ it follows that
  $\p_B(x)\in[j-|R^-|+1,j-|Z|]$. Due to the conflict setting for
  embedding $R^-$, $i$ is the only vertex in $[j-|R^-|,i]$ that may be
  in edge-conflict with $s$. As $X$ is a central-star and
  $\rootof(B)=i$, we cannot have $x=i$ because then $B$ would be a star.
  It follows that $x>i$ and so $[j-|R^-|,i]$ is not in edge-conflict
  with $S$. Therefore $[j-|R^-|,i]$ and $S$ are in degree-conflict.
  Then $|X|\ge 3$ by Lemma~\ref{lem:degcon3} and
  \begin{equation}\label{eq:degcon4}
    \deg_S(s)+\deg_X(x)\ge|S|.
  \end{equation}
  Depending on $\p_B(x)$ we consider two final subcases.

  \case{3.2.1} $\p_B(x)\in[j-|R^-|+2,j-|Z|]$
  (\figurename~\ref{fig:gen_322}). Then the edge $\{x,\p_B(x)\}$
  encloses $j-|R^-|+1$ so that, in particular,
  $\{i,j-|R^-|+1\}\notin\EB$. We provisionally place $s$ at
  $i=\rootof(\treeatt{[i,j-|R^-|]}{i})$ and claim that $[j-|R^-|+1,j]$
  and $R^-$ are not in conflict.

  To prove the claim, consider $W^*:=\treeatt{[j-|R^-|+1,j]}{j-|R^-|+1}$
  and suppose it is a central-star. (If it is not, then we are done.) If
  $\rootof(W^*)>j-|R^-|+1$, then by 1SR and $\{x,\p_B(x)\}\in\EB$ we
  have $\p_B(\rootof(W^*))=x$, in contradiction to LSFR for $x$.
  Therefore $\rootof(W^*)=j-|R^-|+1$. In order for $j-|R^-|+1$ to be the
  local root for $W^*$ in the presence of $\{x,\p_B(x)\}\in\EB$, it
  follows that $\p_B(j-|R^-|+1)=x$ and so by 1SR $|W^*|=1$. Therefore by
  Lemma~\ref{lem:degcon3} there is no degree-conflict between
  $[j-|R^-|+1,j]$ and $R^-$. As $\{x,\p_B(x)\}\in\EB$ prevents any
  connection in $B$ from $j-|R^-|+1$ to $i$ and to vertices outside of
  $[i,j]$, there is no edge-conflict between $[j-|R^-|+1,j]$ and $R^-$,
  either. This proves the claim. Recursively embed $R^-$ onto
  $[j-|R^-|+1,j]$. Recall that $\rootof(B)=i$. There is no conflict for
  embedding $S$ onto $[i,j-|R^-|]$ since $\{i,r\}\not\in\EB$ and
  $\treeatt{[i,j-|R^-|]}{i}$ is not a central-star of size at least 2 by
  LSFR at $i$. Finish the packing by recursively embedding $S$ onto
  $[i,j-|R^-|]$.

  \case{3.2.2} $\p_B(x)=j-|R^-|+1$. Then by 1SR $\p_B(x)$ is the only
  neighbor of $x$ outside of $X$ in $B$. We provisionally place $r$ at
  $j$ and employ a blue-star embedding for $S$, starting from
  $\sigma=x$ with $\varphi=(i,\ldots)$, that is, $\varphi$ takes
  vertices from left to right, skipping over $[x,\p_B(x)]$. Let us argue
  that the conditions for the blue-star embedding hold.

  In the terminology of the blue-star embedding we have $B^*=X$ and
  $B^+=X\cup\{\p_B(x)\}$. \ref{gg:ec} holds because
  $\{x,j\}\notin\EB$. For the first inequality of \ref{gg:dc} we have to
  show $|S|\le|X|+\deg_S(s)$, which is immediate from
  \eqref{eq:degcon4}. For the second inequality of \ref{gg:dc} we have
  to show $|X|+1+\deg_S(s)\le|I|-1$. This follows from
  $|X|+1+\deg_S(s)\le|S|+(|S|-1)\le|R^-|+|S|-1=|I|-1$. Regarding
  \ref{gg:int} note that in $\varphi$ we take the vertices of
  $B\setminus B^+$ from left to right. As there are not enough vertices
  in $[i,x-1]$ to embed the neighbors of $s$ (which causes the
  degree-conflict), $\varphi$ reaches beyond $\p_B(x)$ and so
  $B\setminus(B^+\cup\varphi)$ forms an interval. In particular,
  $\varphi$ includes $\p_B(x)+1$ and we may simply move $\p_B(x)+1$ to
  the front of $\varphi$, establishing the second condition in
  \ref{gg:cs}. Regarding the remaining two conditions in
  \ref{gg:cs} 
  note that $S$ is not a star by assumption and that $\p_B(x)+1$ is not
  a neighbor of $x$ in $B$ because $\p_B(x)$ is the only neighbor of $x$
  outside of $X$.

  Therefore, we can blue-star embed $S$ as claimed, which
  leaves us with an interval $[i',j']$, where $j=j'$. As
  $\{x,j\}\notin\EB$ and $j$ is not the local root of $B$ ($i$ is),
  there is no edge-conflict between $[j',i']$ and $R^-$. As there is no
  degree-conflict between $[j,j-|R^-|+1]$ and $R^-$ and the number of
  neighbors of $j$ in $B[i',j']$ can only decrease compared to
  $B[j-|R^-|+1,j]$ (if they appear in $\varphi$), there is no
  degree-conflict between $[j',i']$ and $R^-$, either. Therefore, we can
  complete the packing by embedding $R^-$ onto $[j',i']$ recursively.
\end{proof}

\section{Embedding the red tree: a unary root}\label{subsec:rec_unary}

In this section we handle all cases where the root $r$ of $R$ is unary.

\begin{proposition}\label{prop:rec_unary_star}
  If $\deg_R(r)=1$ and $S$ is a star, then there is an ordered plane
  packing of $B$ and $R$ onto $I$.
\end{proposition}
\begin{proof}
  Since $\deg_R(r)=1$ and $R$ is not a star by assumption, $S$ must be a
  dangling star. Thus, we know exactly what $R$ looks like: it is rooted
  at $r$, which has a single child $s$, which has a single child $q$,
  which finally has zero or more leaf children.

  \case{1} $\{i,j\}\not\in\EB$. We consider three cases.

  \begin{figure}[b]
    \centering
    \subfloat[Case~1.1]{\includegraphics{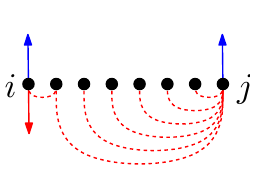}\label{fig:unary_star_ij_isolated}}\hfil%
    \subfloat[Case~1.2]{\includegraphics{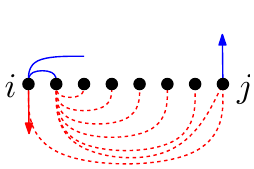}\label{fig:unary_star_i_not_isolated}}\hfil%
    \subfloat[Case~1.3]{\includegraphics{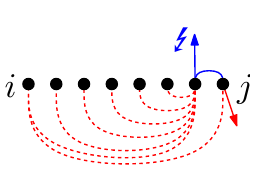}\label{fig:unary_star_j_not_isolated_1}}\\
    \subfloat[Case~1.3]{\includegraphics{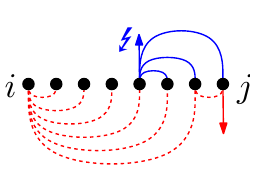}\label{fig:unary_star_j_not_isolated_2}}\hfil%
    \subfloat[Case~2.1]{\includegraphics{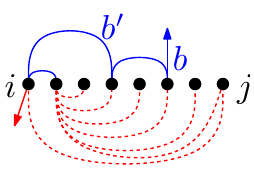}\label{fig:unary_star_ij_used_1}}\hfil%
    \subfloat[Case~2.2]{\includegraphics{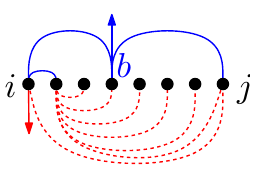}\label{fig:unary_star_ij_used_2}}
    \caption{The case analysis in the proof of
      Proposition~\ref{prop:rec_unary_star}.}
    \label{fig:unary_star}
  \end{figure}

  \case{1.1} $i$ and $j$ are both isolated in $B$. Embed $r$ to $i$, $s$
  to $i+1$, $q$ to $j$, and the children of $q$ onto $[j-1,i+2]$. See
  \figurename~\ref{fig:unary_star_ij_isolated}. Note that $i$ is not in
  edge-conflict due to the placement invariant. Every red edge is
  incident to $i$ or $j$ and hence does not occur in $B$ by assumption.

  \case{1.2} $i$ is not isolated in $B$. If $\treeat{i}$ is a
  central-star, flip it if necessary to put its root (which is not in
  edge-conflict by the peace invariant) at $i$. Otherwise, use
  the leaf-isolation shuffle to put a leaf at $i+1$ and its parent at
  $i$. Since $\treeat{i}$ is not a central-star, by
  Proposition~\ref{prop:leafshuffle}, this will place the root of
  $\treeat{i}$ at some position $x>i+1$. In both cases, embed $r$ onto
  $i$, $s$ onto $j$, $q$ onto $i+1$, and the children of $q$ onto
  $[i+2,j-1]$. See \figurename~\ref{fig:unary_star_i_not_isolated}. The
  edge $\{r,s\}$ is not used by $B$ since $\{i,j\}$ is not used by
  assumption (and the leaf-isolation shuffle cannot change that). The
  red edges incident to $q$ are not used since the only neighbor of
  $i+1$ in $B$ is $i$.

  \case{1.3} $i$ is isolated and $j$ is not isolated in $B$. Flip
  $\treeat{j}$ if its root is currently at $j$. Note that $[j,i]$ is not
  in degree-conflict for embedding $R$: this would imply that
  $B$ is a star since $\deg_R(r)=1$. If $[j,i]$ is not in
  conflict, then the invariants hold for $[j,i]$ and we can apply
  Case~1.2 by embedding $R$ on $[j,i]$ instead of $[i,j]$. Otherwise,
  $\treeat{j}$ is a central-star on at least two vertices.

  If $|\treeat{j}|=2$, then embed $r$ onto $j$, $s$ onto $i$, $q$ onto
  $j-1$, and the children of $q$ onto $[j-2,i+1]$. See
  \figurename~\ref{fig:unary_star_j_not_isolated_1}. If
  $|\treeat{j}|\geq3$, then embed $r$ onto $j$, $s$ onto $j-1$, $q$ onto
  $i$, and the children of $q$ onto $[i+1,j-2]$. See
  \figurename~\ref{fig:unary_star_j_not_isolated_2}. This works because
  the root of $\treeat{j}$ is not at $j$ (so $j$ is not in
  edge-conflict), the size of the star $\treeat{j}$ is at least three
  (so $\{j-1,j\}$ is not used), and $i$ is isolated in $B$ (so the
  red edges incident to $q$ are not used by $B$).

  \case{2} $\{i,j\}\in\EB$. Let $b$ be the root of $B$.
  We claim that (1) some vertex of $B$ has distance at least three
  to $b$ or (2) $\deg_{B}(b)\geq 2$. To prove the claim, suppose
  that all vertices in $B$ have distance at most two to $b$ and
  that $b$ is unary. Then the child of $b$ has distance one to all other
  vertices of $B$: hence $B$ is a star centered at the child
  of $b$, a contradiction. We perform a case analysis on whether (1) or
  (2) holds.

  \case{2.1} Some vertex $v$ of $B$ has distance at least three to
  $b$. Let $b'$ be the child of $b$ that contains $v$ in its subtree
  $B'$. Let $w$ be the size of $B'$. We re-embed $B$ as follows.
  $B'$ is not a central-star by choice of $v$. Hence, by
  Proposition~\ref{prop:leafshuffle}, we can use the leaf-isolation
  shuffle to embed $B'$ on $[i,i+w-1]$, placing a leaf at $i+1$, its
  parent at $i$, and the root $b'$ at some position in $[i+2,i+w-1]$.
  Complete this embedding of $B'$ to any one-page book embedding of
  $B$. Note that this embedding does not use the edge
  $\{i,j\}$. Embed $r$ at $i$, $s$ at $j$, $q$ at $i+1$, and the
  children of $q$ at $[i+2,j-1]$. See
  \figurename~\ref{fig:unary_star_ij_used_1}. This works because $b$ is
  not at $i$ (so $i$ is not in edge-conflict), $B$ does not use the
  edge $\{i,j\}$ (so $\{r,s\}$ is not used by $B$), and $i+1$ is
  isolated in $B[i+1,j]$ (so the red edges incident to $q$ are not used
  by $B$).

  \case{2.2} $\deg_{B}(b)\geq 2$. Since $B$ is not a star,
  some vertex $v$ has distance at least two to $b$ in $B$. Let $b'$
  be the child of $b$ that contains $v$ in its subtree $B'$. Let $w$ be
  the size of $B'$. We re-embed $B$ as follows. Use the
  leaf-isolation shuffle to embed $B'$ together with $b$ on $[i,i+w]$,
  placing a leaf at $i+1$, its parent at $i$, and $b$ at $i+w$. Complete
  this embedding to any one-page book embedding of $B$. Note that
  this embedding does not use the edge $\{i,j\}$. Finish by using the
  same embedding for $R$ as in Case~2.1. See
  \figurename~\ref{fig:unary_star_ij_used_2}.
\end{proof}

\begin{proposition}\label{prop:rec_unary_regular_ij_used}
  If $\deg_R(r)=1$, $S$ is not a star, and $\{i,j\}\in\EB$,
  then there is an ordered plane packing of $B$ and $R$ onto $I$.
\end{proposition}
\begin{proof}
  Flip $B$ if necessary to put its root at $j$. The general plan is
  to embed $r$ onto $i$ and $S$ recursively onto $[i+1,j]$. This works
  unless (1) $B[i+1,j]$ is a star, (2) $\{i,i+1\}\in\EB$, or (3) there
  is a conflict for embedding $S$ onto $[i+1,j]$. Below, we find an ordered
  plane packing under a weaker condition than (1) to allow for reuse in
  cases (2) and (3).
  In case (2), by LSFR, $\treeatt{[i,j-1]}{i}$ is a central-star on at
  least two vertices. In case (3), $\treeatt{[i+1,j]}{i+1}$ is a
  central-star. We deal with these cases below.

  \case{1} $B[i+1,j]$ is a star or $B[i,j-1]$ is a star. If $B[i,j-1]$
  is a star, then we flip $B$ to reduce to the case that
  $B[i+1,j]$ is a star. Thus, in the following, assume that $B[i+1,j]$
  is a star and that the root of $B$ may be either at $i$ or at
  $j$. We know exactly what $B$ looks like: since $B$ is not a
  star, the star $B[i+1,j]$ must be centered at $i+1$ and rooted at $j$.
  Flip the blue embedding at $[i+1,j]$: this puts the star-center at
  $j$. Note that $\{i,j\}\not\in\EB$. Embed $r$ onto $j$. The
  interval $[i,j-1]$ is in edge-conflict with $S$ if the root of
  $B$ is now at $i+1$. Hence, we embed $S$ explicitly. Embed $s$
  onto $i$. Since $S$ is not a star, it must have a subtree of size
  $k\geq2$. Embed this subtree explicitly at $[i+k,i+1]$. Embed the
  other subtrees of $s$ explicitly on the remainder. See
  \figurename~\ref{fig:unary_regular_ij_used_1}.

  \begin{figure}[t]
    \centering
    \subfloat[Case~1]{\includegraphics{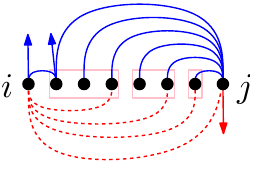}\label{fig:unary_regular_ij_used_1}}\hfil
    \subfloat[Case~2]{\includegraphics{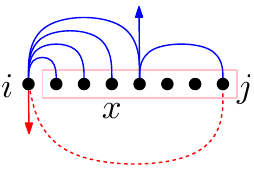}\label{fig:unary_regular_ij_used_2}}\hfil
    \subfloat[Case~3]{\includegraphics{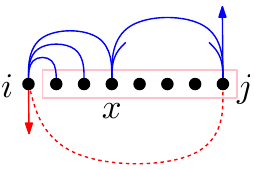}\label{fig:unary_regular_ij_used_3}}\hfil
    \caption{The case analysis in the proof of
      Proposition~\ref{prop:rec_unary_regular_ij_used}.}
    \label{fig:unary_regular}
  \end{figure}

  \case{2} $\treeatt{[i,j-1]}{i}$ is a central-star on at least two
  vertices. Let $x$ be such that $\treeatt{[i,j-1]}{i}=B[i,x]$. By Case~1
  we may assume that $x\leq j-2$. By LSFR at $i$ and by choice of $x$,
  $B[x+1,j]$ is a tree. Flip $B[x+1,j]$. Since $x\leq j-2$, the root of
  $B$ is no longer at $j$. Embed $r$ onto $i$ and $S$ recursively onto
  $[j,i+1]$. See \figurename~\ref{fig:unary_regular_ij_used_2}. Since
  $\{i,j\}\not\in\EB$ after flipping and $i+1$ is isolated in
  $B[i+1,j]$, this works unless $[j,i+1]$ is in conflict for $S$. Then
  $\treeatt{[i+1,j]}{j}$ is a central-star that is rooted at the root of
  $B$. But this contradicts LSFR at $j$ before flipping: a
  contradiction. Hence, there is no conflict for $S$.

  \case{3} $\treeatt{[i+1,j]}{i+1}$ is a central-star. Let $x$ be such
  that $B[i+1,x]=\treeatt{[i+1,j]}{i+1}$. Since $\{i,j\}\in\EB$,
  $B[i+1,x]$ is rooted and centered at $x$ and the parent of $x$ is at
  $i$. Hence $B[i,x]$ is a dangling star. By Case~1 we may assume that
  $x\leq j-2$. Flip $B[i,x]$. Embed $r$ at $i$ and $S$ recursively at
  $[j,i+1]$. See \figurename~\ref{fig:unary_regular_ij_used_3}. Since
  $\{i,j\}\not\in\EB$ after flipping and $i+1$ is isolated in
  $B[i+1,j]$, this works unless $[j,i+1]$ is in conflict for $S$. Then
  $\treeatt{[i+1,j]}{j}$ is a central-star that is rooted at the root of
  $B$. But this contradicts LSFR at $j$ before flipping: a
  contradiction. Hence, there is no conflict for $S$.
\end{proof}

\begin{proposition}\label{prop:rec_unary_regular_ij_not_used}
  If $\deg_R(r)=1$, $S$ is not a star, and
  $\{i,j\}\not\in\EB$, then there is an ordered plane packing
  of $B$ and $R$ onto $I$.
\end{proposition}
\begin{proof}
  The general plan is to embed $r$ onto $i$ and $S$ recursively onto
  $[j,i+1]$. Since $\{i,j\}\not\in\EB$ and $S$ is not a star,
  this works unless (1) $B[i+1,j]$ is a star or (2) there is a
  conflict for embedding $S$ onto $[j,i+1]$. In case (2), the star
  $B^*:=\treeatt{[j,i+1]}{j}$ is either in edge-conflict or in
  degree-conflict for embedding $S$. If it is in edge-conflict, then
  there must be an edge from the root of $B^*$ to $r$. By 1SR, the root
  of $B^*$ must be at $j$. But that means that
  $\{i,j\}\in\EB$, a contradiction. Thus, in case (2), there
  is a degree-conflict for embedding $S$ onto $[j,i+1]$. We deal with
  these cases below.

  \case{1} $B[i+1,j]$ is a star. Since $\{i,j\}\not\in\EB$, vertex $i$
  is isolated in $B$. Flip $\treeat{j}=B[i+1,j]$ if necessary to put the
  center of $B[i+1,j]$ at $j$. If the root of $B[i+1,j]$ is at $i+1$,
  then embed $r$ onto $j$ and $S$ recursively onto the independent set
  $[i,j-1]$. Since $i$ is isolated in the blue embedding, $[i,j-1]$ is
  not in conflict for $S$. If the root of $B[i+1,j]$ is at $j$, then
  flip the blue embedding at $[j-1,j]$. This places the root at $j-1$
  and a leaf of the star at $j$. After flipping, the interval $[i,j-1]$
  still satisfies the invariants. Embed $r$ onto $j$ (which is not in
  edge-conflict) and $S$ recursively onto $[i,j-1]$. See
  \figurename~\ref{fig:unary_regular_ij_not_used_star}. Since $i$ is
  isolated in the blue embedding, $[i,j-1]$ is not in conflict for $S$.

  \begin{figure}[b]
    \centering
    \subfloat[Case~1]{\includegraphics{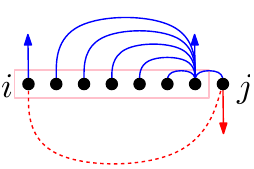}\label{fig:unary_regular_ij_not_used_star}}\hfil
    \subfloat[Case~2.1]{\includegraphics{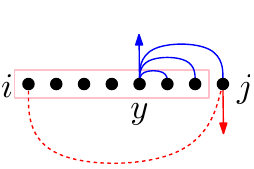}\label{fig:unary_regular_ij_not_used_no_conflict}}\hfil
    \subfloat[Case~2.2:$k_2>0$]{\includegraphics{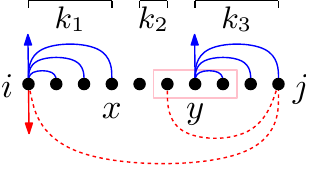}\label{fig:unary_regular_ij_not_used_conflict_k2}}\hfil
    \subfloat[Case~2.2:
    $k_2=0$]{\includegraphics{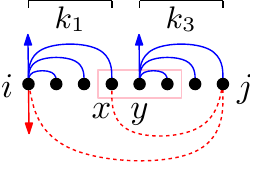}\label{fig:unary_regular_ij_not_used_conflict_no_k2}}\hfil
    \caption{The case analysis in the proof of
      Proposition~\ref{prop:rec_unary_regular_ij_not_used}.}
    \label{fig:unary_regular_ij_not_uses}
  \end{figure}

  \case{2} There is a degree-conflict for embedding $S$ onto $[j,i+1]$.
  Let $y$ be such that $\treeatt{[j,i+1]}{j}=B[y,j]$. Due to the
  degree-conflict, $B[y,j]$ is a central-star on at least three
  vertices. Since $\treeatt{[j,i+1]}{j}=B[y,j]$, the root of $B[y,j]$ is
  not adjacent to any vertex in $[i+1,y-1]$. By 1SR, if it were adjacent
  to $i$, then the root of $B[y,j]$ must be at $j$: this however,
  violates the assumption that $\{i,j\}\not\in\EB$. Hence,
  $B[y,j]=\treeat{j}$. Since $B[y,j]$ is a tree and $y\leq j-2$,
  $B[i,j-1]$ is not a star. We distinguish two cases.

  \case{2.1} There is no conflict for embedding $S$ onto $[i,j-1]$.
  Flip $B[y,j]$ if necessary to put its root at $y$. This preserves
  all invariants on $[i,j]$. Embed $r$ onto $j$ (which is not the root
  of $B[y,j]$) and $S$ recursively onto $[i,j-1]$. See
  \figurename~\ref{fig:unary_regular_ij_not_used_no_conflict}. This
  works by the assumption that there is no conflict for embedding
  $S$ onto $[i,j-1]$ before flipping $B[y,j]$ and since
  $\treeat{i}\neq\treeat{j}$ due to $\{i,j\}\not\in\EB$.

  \case{2.2} There is a conflict for embedding $S$ onto $[i,j-1]$.  By
  the 1SR and the fact that $\{i,j\}\not\in\EB$, there is a
  degree-conflict for embedding $S$ onto $[i,j-1]$. Let $x$ be such that
  $\treeatt{[i,j-1]}{i}=B[i,x]$. By the same argumentation that proved
  $B[y,j]=\treeat{j}$ we have $B[i,x]=\treeat{i}$. Thus, we can divide
  $B$ into three disjoint parts: $B[i,x]$ (a central-star),
  $B[x+1,y-1]$ (about which we know nothing), and $B[y,j]$ (a
  central-star). For notational convenience, let $k_1=|[i,x]|$,
  $k_2=|[x+1,y-1]|$, and $k_3=|[y,j]|$ be the corresponding sizes. Let
  $d=\deg_S(s)$ and let $v_1,\dots,v_d$ be the children of $s$, ordered
  by increasing size of their subtrees ($\tr_S(v_d)$ is the largest).
  Since $S$ is not a star $|\tr_S(v_d)|\geq 2$. Let $\lambda$ be the
  number of leaf children of $s$. Then $|\tr_S(v_\ell)|=1$ if and only if
  $\ell\leq\lambda$.

  Flip $B[i,x]$ if necessary to put the root (and center) at $i$ and
  flip $B[y,j]$ if necessary to put the root (and center) at $y$. We
  first explain how to embed $R$ and then prove that it always works.
  Refer to \figurename~\ref{fig:unary_regular_ij_not_used_conflict_k2}
  for the case $k_2>0$ and
  \figurename~\ref{fig:unary_regular_ij_not_used_conflict_no_k2} for the
  case $k_2=0$. Embed $r$ onto $i$ and $s$ onto $j$. This works so far:
  by the peace invariant the root of $B[i,x]=\treeat{i}$ is not
  in conflict and $\{i,j\}\not\in\EB$. Next, embed $\tr(v_d)$
  recursively onto $[y-1,y+|\tr(v_d)|-2]$. Since
  $\{y-1,j\}\not\in\EB$ and $y-1$ is isolated in
  $B[y-1,y+|\tr(v_d)|-2]$, this works provided $\tr(v_d)$ fits inside
  $[y-1,j-1]$, i.e. provided $|\tr(v_d)|\leq |[y-1,j-1]|=|[y,j]|=k_3$.
  Next, embed a leaf child of $s$ on each vertex in $[x+1,y-2]$ (this
  interval may be empty). This embeds the children $v_1,\dots,v_{k_2-1}$
  and works provided that $\lambda\geq k_2-1$. This leaves two disjoint
  intervals to embed the remaining subtrees
  $\tr(v_{\max(1,k_2)}),\dots,\tr(v_{d-1})$ of $s$:
  $I_1:=[i+1,\min(x,y-2)]$ and $I_2:=[y+|\tr(v_d)|-1,j-1]$. Thus, it
  remains to prove that (i) $|\tr(v_d)|\leq k_3$, and (ii) $\lambda\geq
  k_2-1$, and that (iii) we can distribute the remaining subtrees over
  $I_1$ and $I_2$.

  We begin by showing that $d$ must be large. Since there is a
  degree-conflict for embedding $S$ onto $[j,i+1]$ we have
  $k_1+k_2-1<d$, and since there is a degree-conflict for embedding $S$
  onto $[i,j-1]$ we have $k_2+k_3-1<d$:
  \begin{align}
    k_1+k_2&\leq d;\label{eq:unary_2dc_left}\\
    k_2+k_3&\leq d.\label{eq:unary_2dc_right}
  \end{align}
  Recall that $k_1+k_2+k_3=|R|=|S|+1$. Adding~\eqref{eq:unary_2dc_left}
  and~\eqref{eq:unary_2dc_right} yields $2d\geq
  k_1+k_2+k_3+k_2=|S|+1+k_2$ and so
  \begin{equation}
    \label{eq:unary_2dc_deglower}
    d\geq\frac{|S|+1+k_2}{2},
  \end{equation}
  \subsubparagraph{Proof of (i)} We must show that $|\tr(v_d)|\leq k_3$.
  Using~\eqref{eq:unary_2dc_left} we get
  $\sum_{\ell=1}^{d-1}|\tr(v_\ell)|\geq d-1\geq k_1+k_2-1=|S|-k_3$. Since the
  total size of the subtrees at the children of $S$ is $|S|-1$ we have
  $|\tr(v_d)|=|S|-1-\sum_{\ell=1}^{d-1}|\tr(v_\ell)|\leq |S|-1-|S|+k_3=k_3-1<k_3$,
  which completes the proof of (i). 

  \subsubparagraph{Proof of (ii)} We must show that $\lambda\geq k_2-1$.
  Since $|\tr(v_\ell)|\geq 2$ for all $\ell$, $\lambda+1\leq \ell\leq d$,
  we have $2(d-\lambda)+\lambda\leq |S|-1$ and so
  \[\lambda\geq 2d-|S|+1\overset{\eqref{eq:unary_2dc_deglower}}{\geq}
  (|S|+1+k_2)-|S|+1=k_2+2.\]

  \subsubparagraph{Proof of (iii)} It remains to prove that we can
  distribute $\tr(v_{\max(1,k_2)}),\dots,\tr(v_{d-1})$ over the disjoint
  intervals $I_1=[i+1,\min(x,y-2)]$ and $I_2=[y+|\tr(v_d)|-1,j-1]$. We
  use the following observation on partitioning natural numbers.
  \begin{observation}
    \label{obs:unary_2dc_partition}
    Let $n$ and $t$ be positive integers with $t\geq \lfloor
    n/2\rfloor+1$ and let $a_1\leq\dots\leq a_t$ be positive integers
    with $\sum_{i=1}^t a_i=n$. Then for all $0\leq k\leq n$ there exists
    a set $J_k\subseteq [1,t]$ such that $\sum_{i\in J_k} a_i=k$.
  \end{observation}
  \begin{proof}
    We prove the statement by induction on $n$. The statement is true
    for $n=1$: in this case we must have $t=1$ and $a_1=1$, and so
    $J_0=\emptyset$ and $J_1=\{1\}$ work. Suppose that the statement
    holds for all positive integers smaller than $n$. It suffices to
    prove the statement for $k\geq \lceil n/2 \rceil$ since we can
    choose $J_k=[1,t]\setminus J_{n-k}$ for $k<\lceil n/2 \rceil$. If
    $a_t=1$ then $a_1=\dots=a_t=1$ and we choose $J_k=[1,k]$. Otherwise,
    by the assumption on $t$ we have $a_t=n-\sum_{i=1}^{t-1}a_i\leq
    n-t+1\leq \lceil n/2 \rceil$ and hence $k-a_t\geq 0$.
    By the assumption on $t$ and since $a_t\geq 2$ we have $t-1\geq
    \lfloor n/2 \rfloor\geq \lfloor (n-a_t)/2\rfloor +1$. Hence, by the
    induction hypothesis, there exists a set $J_{k-a_t}\subseteq
    [1,t-1]$ with $\sum_{i\in J_{k-a_t}}a_i= k-a_t$. Choose
    $J_k=J_{k-a_t}\cup\{t\}$ to complete the proof.
  \end{proof}
  The total size of the remaining subtrees is
  $n:=|S|-1-\sum_{\ell=1}^{k_2-1}|\tr(v_\ell)|-|\tr(v_d)|\leq
  |S|-1-\max(0,k_2-1)-2=|S|-2-\max(1,k_2)$ since $|\tr(v_d)|\geq 2$.
  Then
  \[t:=d-1\overset{\eqref{eq:unary_2dc_deglower}}{\geq}
  \frac{|S|+1+k_2}{2}-1= \frac{|S|-3+k_2}{2}+1\geq \frac{n}{2}+1,\]
  where the last step uses that $|S|-3+k_2\geq |S|-2-k_2$ for $k_2\geq
  1$ and $|S|-3+k_2\geq |S|-2-1$ for $k_2=0$.
  Hence, $n$ and $t$ satisfy the precondition of
  Observation~\ref{obs:unary_2dc_partition}. We apply the observation
  with $k=|I_1|$. This gives us a set $J_k$ such that $\sum_{\ell\in J_k}
  |\tr(v_\ell)|=|I_1|$ and $\sum_{\ell\in [1,d-1]\setminus J_k} |\tr(v_\ell)|=|I_2|$.

  Since $B[I_1]$ and $B[I_2]$ have no internal edges and no edges to the
  position of $r$ at $j$, we can embed the subtrees $\tr(v_\ell)$ with
  $\ell\in S_k$ explicitly from left to right on $I_1$ and the remaining
  subtrees explicitly from left to right on $I_2$. This completes the
  proof.
\end{proof}

Proposition~\ref{prop:rec_unary_star},
Proposition~\ref{prop:rec_unary_regular_ij_used}, and
Proposition~\ref{prop:rec_unary_regular_ij_not_used} together prove the
following.

\begin{lemma}
  \label{lem:rec_unary}
  If $\deg_R(r)=1$, then there is an ordered plane packing of $B$ and
  $R$ onto $I$.
\end{lemma}

\section{Embedding the red tree: a singleton subtree}
\label{subsec:rec_singleton}
Here we completely handle the case $|S|=1$.

\begin{lemma}\label{lem:rec_singleton}
  If $|S|=1$, then $R$ and $B$ admit an ordered plane packing onto
  $[i,j]$.
\end{lemma}
\begin{proof}
  We distinguish two cases.

  \case{1} $R^-$ is not a star. We first describe an embedding that
  works whenever $B[i,j-1]$ is a star. Flip $B[i,j-1]$ if necessary to
  put its center at $j-1$. In addition to the star at $[i,j-1]$, the
  blue embedding may use the edge $\{i,j\}$. Note that it cannot use
  $\{j-1,j\}$, as this would imply that $B$ is a star. Thus, $j$ is
  isolated in $B[i+1,j]$. Embed $r$ onto $i+1$ and $s$ onto $i$. Let $U$
  be a largest subtree of $r$ in $R^-$. Since $R^-$ is not a star,
  $|U|\geq 2$. Embed $U$ recursively onto $[j,j-|U|+1]$. Since $j$ is
  locally isolated in $B[j,j-|U|+1]$ and $j$ is not adjacent to $i+1$
  (which is where we embedded $r$), this always works. Embed the
  remaining subtrees of $r$ in $R^-$ explicitly on $[i+2,j-|U|]$.

  Assume now that $B[i,j-1]$ is not a star.
  If $\{i,j\}\not\in E(B)$, then we embed $s$ at $j$, and recursively
  embed $R^-$ onto $[i,j-1]$.
  $R^-$ has no edge-conflict with $[i,j-1]$ by the peace invariant.
  It also has no degree-conflict with $[i,j-1]$:
  otherwise $R$ would already have had a degree-conflict with $[i,j]$.

  So assume that $\{i,j\}\in E(B)$. Flip $B$ if necessary to put
  its root at $j$. If $B[i,j-1]$ is a star now, then use the embedding
  described in the first paragraph to find an ordered plane packing.
  Otherwise, $r$ is not in edge-conflict with any vertex in $[i,j-1]$.
  The general plan is to embed $s$ at $j$ and $R^-$ recursively onto
  $[j-1,i]$. Since $B$ is not a star and $B$ is rooted at $j$, the edge
  $\{j-1,j\}$ is not used. Hence, this works unless there is a
  conflict for embedding $R^-$ onto $[j-1,i]$. This means in
  particular that $\treeatt{[i,j-1]}{j-1}$ is a central-star
  $B^*=B[x,j-1]$. See
  \figurename~\ref{fig:singleton_rm_no_star_default}. By assumption,
  $i+1\leq x\leq j-2$. Due to the presence of the edge $\{i,j\}$ and
  since $x\geq i+1$, the root (and hence also the center) of $B^*$ must
  be at $x$.

  \begin{figure}[b]
    \centering\hfil%
    \subfloat[Case~1]{\includegraphics{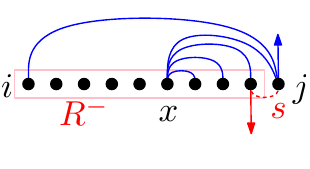}\label{fig:singleton_rm_no_star_default}}\hfil%
    \subfloat[Case~1.1]{\includegraphics{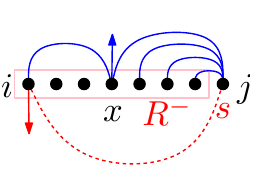}\label{fig:singleton_rm_no_star_x_large_1}}\hfil%
    \subfloat[Case~1.1]{\includegraphics{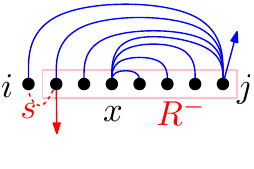}\label{fig:singleton_rm_no_star_x_large_2}}\hfil%
    \subfloat[Case~1.2]{\includegraphics{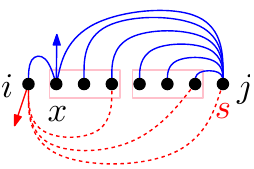}\label{fig:singleton_rm_no_star_x_small}}\hfil%
    \caption{Case~1 in the proof of Lemma~\ref{lem:rec_singleton}.}
  \end{figure}

  \case{1.1} $x\geq i+2$. Flip $B[x,j]$. Note that afterwards
  $\{i,j\}\not\in\EB$ and $B[i,j-1]$ satisfies 1SR and LSFR.
  Embed $s$ onto $j$ and $R^-$ recursively onto $[i,j-1]$. See
  \figurename~\ref{fig:singleton_rm_no_star_x_large_1}. Since $|B^*|\geq
  2$, the interval $[i,j-1]$ contains at least one leaf of $B^*$ and so
  $B[i,j-1]$ is not a star. Hence, this works unless there is a
  conflict for embedding $R^-$ onto $[i,j-1]$. In that case, note
  that $\treeatt{[i,x]}{i}$ is now formed by the root of $B$ and its
  subtrees other than $B^*$. Since $\treeatt{[i,x]}{i}$ is a
  central-star, it follows that the subtrees of the root $b$ of $B$
  other than $B^*$ are all leaves. Flip $B[x,j]$ again to restore the
  original embedding. Embed $s$ onto $i$ and $R^-$ recursively onto
  $[i+1,j]$. See \figurename~\ref{fig:singleton_rm_no_star_x_large_2}.
  Since $x\geq i+2$, $B[i+1,j]$ is a tree that is not a star and
  $\{i,i+1\}\not\in\EB$. Hence, the peace invariant holds for
  $R^-$.

  \case{1.2} $x=i+1$. Flip $B[x,j]$. Embed $r$ onto $i$ and $s$ onto
  $j$. Embed the remaining subtrees of $r$ in $R$ explicitly onto the
  independent set $B[i+1,j-1]$, putting the largest one (which has size
  at least two) next to $i$. See
  \figurename~\ref{fig:singleton_rm_no_star_x_small}.

  \case{2} $R^-$ is a star. Then $\deg_R(r)=2$ and the child $q$ of $r$
  in $R^-$ is the root and center of a star $Q=\tr(q)$.

  \case{2.1} $\{i,j\}\in\EB$. Let $b$ be the root of $B$. If
  $\deg_B(b)=1$, then flip $B$ if necessary to put its root at
  $j$. Then $j$ is isolated in $B[i+1,j]$ and $\{i,i+1\}\not\in\EB$
  since $B$ is not a star and by LSFR. Embed $r$ onto $i+1$, $s$ onto
  $i$, $q$ onto $j$, and the children of $q$ onto $[j-1,i+2]$. See
  \figurename~\ref{fig:singleton_rm_star_ij_1}.

  \begin{figure}[t]
    \centering\hfil%
    \subfloat[Case~2.1]{\includegraphics{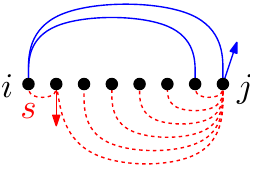}\label{fig:singleton_rm_star_ij_1}}\hfil%
    \subfloat[Case~2.1]{\includegraphics{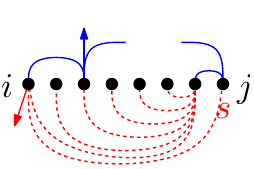}\label{fig:singleton_rm_star_ij_2}}\hfil%
    \subfloat[Case~2.2]{\includegraphics{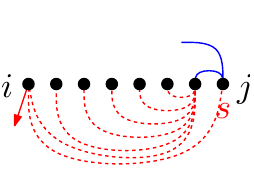}\label{fig:singleton_rm_star_no_ij_1}}\hfil%
    \subfloat[Case~2.2]{\includegraphics{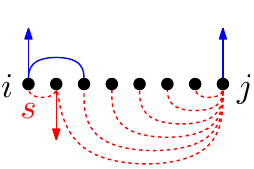}\label{fig:singleton_rm_star_no_ij_2}}\hfil%
    \subfloat[Case~2.2]{\includegraphics{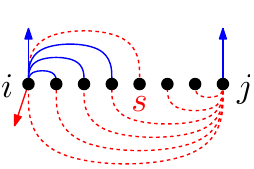}\label{fig:singleton_rm_star_no_ij_3}}\hfil%
    \caption{Case~2 in the proof of Lemma~\ref{lem:rec_singleton}.}
  \end{figure}

  If $\deg_B(b)\geq2$, then flip $B$ if necessary to put its root
  at $j$. Let $x$ be such that $\treeatt{[i,j-1]}{i}=B[i,x]$, which is
  a smallest subtree of $b$. Since $B$ is not a star, $B[x+1,j]$
  is not a central-star. Flip $B[x+1,j]$. This
  puts the root $b$ at $x+1$. Use a leaf-isolation-shuffle on $B[x+1,j]$
  to embed a leaf at $j-1$, its parent of $j$, and the root at $x+1$.
  This works by Proposition~\ref{prop:leafshuffle}. Embed $r$ onto $i$,
  $s$ onto $j$, $q$ onto $j-1$ and the children of $q$ onto $[j-1,i+1]$.
  See \figurename~\ref{fig:singleton_rm_star_ij_2}.

  \case{2.2} $\{i,j\}\not\in\EB$. Then $\treeat{i}\neq\treeat{j}$. If
  $|\treeat{j}|\geq2$, then perform a leaf-isolation-shuffle to put a
  leaf at $j-1$ and its parent at $j$. Since $\treeat{i}\neq\treeat{j}$,
  this does not touch the blue vertex at $i$. Embed $r$ onto $i$, $s$
  onto $j$, $q$ onto $j-1$, and the children of $q$ onto $[j-2,i+1]$.
  See \figurename~\ref{fig:singleton_rm_star_no_ij_1}.

  If $|\treeat{j}|=1$ and $\treeat{i}$ is not a central-star, then
  flip $\treeat{i}$ if necessary to put its root at $i$. Since it is
  not a central-star, $\{i,i+1\}\not\in\EB$. Embed $r$ onto $i+1$, $s$
  onto $i$, $q$ onto $j$, and the children of $q$ onto $[j-1,i+2]$. See
  \figurename~\ref{fig:singleton_rm_star_no_ij_2}.

  Finally, if $|\treeat{j}|=1$ and $\treeat{i}$ is a central-star, then
  let $x$ such that $B[i,x]=\treeat{i}$. We have $x\leq j-2$ by the
  peace invariant. Flip $B[i,x]$ if necessary to put its root
  at $i$. By the peace invariant, $i$ is not in edge-conflict
  with $r$. Embed $r$ onto $i$, $s$ onto $x+1$, $q$ onto $j$, and the
  children of $q$ onto $[j-1,x+2]$ and $[x,i+1]$. See
  \figurename~\ref{fig:singleton_rm_star_no_ij_3}.
\end{proof}

\section{Embedding the red tree: a large blue star}
\label{subsec:rec_large_blue_star}
In this and the following section we handle the case that $B[i,j-|S|]$
is a star. The graphs $S$, $R^-$, and $B[j-|S|+1,j]$ may or not be
stars. The case that we actually handle is more general, as specified in
the following
\begin{lemma}\label{lem:rec_large_blue_star}
  If $B[i,x]$ is a star, for $x\in[j-|S|,j-1]$, then $R$ and $B$ admit
  an ordered plane packing onto $[i,j]$.
\end{lemma}
\begin{proof}
  By Lemma~\ref{lem:rec_unary} and Lemma~\ref{lem:rec_singleton}, we may
  assume $\deg_R(r)\geq2$ and $|S|\geq2$. The following observation does
  not depend on the context of this proof.
  \begin{observation}\label{obs:rmsize}
    $|R^-|\ne 2$.
  \end{observation}
  \begin{proof}
    If $|R^-|=2$, then by the minimality of $S$ we have $|S|=1$. It
    follows that $|R|=3$ and so $R$ is a star, contrary to our assumption.
  \end{proof}
  By Observation~\ref{obs:rmsize}, $|R^-|\ge 3$. Select $x$ maximally so
  that $B^*=B[i,x]$ is a star, and let $d=\deg_{R^-}(r)\ge 1$. Note that
  $|B[i,x]|\geq|R^-|\geq3$. We distinguish two cases.

  \case{1} $B^*$ is a central-star. Then by LSFR we have
  $\treeat{i}=B^*$. If necessary, flip $B^*$ to put its root and center
  at $i$. We will use a blue-star embedding to embed $R^-$ from
  $\sigma=i$ with $\varphi=(x+1,\dots,x+d)$. Let us first check the
  conditions for the blue-star embedding. \ref{gg:ec} holds by
  \ref{inv:starconflict} for embedding $R$ onto $[i,j]$. For \ref{gg:dc}
  we must show $|R^-|\leq|B^*|+d$ and $|B^*|+d\leq|I|-1$. We wish to
  argue that at least one leaf of $B^*$ remains after the blue-star
  embedding, and thus we show $|R^-|<|B^*|+d$. This inequality holds
  since $x\geq j-|S|$ and $d\geq 1$. For the second inequality, by
  \ref{inv:starconflict}, we have $|B^*|\leq|I|-\deg_R(r)$ and so
  $|B^*|+d=|B^*|+\deg_{R}(r)-1\le|I|-1$. \ref{gg:int} and \ref{gg:cs}
  hold since $B\setminus(B^*\cup\varphi)$ forms an interval. Hence, by
  Proposition~\ref{p:greedygrab}, the blue-star embedding succeeds and
  leaves an interval $[i',j']=[i',j]$ such that $j$ is not in
  edge-conflict for embedding $s$ and a non-empty prefix of $[i',j]$
  consists of isolated vertices that are in edge-conflict for embedding
  $s$ (these are leaves of $B^*$). Recursively embed $S$ onto $[j,i']$.
  This works unless $S$ is a star or $\treeatt{[i',j]}{j}$ is in conflict
  (which must be a degree-conflict) for embedding $S$.

  \case{1.1} $S$ is a star. If $S$ is a dangling star then embed $s$
  onto $j$, the child $s'$ of $s$ onto $i'$ (which is locally isolated),
  and the children of $s'$ onto $[i'+1,j-1]$. Otherwise, $S$ is a
  central-star. If there is a locally isolated vertex in $B[i',j]$ that
  is not in edge-conflict, then use this vertex to embed $s$ and embed
  the children of $s$ on the remainder. Otherwise, undo the blue-star
  embedding. Consider the blue vertex at $j$. It does not get consumed
  by the blue-star embedding. Since it was not isolated after the
  blue-star embedding, it is not isolated now. By choice of $x$, we
  have $\treeat{j}=\treeatt{[x+1,j]}{j}$. Perform a
  leaf-isolation-shuffle on $\treeat{j}$ to place a leaf $\ell$ at $j-1$
  and its parent at $j$. Perform the original blue-star embedding, but
  now with $\varphi=(j,x+1,\dots,x+d-1)$ if $d\geq 2$ and $\varphi=(j)$ if
  $d=1$. The conditions of the blue-star embedding still hold. The
  resulting interval $[i',j']$ contains the now isolated vertex $\ell$
  and we embed $S$ by placing $s$ onto $\ell$ and embedding the children
  of $s$ on the remainder.

  \case{1.2} $\treeatt{[i',j]}{j}$ is a central-star that raises a
  degree-conflict. Note that $[i',j]$ is composed of some locally
  isolated vertices plus some a suffix of the interval $[i,j]$ before
  the blue-star embedding. Undo the blue-star embedding. Now
  $B[z,j]$ is a central-star for some minimal $z$. We claim that we may
  assume that $B[z,j]$ is rooted at $j$. Indeed, if $B[z,j]$ is rooted
  at $z$, then by 1SR we have $B[z,j]=\treeat{j}$ and we can flip
  $\treeat{j}$ to establish the claim. Perform the original blue-star
  embedding for $R^-$, but now with $\varphi=(j,x+1,\dots,x+d-1)$ if $d\geq
  2$ and $\varphi=(j)$ if $d=1$. In the remaining interval $[i',j']$, the
  vertex $j'$ is a leaf of what was the central-star $B[z,j]$ before the
  blue-star embedding. Hence, $j'$ is locally isolated and not in
  edge-conflict with $s$. Recursively embed $S$ onto $[j',i']$ to
  complete the embedding.

  \case{2} $B^*$ is a dangling star. In this case \ref{inv:starconflict}
  does not tell us anything about the size of $B^*$ (because it applies
  to central-stars only). If the root of $B^*$ is at $x$, then its
  center is at $i$ and by 1SR $i$ is the only neighbor of $x$ in $B$.
  Hence by flipping $B^*$ we may suppose that the root of $B^*$ is at
  $i$. Note that $i$ may have more neighbors, in addition to the center
  of $B^*$ at $x$. Also note that $i$ may be in conflict with $r$, in
  case we flipped $B^*$ (the original vertex at $i$ cannot be in
  conflict by \ref{inv:placement}). We distinguish two cases.

  \case{2.1} $x=j-1$. In this case we know almost completely what $B$
  looks like: $B[i,j-1]$ is a star rooted at $i$ and centered at $j-1$
  and the edge $\{i,j\}$ may or may not be used. We embed $R$ explicitly
  as follows. Since $\deg_R(r)\geq2$, there is a subtree $W=\tr(w)$ of
  $r$ different from $S$. Embed $r$ onto $i+|W|$ and embed $W$
  explicitly onto the independent set at $[i,i+|W|-1]$. Since
  $|S|\geq2$, we know that $|R'|<|B[i,j-1]|$, and hence $r$ is not
  embedded at the center of the star $B[i,j-1]$. If $S$ is not a star,
  embed it recursively onto $[j,j-|S|+1]$. This works because $j$ is
  locally isolated in $B[j-|S|+1,j]$ and $j$ is not adjacent to $i+|W|$
  (which is where we embedded $r$). If $S$ is a central-star, embed $s$
  onto $j$ and its children onto $[j-1,j-|S|+1]$. If $S$ is a dangling
  star, embed $s$ onto $j-|S|+1$, the child $s'$ of $s$ onto $j$, and
  the children of $s'$ onto $[j-1,j-|S|+2]$. Embed the remaining
  subtrees (if any) of $r$ on the remaining interval $[i+|W|+1,j-|S|]$,
  which forms a locally independent set, none of whose vertices are
  adjacent to $r$.

  \case{2.2} $x\le j-2$ and $\treeatt{[j,j-|S|+1]}{j}$ is a central-star
  $B^{**}$ on $|B^{**}|\ge|S|-\deg_S(s)+1$ (in particular, this holds if
  $S$ has a degree-conflict for embedding on $[j,j-|S|+1]$). We
  distinguish two subcases.

  \case{2.2.1} $\{i,j\}\in\EB$. Then the root and center $b^{**}$ of
  $B^{**}$ must be at $j$ and cannot be the root of $B$ because then
  LSFR would imply that $B$ is a star. Therefore $i$ is the root of $B$
  (\figurename~\ref{fig:large_blue_star_221_ij_1}) and it is not in
  conflict with $r$ due to \ref{inv:placement}. We modify the embedding
  of $B$ as follows: Move $b^{**}$ to $j-|S|$ and all leaves of $B^{**}$ (as
  $|B^{**}|\ge|S|-\deg_S(s)+1\ge 2$, there is at least one) in sequence
  immediately to the right of $b^{**}$, at position $j-|S|+1$ and onward,
  shifting all vertices between there and $j$ to the right accordingly.
  Draw the edge $\{i,b^{**}\}$ below the spine to avoid crossings, and all
  other edges incident to $b^{**}$ above the spine
  (\figurename~\ref{fig:large_blue_star_221_ij_2}).
  \begin{figure}[htbp]
    \centering\hfil%
    \subfloat[]{\includegraphics{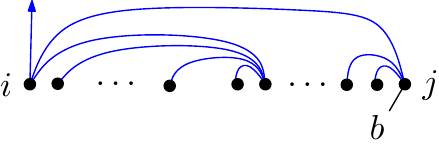}\label{fig:large_blue_star_221_ij_1}}\hfil
    \subfloat[]{\includegraphics{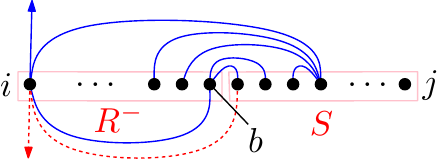}\label{fig:large_blue_star_221_ij_2}}\hfil
    \caption{$\{i,j\}\in\EB$
      (Case~2.2.1).\label{fig:large_blue_star_221_ij}}
  \end{figure}

  We place $r$ at $i$ and explicitly embed $R^-$ onto $[i,j-|S|]$, which
  in $B$ consists of a single edge $\{i,j-|S|\}$ with isolated vertices
  (at least one because $|R^-|\ge 3$) in between. Recall that $R^-$ is
  not a central-star and so we can embed it as described. It remains to
  embed $S$ onto $[j-|S|+1,j]$. As $j-|S|+1$ is a leaf of $B^{**}$,
  which is isolated on $[j-|S|+1,j]$, there is no conflict for this
  embedding and $B[j-|S|+1,j]$ is not a star. Therefore, if $S$ is not a
  star, then we can complete the packing recursively by embedding $S$
  onto $[j-|S|+1,j]$.

  It remains to consider the case that $S$ is a star. If $S$ is a
  central-star, then we can put this center at the locally isolated
  vertex $j-|S|+1$. Otherwise, $S$ is a dangling star with $|S|\ge 3$.
  As $|B^{**}|\ge|S|-\deg_S(s)+1=|S|\ge 3$, we have at least two locally
  isolated vertices (leaves of $B^{**}$) at $j-|S|+1$ and $j-|S|+2$. We
  put the root of $S$ at $j-|S|+1$ and the center at $j-|S|+2$ to
  complete the packing.

  \case{2.2.2} $\{i,j\}\notin\EB$. Let us consider the
  central-star $B^{**}=\treeatt{[j,j-|S|+1]}{j}$. We claim that
  $B^{**}=\treeat{j}$. Indeed, if the root $b^{**}$ of $B^{**}$ is on
  the left, then by 1SR $B^{**}=\treeat{j}$. Otherwise, $b^{**}$ is at
  $j$. By definition of $B^{**}$, $b^{**}$ has no neighbors in
  $B[j,j-|S|+1]\setminus B^{**}$. Since $B^*=B[i,x]$ is a star and
  $x\geq j-|S|$, the only remaining possible neighbor of $b^{**}$ would
  be $i$, but this is excluded by the assumption. We conclude that
  $B^{**}=\treeat{j}$. If necessary, flip $B^{**}$ to put its root (and
  center) at $j$.

  \case{2.2.2.1} $x=j-|S|$
  (\figurename~\ref{fig:large_blue_star_221_1}). Then we change the
  embedding of $B$ by moving one leaf $\ell$ of $B^{**}$ all the way to
  the left at $i$. As a leaf of $B^{**}$ it is not in conflict with $r$,
  and so we can map $r$ to $\ell=i$ and embed $R^-$ explicitly onto the
  locally independent set $B[i,j-|S|]$. If $S$ is not a star, then we
  recursively embed $S$ onto $[j-|S|+1,j]$
  (\figurename~\ref{fig:large_blue_star_221_2}). Note that $j-|S|+1$ is
  the center of $B^*$, which is isolated in $B[j-|S|+1,j]$ and not
  adjacent to the leaf of $B^{**}$ at $i$. Therefore, $B[j-|S|+1,j]$ is
  not a star and there is no degree-conflict and no edge-conflict for
  embedding $S$ onto $[j-|S|+1,j]$.
  \begin{figure}[htbp]
    \centering\hfil%
    \subfloat[]{\includegraphics{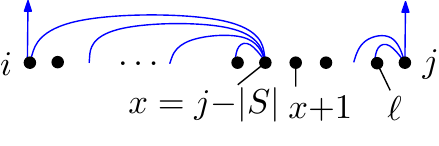}\label{fig:large_blue_star_221_1}}\hfil
    \subfloat[]{\includegraphics{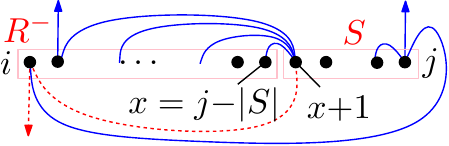}\label{fig:large_blue_star_221_2}}\hfil
    \caption{$x=j-|S|$ (Case~2.2.2.1).\label{fig:large_blue_star_221}}
  \end{figure}

  It remains to consider the case that $S$ is a star. If $S$ is a
  central-star, then the center can be embedded on the isolated vertex
  at $j-|S|+1$. Otherwise, $S$ is a dangling star with $|S|\ge 3$. Then
  at least one more leaf of $B^{**}$ remains at $j-1$, where we can
  embed the root of $S$. The center of $S$ is again mapped to the
  isolated vertex $j-|S|+1$ and the edge $\{j-|S|+1,j\}$ is drawn as a
  biarc, crossing the spine between $j-2$ and $j-1$.

  \case{2.2.2.2} $x\ge j-|S|+1$
  (\figurename~\ref{fig:large_blue_star_223_1}). Then we change the
  embedding of $B$ by simultaneously moving the root of $B^*$ to $x$ and
  and moving all other vertices of $B^*$ to the left by one
  (\figurename~\ref{fig:large_blue_star_223_2}). Embed $r$ at $i$. We
  will use a blue-star embedding to embed $S$ on $B[i+1,j]$ from
  $\sigma=j$ where $\varphi$ consists of the rightmost $\deg_S(s)$
  non-neighbors of $j$ in $B$ from right to left. Note that
  $\treeatt{[i+1,j]}{j}=B^{**}$ and hence $B^*=B^+$ in the terminology
  of the blue-star embedding. Let us check the conditions for the
  blue-star embedding. \ref{gg:ec} holds because $|R^-|\ge 3$ and so
  $i$ is a leaf of $B^*$ that is adjacent to $x-1\ne j$ only.
  \ref{gg:int} and \ref{gg:cs} hold because
  $B[i+1,j]\setminus(B^{**}\cup\varphi)$ forms an interval. For \ref{gg:dc}
  we must show $S|\leq|B^{**}|+\deg_S(s)$ and
  $|B^{**}|+\deg_S(s)\leq|I|-2$. The first inequality follows from the
  assumption of Case~2.2. For the second inequality, we have
  $|B^{**}|\leq |B[x+1,j]|\leq|S|-1$ and $\deg_S(s)\leq|S|-1$, and so
  $|B^{**}|+\deg_S(s)\leq 2|S|-2<|I|-2$, since $|S|<|I|/2$. Hence, the
  conditions for the blue-star embedding are satisfied.

  Since $\deg_S(s)\geq|S|-|B^{**}|+1$, we have $\varphi\supset
  B[j-|S|,j]\setminus B^{**}$, and hence the blue-star embeding embeds
  a child of $s$ onto the root of $B^{*}$, which was embedded at $x$,
  and the center of $B^{*}$, which was embedded at $x-1$. Therefore,
  the remaining vertices not used for the embedding of $S$ form an
  independent set in $B$ and we can explicitly embed $R^-$ on them.
  \begin{figure}[htbp]
    \centering\hfil%
    \subfloat[]{\includegraphics{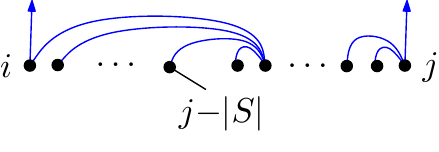}\label{fig:large_blue_star_223_1}}\hfil
    \subfloat[]{\includegraphics{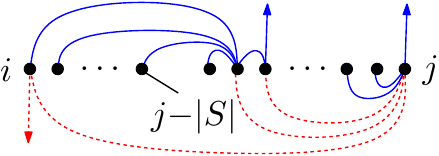}\label{fig:large_blue_star_223_2}}\hfil
    \caption{$x\ge j-|S|+1$
      (Case~2.2.2.2).\label{fig:large_blue_star_223}}
  \end{figure}

  \case{2.3} $x\le j-2$ ($\Rightarrow|S|\ge 2\Rightarrow|R|\ge 5$) and
  $\treeatt{[j,j-|S|+1]}{j}$ is \textbf{not} a central-star $B^{**}$ on
  $|B^{**}|\ge|S|-\deg_S(s)+1$ vertices. We first prove a claim and then
  distinguish two subcases.

  \emph{Claim:} We may suppose that $S$ is a star or $x=j-|S|$. To prove
  the claim, suppose that $x\ge j-|S|+1$. Then $j-|S|$ is a leaf of
  $B^*$ and we can explicitly embed $R^-$ onto the independent set
  $[j-|S|,i]$. As $x\neq j$ is the only neighbor of $j-|S|$ in $B$, we
  have $\{j-|S|,j\}\notin\EB$ and so there is no edge-conflict for
  embedding $S$ onto $[j,j-|S|+1]$. By assumption there is no
  degree-conflict for this embedding, either, and $B[j-|S|+1,j]$ is not
  a star because the root of $B^*$ is part of it but $B^*\ne B$. The
  only remaining obstruction for the recursive embedding of $S$ onto
  $[j,j-|S|+1]$ is $S$ being a star. This proves the claim.

  \case{2.3.1} $\{i,x+1\}\notin\EB$. Then we rearrange the
  embedding of $B$ as follows: move the center $c$ of $B^*$ to $j-|S|+1$
  and the vertex $b'$ at $x+1$ (the leftmost vertex not in $B^*$) to
  $j-|S|$. In order to avoid crossings with the edge(s) incident to
  $b'$, draw all edges between $c$ and its neighbors in $[i,j-|S|-1]$
  below the spine, whereas edges to neighbors in $[j-|S|+2,j]$ remain
  above the spine (\figurename~\ref{fig:large_blue_star}). After this
  transformation $B[i,j-|S|]$ is an independent set, on which we can
  embed $R^-$ explicitly. However, we have to take care because of the
  blue edges drawn below the spine and the (possibly) conflicting root
  $i$. Without loss of generality suppose that the root of $B^*$ at $i$
  is in conflict with $r$.
  \begin{figure}[htbp]
    \centering\hfil%
    \subfloat[]{\includegraphics{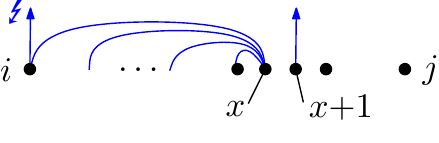}\label{fig:large_blue_star_1}}\hfil
    \subfloat[]{\includegraphics{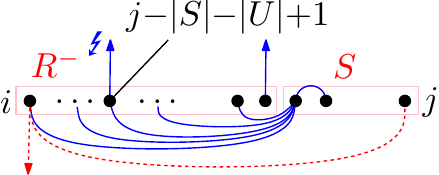}\label{fig:large_blue_star_2}}\hfil
    \caption{$x\le j-2$ and $\{i,x+1\}\notin\EB$
      (Case~2.3).\label{fig:large_blue_star}}
  \end{figure}

  Recall that $R^-$ is not a central-star and so there is at least one
  non-leaf child $u$ of $r$ in $R^-$. Denote $U=\tr(u)$ and map both $u$
  and the conflicting root of $B^*$ to $j-|S|-|U|+1$ (by exchanging the
  order of leaves of $B^*$ in $B[i,j-|S|-1]$). As $|U|\ge 2$, we have
  $j-|S|-|U|+1\le j-|S|-1$ and so the local order for the roots of
  subtrees from $B$ is maintained. On the other hand, we have
  $j-|S|=i+|R^-|-1$ and $|U|\le|R^-|-1$, which imply
  $j-|S|-|U|+1\ge(i+|R^-|-1)-(|R^-|-1)+1=i+1$. Therefore (the leaf now
  at) $i$ is not in conflict with $r$.

  As $\{i,x+1\}\notin\EB$ initially, after the transformation
  we have $\{j-|S|-|U|+1,j-|S|\}\notin\EB$ and so
  $[j-|S|-|U|+1,j-|S|]$ is an independent set in $B$. Therefore, we can
  embed $U$ onto $[x-|U|+1,x]$ explicitly, drawing all edges above the
  spine, and complete the embedding of $R^-$ by embedding $R^-\setminus
  U$ onto $[i,x-|U|]$ explicitly, again drawing all edges above the
  spine. After these changes to the embedding of $B$, the only neighbor
  of $i$ in $B$ is $j-|S|+1$. Together with $|S|\ge 2$ it follows
  that there is no edge-conflict for recursively embedding $S$ onto
  $[j,j-|S|+1]$. We also know that $B[j,j-|S|+1]$ is not a star because
  it contains part of $B^*$ (at least the center at $j-|S|+1$) and at
  least one more vertex not connected to that part of $B^*$: the vertex
  at $j$. (As $|S|\ge 2$, there were at least two such vertices
  initially, but one, the vertex $b'$, has been moved and used for
  embedding $R^-$.) Two possible obstructions for the recursive
  embedding of $S$ onto $[j,j-|S|+1]$ remain: a degree-conflict or $S$
  is a star. We conclude by considering both cases.

  \case{2.3.1.1} $S$ is a star. Undo the rearrangement of $B^*$. We will
  redo the rearrangement, but first do some other modifications.

  Suppose first that $|\treeatt{[x+1,j]}{x+1}|\geq2$ or
  $|\treeatt{[x+1,j]}{x+2}|=1$. In the former case, use a leaf-isolation
  shuffle on $B[x+1,j]$ to place a leaf at $x+2$ and its parent at
  $x+1$. We can apply the shuffle because $x\le j-2$ and therefore
  $|B[x+1,j]|\ge 2$. After the modification (as described in the first
  paragraph of Case~2.3.1) we proceed as follows. If $S$ is a
  central-star, then $s$ can be placed at $x+2$, which is adjacent to
  $j-|S|$ only and therefore locally isolated in $B[j-|S|+1,j]$.
  Otherwise, $S$ is a dangling star. Then either there is a (non-root)
  leaf of $B^*$ in $[j-|S|+1,j]$ or the center $c$ of $B^*$ is isolated
  in $[j-|S|+1,j]$. In either case, we put the center of $S$ on $x+2$.
  In the former case, we put the root of $S$ on $j-|S|+2$ (the leftmost
  leaf of $B^*$ in $[j-|S|+1,j]$, and draw the edge $\{c,x+2\}$ as a
  biarc that crosses the spine between $j-|S|+2$ and $j-|S|+3$. In the
  latter case, we put the root of $S$ on $j-|S|+1=c$. Either way, we can
  complete the star $S$ and the embedding of $R^-$ works just as before.

  Otherwise, $|\treeatt{[x+1,j]}{x+1}=1$ and
  $|\treeatt{[x+1,j]}{x+2}|\geq2$. Since $\{i,x+1\}\not\in\EB$ by the
  assumption of Case~2.3.1, we know that
  $\treeat{x+1}=\treeatt{[x+1,j]}{x+1}$ and hence also
  $\treeat{x+2}=\treeatt{[x+1,j]}{x+2}$. It follows that $x\leq j-3$ and
  hence $|S|\geq 3$. Perform a leaf-isolation-shuffle on $\treeat{x+2}$
  to place a leaf at $x+3$ and its parent at $x+2$. Rearrange the
  embedding of $B$ as follows: move the center $c$ of $B^*$ to
  $j-|S|+1$, the vertex $b'$ at $x+1$ to $j-|S|-1$, and the vertex $b''$
  at $x+2$ to $j-|S|$. Draw the blue edges as explained in the first
  paragraph of Case~2.3.1. We embed $S$ analogously to the previous
  paragraph, using $x+3$ as the location for the star-center. To embed
  $R^-$, let us consider the embedding $B[i,j-|S|]$. It is again an
  independent set. As opposed Case~2.3.1, however, we have local roots
  at $j-|S|-1$ and at $j-|S|$. Fortunately, since $|S|\geq3$ and $S$ is
  a smallest subtree, also $|U|\geq3$, and hence $j-|S|-|U|+1\leq
  j-|S|-2$, as required. Hence, we can embed $R^-$ explicitly,
  analogously to the second paragraph of Case~2.3.1.

  \case{2.3.1.2} There is a degree-conflict for embedding $S$ onto
  $[j,j-|S|+1]$. Then this conflict must have been created by our
  transformation of the embedding of $B$. Before this transformation
  there was no degree-conflict by assumption (Case~2.2 handles this
  scenario). In other words, $b'$ is the root of a star $B[x+1,j]$ in
  the initial embedding whose center is at $j$. After moving $b'$ out of
  the interval $[j-|S|+1,j]$, $j$ became the local root, which raised
  the degree-conflict. By our claim and the preceding Case~2.3.1.1 we
  may suppose that $x=j-|S|$
  (\figurename~\ref{fig:large_blue_star_2312_1}). We use a different,
  explicit embedding as follows: flip the star $B[x+1,j]$ so that its
  root is at $j$ and the center is at $x+1$ and draw all edges below the
  spine. Next move the center at $x+1$ to $i$ instead, shifting all
  vertices in between to the right by one. Then put $r$ at $i$ (not
  being the root of $B[x+1,j]$ it is not in conflict), and explicitly
  embed $R^-$ onto the (now) independent set $[i,x]$. Finally,
  explicitly embed $S$ onto the (now) independent set $[x+1,j]$
  (\figurename~\ref{fig:large_blue_star_2312_2}). Note that $B$ might be
  a tree (in which case the two roots in the figure are actually
  connected), but the embedding works also in this case.
  \begin{figure}[htbp]
    \centering\hfil%
    \subfloat[]{\includegraphics{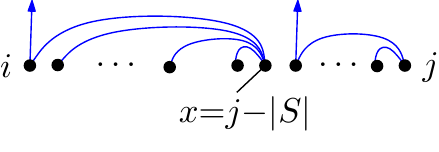}\label{fig:large_blue_star_2312_1}}\hfil
    \subfloat[]{\includegraphics{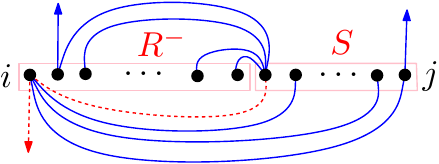}\label{fig:large_blue_star_2312_2}}\hfil
    \caption{A new degree-conflict for $S$ on $[j,j-|S|+1]$
      (Case~2.3.1.2).\label{fig:large_blue_star_2312}}
  \end{figure}

  \case{2.3.2} $\{i,x+1\}\in\EB$ and $x\ge j-|S|+1$. Then by
  our claim we may suppose that $S$ is a star. We embed $R^-$ explicitly
  onto $[j-|S|,i]$, noting that $j-|S|$ is a non-root leaf of $B^*$ and,
  therefore, not in conflict with $r$. If $S$ is a central-star, then we
  put the center at $x+1$, which is connected to $i$ only and therefore
  locally isolated on $[j-|S|+1,j]$. Otherwise, $|S|\ge 3$ and $S$ is a
  dangling star. Given that $x\le j-2$, we have $x+1\ne j$ and therefore
  can put the root of $S$ on $j$ and the center on $x+1$.

  \case{2.3.3} $\{i,x+1\}\in\EB$ and $x=j-|S|$. We distinguish
  two final subcases.

  \case{2.3.3.1} The root of $B[i,x+1]$ is at $i$. Then we change the
  embedding of $B$ by moving the vertex at $x+1$ to $i$ and shifting the
  vertices in between to the right by one. We explicitly embed $R^-$
  onto $[i,j-|S|]$. This is possible because $B[i,j-|S|]$ is an
  independent set except for the single edge $\{i,i+1\}$ and $R^-$ is
  not a central-star. Then if $S$ is a central-star, we embed
  the center at $j-|S|+1$, which is an isolated vertex in $[j-|S|+1,j]$.
  If $S$ is a dangling star, then we embed the root at $j$ and
  the center at $j-|S|+1$. Otherwise, $S$ is not a star and we
  recursively embed $S$ onto $[j-|S|+1,j]$. Recall that $j-|S|+1$ is a
  locally isolated vertex and $\{i,j-|S|+1\}\notin\EB$
  (because $i$ is a leaf whose only neighbor is at $i+1\ne j-|S|+1$).
  Therefore, there is no conflict for the recursive embedding and
  $B[j-|S|+1,j]$ is not a star.

  \case{2.3.3.2} The root of $B[i,x+1]$ is at $x+1$ (and possibly in
  conflict with $r$; \figurename~\ref{fig:large_blue_star_2331_1}).

  If $S$ is a central-star, then we change the embedding of $B$ as
  follows: First flip $B^*$ so that its center is at $i$ and then
  exchange the vertices at $i$ (center $c$ of $B^*$) and $i+1$ (a leaf
  of $B^*$). Put $r$ at $i$, which is a leaf of $B^*$ and therefore not
  in conflict. Then put $s$ at $x+1$, whose only neighbor is (now) at
  $x$ (originally at $i$), drawing the edge $\{r,s\}$ above the spine.
  Next put a leaf of $S$ on $i+1$ (center $c$ of $B^*$), again drawing
  the edge $\{s,c\}$ above the spine. Put the remaining leaves of $S$ on
  the vertices $[x+2,j-1]$, drawing the edges to $s$ below the spine.
  This leaves us with a set of isolated vertices, all accessible from
  below the spine, on which we can complete an explicit embedding of
  $R^-$ (\figurename~\ref{fig:large_blue_star_2331_2}).
  \begin{figure}[htbp]
    \centering\hfil%
    \subfloat[]{\includegraphics{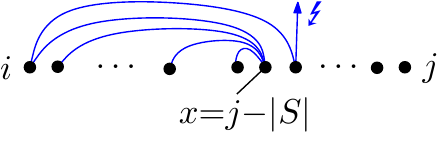}\label{fig:large_blue_star_2331_1}}\hfil
    \subfloat[]{\includegraphics{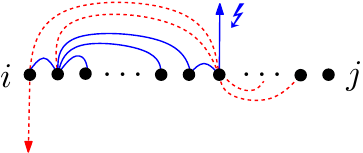}\label{fig:large_blue_star_2331_2}}\hfil
    \subfloat[]{\includegraphics{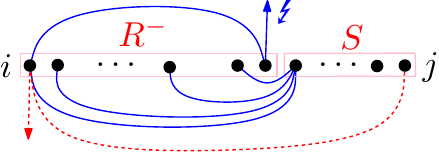}\label{fig:large_blue_star_2331_3}}\hfil
    \caption{$\{i,x+1\}\in\EB$, $x=j-|S|$ and $S$ is a
      central-star (Case~2.3.3.1).\label{fig:large_blue_star_2331}}
  \end{figure}

  A similar embedding also works for a dangling star $S$: Put $r$ at
  $i+2$ (which is another leaf of $B^*$ because $|R^-|\ge 3$), put $s$
  at $i$ and the center of $S$ at $x+1$.

  Otherwise, $S$ is not a star. Then we modify the embedding of $B$ by
  drawing all edges of $B^*$ below the spine and exchanging $x$ and
  $x+1$. Explicitly embed $R^-$ onto $[i,j-|S|]$. This is possible
  because $B[i,j-|S|]$ is an independent set except for the edge
  $\{i,j-|S|\}$ and $R^-$ is not a central-star. Recursively embed $S$
  onto $[j,j-|S|+1]$ (\figurename~\ref{fig:large_blue_star_2331_3}). As
  $j-|S|+1$ is a locally isolated vertex in $B[j-|S|+1,j]$, we know that
  $B[j-|S|+1,j]$ is not a star. There is no degree-conflict by
  assumption (Case~2.2 handles this scenario) and---as opposed to
  Case~2.3.1.2---we do not change $\treeat{j}$ here. Again by assumption
  $\{i,j\}\notin\EB$ and so there is no edge-conflict for the
  recursive embedding of $S$, either.
\end{proof}

\section{Embedding the red tree: a large red star}
\label{subsec:rec_large_red_star}
In this section we handle the case where $R^-$ is a star. If $R^-$ is a
star, then it must be a dangling star: otherwise, by the choice of $S$
as a smallest subtree, $R$ would be a star. Let $q$ be the child of $r$
in $R^-$ and let $Q=\tr(q)$. Then $Q$ is a central-star. Our default
approach in this case is to explicitly embed $Q$ and recursively embed
$S^+:=R\setminus Q$. Note that $\deg_{S^+}(r)=1$. Consequently, when we
try to recursively embed $S^+$ onto some interval $[x,y]$, there can be
a degree-conflict only if $B[x,y]$ is a star: a case we must handle
separately anyway. Hence, for a recursive embedding of $S^+$ it suffices
to check that we are not embedding against a star, to establish the
placement invariant, and to check that there is no edge-conflict.

\begin{proposition}\label{prop:rec_large_red_star_ij_not_used_sp_not_star}
  If $R^-$ is a star, $S^+$ is not a star, and
  $\{i,j\}\not\in\EB$, then $R$ and $B$ admit an ordered plane
  packing onto $[i,j]$.
\end{proposition}
\begin{proof}
  Let $d:=\deg_Q(q)$. We have $|S^+|\geq 4$ since $S^+$ is not a star.
  Hence, $|Q|\geq|S|\geq 3$. Flip $\treeat{j}$ if necessary to put its
  root at $j$. We first try the following. Use the red-star
  embedding to embed $q$ at $j$ and the children of $q$ on the
  $\deg_Q(q)$ rightmost non-neighbors of $j$ in $[i+1,j-1]$. Let $I'$ be
  the interval of remaining vertices. Embed $S^+$ recursively onto $I'$.
  See \figurename~\ref{fig:large_red_ij_not_used_default}.

  \begin{figure}[b]
    \centering\hfil%
    \subfloat[Default]{\includegraphics{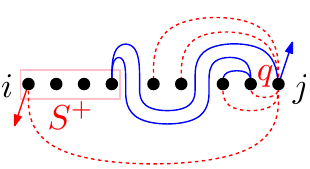}\label{fig:large_red_ij_not_used_default}}\hfil%
    \subfloat[Case~1.1]{\includegraphics{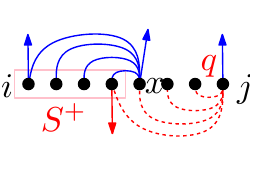}\label{fig:large_red_ij_not_used_ix_star_large}}\hfil%
    \subfloat[Case~1.2]{\includegraphics{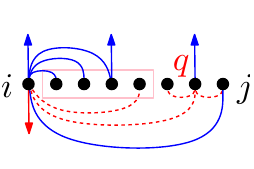}\label{fig:large_red_ij_not_used_ix_star_equal}}\hfil%
    \subfloat[Case~2]{\includegraphics{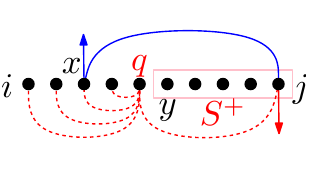}\label{fig:large_red_ij_not_used_dc_default}}\hfil%
    \caption{The case analysis in the proof of
      Proposition~\ref{prop:rec_large_red_star_ij_not_used_sp_not_star}~(Part~1/6).}
  \end{figure}

  Let us first consider the conditions under which the embedding of
  $S^+$ works.
  The embedding fails if $B[I']$ is a star, which happens only if
  (Case~1) $B^*:=\treeat{i}$ is a star with $|B^*|\geq|S^+|$ and
  $\deg_B(j)=0$. Otherwise, suppose there is an edge-conflict for
  embedding $S^+$ onto $I'$. Then $B^*=\treeatt{[I']}{i}$ is a
  central-star rooted at $b^*$. By choice of $I'$, we have
  $B^*=\treeat{i}$. By~\ref{inv:starconflict}, $b^*$ has no edge to the
  parent of $r$. If it had an edge to $j$ (which is where we embedded
  $q$), then by 1SR $b^*$ is embedded at $i$. But then $\{i,j\}\in\EB$,
  a contradiction. As argued at the start of this section, there can be
  no degree-conflict for $S^+$. Hence, \ref{inv:starconflict} holds.

  The embedding of $Q$ works unless there is a degree-conflict for
  placing $q$ onto $j$ and embedding the children of $q$ onto
  $[j-1,i+1]$, that is unless (Case~2)
  $\deg_Q(q)+\deg_B(j)\geq|I|-1$. We deal with both cases below.

  \case{1} $B^*:=\treeat{i}$ is a star with $|B^*|\geq|S^+|$ and
  $\deg_B(j)=0$. Let $x$ be such that $B[i,x]=B^*$.

  \case{1.1} $|B^*|>|S^+|$. Flip $B[i,x]$ if necessary to put its
  center at $x$. Embed $q$ onto $j$ and the children of $q$ explicitly
  onto $[j-1,j-\deg_Q(q)]$. See
  \figurename~\ref{fig:large_red_ij_not_used_ix_star_large}. This works
  since $\deg_B(j)=0$. Embed $S^+$ recursively onto $[j-\deg_Q(q)-1,i]$.
  This works since $|B^*|>|S^*|$ and hence $B[i,j-\deg_Q(q)-1]$ is an
  independent set and $j-\deg_Q(q)-1$ is not the root of $B^*$.

  \case{1.2} $|B^*|=|S^+|$. Flip $B[i,x]$ if necessary to put its center
  at $i$. By the peace invariant, this star-center is not in
  edge-conflict with $r$. Since $|S^+|\geq 4$, the blue vertex at $i+1$
  is a leaf in $B^*$ that is not the root of $B^*$. We change the blue
  embedding as follows. Simultaneously move $B[i+2,j]$ to $[i+1,j-1]$
  and $i+1$ to $j$. The edge $\{i,j\}$ is drawn in the lower
  halfplane. Note that $\treeatt{[i,x]}{x}$ is now an isolated vertex.
  Embed $q$ onto $j-1$ and the children of $q$ onto $j$ and
  $[j-2,j-\deg_Q(q)]$. This works because $j-1$ is isolated. Embed $r$
  at $i$. Embed $S$ recursively onto $[x,i+1]$ if $S$ is not a star. See
  \figurename~\ref{fig:large_red_ij_not_used_ix_star_equal}. Otherwise,
  $S$ is a dangling star. In this case, embed $s$ at $x$, the child $s'$
  of $s$ onto $i+1$ and the children of $s'$ onto $[i+2,x-1]$. This
  works because $j-1$ is isolated in $B$ (and so $\{i,j-1\}$ is not
  used) and $x$ is isolated in $B[i,x]$ (and so $\{i,x\}$ is not used
  and there is no conflict for embedding $S$ onto $[x,i+1]$).

  \case{2} $\deg_Q(q)+\deg_B(j)\geq|I|-1$. Then
  $\deg_B(j)\geq|I|-1-\deg_Q(q)=|S^+|$ and so $|S^+|<|\treeat{j}|$. Let
  $x$ and $y$ such that $B[x,j]=\treeat{j}$ and $|[y,j]|=|S^+|$. Since
  $|S^+|<|\treeat{j}|$ we have $x<y$. Flip $B[x,j]$ to put its root at
  $x$. The proof of this case will not rely on the peace invariant,
  except in Case~2.3.3.

  We first try the following. If $x$ has a subtree that is not a
  central-star on at least $|S^+|$ vertices, then rearrange $B[x,j]$ to
  put a smallest such subtree at $j$. Embed $q$ at $y-1$ and the
  children of $q$ at $[y-2,i]$. Embed $S^+$ recursively at $[j,y]$. See
  \figurename~\ref{fig:large_red_ij_not_used_dc_default}. This fails
  immediately if (1) $B[y,j]$ is a star, in which case every subtree of
  $x$ is a central-star on at least $|S^+|$ vertices or a dangling star
  on exactly $|S^+|$ vertices. Otherwise, suppose there is a
  edge-conflict for embedding $S^+$. Then $B^*=\treeatt{[j,y]}{j}$
  is a star rooted at a center $b^*$. Since $x$ is the only vertex on
  $[x,j]$ with edges to the outside of $I$, $b^*$ must be adjacent to
  $y-1$ (which is where we placed $q$). By 1SR, we have $b^*=j$. Hence,
  we must handle the case (2) $\{y-1,j\}\in\EB$ separately.
  This covers the possible issues with the recursive embedding of $S^+$.
  The embedding of $Q$ works unless (3) $y-1$ is not isolated in
  $B[i,y-1]$. We deal with these three cases next.

  \case{2.1} $B[y,j]$ is a star. Then by the rearrangement of $B[x,j]$
  performed above, every subtree of $x$ is a central-star on at least
  $|S^+|$ vertices or a dangling star on exactly $|S^+|$ vertices.

  \case{2.1.1} Some subtree $U$ of $x$ is a dangling star on exactly
  $|S^+|$ vertices. Re-embed $B[x,j]$, placing the root at $j$ and $U$
  as the closest subtree. Afterwards, $U=B[y-1,j-1]$, the center of $U$
  is at $j-1$, and $j$ is isolated in $B[y,j]$. Embed $r$ at $i$, $S$
  recursively onto $[j,y+1]$, $q$ onto $y$, and the children of $q$ onto
  $[y-1,i+1]$. The embedding of $S$ works because $\{i,j\}\not\in\EB$
  and $j$ is isolated in $B[y+1,j]$. The embedding of $Q$ works because
  $y$ is a leaf of $U$ and hence adjacent only to $j-1$ in $B$.

  \case{2.1.2} Every subtree of $x$ is a central-star on at least
  $|S^+|$ vertices. Flip $B[x,j]$ to put its root at $j$. Embed $r$
  onto $i$ and $s$ onto $j$. We embed $S$ explicitly, as follows. Let
  $c_1,\dots,c_d$ be the children of $S$ such that $\tr(c_1)$ is a
  largest subtree. Since $S$ is not a central-star,
  $|\tr(c_1)|\geq2$. Let $v_1,\dots,v_k$ be the children of $j$ ordered
  by proximity of $j$ ($v_1$ is the closest). By the assumption of
  Case~2 we have $\deg_B{j}\geq|S^+|$ and hence $k\geq d+1$. We embed
  $\tr(c_1)$ as follows. Reroute the edges of $v_2$ to its
  $|\tr(c_1)|-1$ rightmost neighbors via the lower halfplane. Embed
  $\tr(c_1)$ explicitly on these $|\tr(c_1)|-1$ rightmost neighbors of $v_2$
  and on $v_1$ in the upper halfplane. For $i\geq2$, we embed $\tr(c_i)$
  as follows. Reroute the edges of $v_{i+1}$ to its $|\tr(c_i)|$ rightmost
  neighbors via the lower halfplane and embed $\tr(c_i)$ explicitly on
  these vertices in the upper halfplane. Since we embedded a vertex of
  $\tr(c_1)$ on $v_1$, $j-1$ is isolated on the remainder. Embed $q$ onto
  $j-1$ and the children of $q$ onto the remainder. See
  \figurename~\ref{fig:large_red_ij_not_used_dc_star}.

  \begin{figure}[t]
    \centering\hfil%
    \subfloat[Case~2.1.2]{\includegraphics{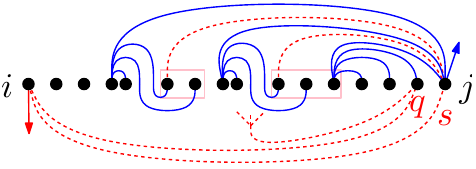}\label{fig:large_red_ij_not_used_dc_star}}\hfil%
    \subfloat[Case~2.2.1]{\includegraphics{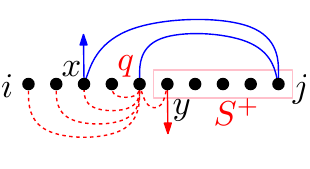}\label{fig:large_red_ij_not_used_dc_not_star_not_yy1_default}}\hfil%
    \subfloat[Case~2.2.1.1]{\includegraphics{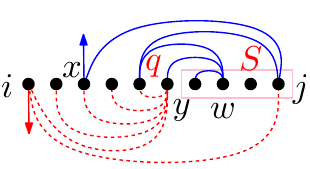}\label{fig:large_red_ij_not_used_dc_not_star_not_yy1_not_star}}\hfil%
    \caption{The case analysis in the proof of
      Proposition~\ref{prop:rec_large_red_star_ij_not_used_sp_not_star}~(Part~2/6).}
  \end{figure}

  \case{2.2} $B[y,j]$ is not a star and $\{y-1,j\}\in\EB$. We
  consider two cases.

  \case{2.2.1} $\{y-1,y\}\not\in\EB$. Embed $q$ onto $y-1$ and
  the children of $q$ onto $[y-2,i]$. This works by 1SR and
  $\{y-1,j\}\in\EB$. Embed $S^+$ recursively onto $[y,j]$. See
  \figurename~\ref{fig:large_red_ij_not_used_dc_not_star_not_yy1_default}.
  Since $B[y,j]$ is not a star by assumption and since $B[x,j]$ is
  rooted at $x$, the only possible issue is an edge-conflict. In
  that case, let $w$ such that $\treeatt{[y,j]}{y}=B[y,w]$. The root
  $b^*$ of $B[y,w]$ is in edge-conflict with $r$. Due to the edge
  $\{x,j\}$ that is used by the blue embedding, the edge-conflict can be
  caused only by an edge from $b^*$ to $y-1$ (which is where we embedded
  $q$). By 1SR, $b^*=w$. Since $\{y-1,y\}\not\in\EB$ and $B[y,j]$ is not
  a star, we have $y+1\leq w\leq j-1$. We consider two cases.

  \begin{figure}[b]
    \centering\hfil%
    \subfloat[Case~2.2.1.2]{\includegraphics{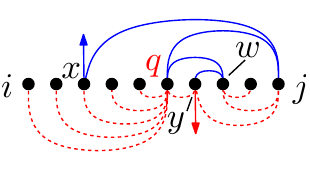}\label{fig:large_red_ij_not_used_dc_not_star_not_yy1_star_1}}\hfil%
    \subfloat[Case~2.2.1.2]{\includegraphics{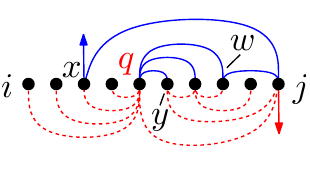}\label{fig:large_red_ij_not_used_dc_not_star_not_yy1_star_2}}\hfil%
    \subfloat[Case~2.2.2.1]{\includegraphics{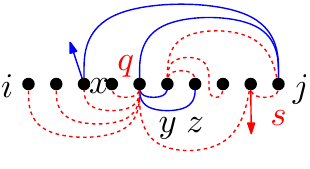}\label{fig:large_red_ij_not_used_dc_not_star_yy1_not_star_1}}\hfil%
    \subfloat[Case~2.2.2.1]{\includegraphics{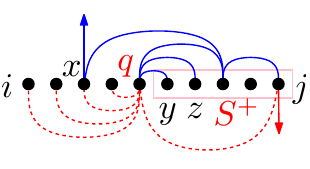}\label{fig:large_red_ij_not_used_dc_not_star_yy1_not_star_2}}\hfil%
    \caption{The case analysis in the proof of
      Proposition~\ref{prop:rec_large_red_star_ij_not_used_sp_not_star}~(Part~3/6).}
  \end{figure}

  \case{2.2.1.1} $S$ is not a star. Embed $r$ onto $i$, $q$ onto $y$,
  the children of $q$ onto $[y-1,i+1]$, and $S$ recursively onto
  $[j,y+1]$. See
  \figurename~\ref{fig:large_red_ij_not_used_dc_not_star_not_yy1_not_star}.
  Since $B[y,w]$ is a star and $y<w$, by 1SR $y$ is isolated on
  $B[i,y]$: hence the embedding of $Q$ works and the edges $\{r,q\}$ and
  $\{r,s\}$ incident to $r$ are not used by the blue embedding. Hence,
  the only possible issues are caused by recursively embedding $S$.
  Suppose there is a conflict for embedding $S$ onto $[j,y+1]$.
  Then $\treeatt{[j,y+1]}{j}$ is a central-star. Since
  $\{x,j\}\in\EB$ and $x<y$ it follows that the root (and hence the
  center) of $\treeatt{[j,y+1]}{j}$ is at $j$. But since $B[y-1,w]$ is a
  subtree of $j$ on more than one vertex, this violates LSFR at $j$.
  Hence, there is no conflict for embedding $S$, which concludes
  this case.

  \case{2.2.1.2} $S$ is a star. Since $S^+$ is not a star, $S$ is a
  dangling star. Let $s'$ be the child of $s$. We distinguish two cases.
  If $w=y+1$, then embed $r$ onto $y$, $q$ onto $y-1$, the children of
  $q$ onto $[y-2,i]$, $s$ onto $j$, $s'$ onto $w=y+1$, and the children
  of $s'$ onto $[y+2,j-1]$. See
  \figurename~\ref{fig:large_red_ij_not_used_dc_not_star_not_yy1_star_1}.
  Since $y$ is not the root of $B[i,x]$, $y$ has no edges to the outside
  of the interval and hence it is safe to embed $r$ there. Since
  $B[y-1,w]$ is a star centered at $w$, $y$ is adjacent only to $w$ in
  the blue embedding and hence $\{r,s\}=\{y,j\}\not\in\EB$ and
  $\{r,q\}=\{y-1,y\}\not\in\EB$. By 1SR, $w$ is isolated in $B[w,j]$ and
  hence we can embed $S$ as described and similarly $y$ is isolated in
  $B[i,y]$ and hence we can embed $Q$ as described.

  If $w\geq y+2$, then flip the blue embedding at $[y-1,w]$. This
  places the center of the star $B[y-1,w]$ at $y-1$ and its root at $w$.
  Since $w\geq y+2$, the vertices $y$ and $y+1$ are adjacent only to
  $y-1$ now. Embed $r$ onto $j$ (which is not the root of $B[x,j]$), $s$
  onto $y$, $s'$ onto $y+1$, the children of $s'$ onto $[y+2,j-1]$, $q$
  onto $y-1$ (the edge $\{y-1,j\}$ is no longer used after flipping),
  and the children of $q$ onto $[y-2,i]$. See
  \figurename~\ref{fig:large_red_ij_not_used_dc_not_star_not_yy1_star_2}.
  After flipping, $y-1$ is isolated in $B[i,y-1]$ and hence the
  embedding of $Q$ works.

  \case{2.2.2} $\{y-1,y\}\in\EB$.

  \case{2.2.2.1} $B[y-1,j]$ is not a star centered at $y-1$. Let $z$ be
  such that $\treeatt{[y-1,j-1]}{y-1}=B[y-1,z]$. Since $\{y-1,y\}\in\EB$
  and $\{y-1,j\}\in\EB$ (Case~2.2), $B[y-1,z]$ is a central-star. Since
  $B[y,j]$ is not a star (Case~2.2), $y\leq z\leq j-2$. By LSFR at $j$
  we know that $\{j-1,j\}\not\in\EB$. Since $z\leq j-2$ we know that
  $\{y-1,j-1\}\not\in\EB$.

  Suppose first that $S$ is a dangling star and let $s'$ be the child of
  $s$ in $S$. Then embed $r$ onto $j-1$, $q$ onto $y-1$, the children of
  $q$ onto $[y-2,i]$, and $s$ onto $j$. Flip $B[y-1,z]$ into the lower
  halfplane. Embed $s'$ onto $y$, drawing the edge $\{s,s'\}$ in the
  upper halfplane. Embed the children of $s'$ onto $[y+1,j-2]$. The
  edges between $s'$ and its children embedded at $[z+1,j-2]$ are drawn
  as biarcs. See
  \figurename~\ref{fig:large_red_ij_not_used_dc_not_star_yy1_not_star_1}.

  Otherwise, $S$ is not a star since $S^+$ is not a star. Flip
  $B[z+1,j]$. Embed $q$ onto $y-1$ and the children of $q$ onto
  $[y-2,i]$. Embed $S^+$ recursively onto $[j,y]$. See
  \figurename~\ref{fig:large_red_ij_not_used_dc_not_star_yy1_not_star_2}.
  Since $y$ is isolated in $B[y,j]$, $B[y,j]$ is not a star. If there is
  a conflict for the embedding of $S^+$ onto $[j,y]$, then
  $\treeatt{[y,j]}{j}$ must be a central-star rooted and centered at
  $z+1$. But this violates the LSFR at $j$ before flipping. Hence, this
  embedding works.


  \case{2.2.2.2} $B[y-1,j]$ is a star centered at $y-1$. Let
  $B':=B[y-1,j]$. We reembed $B[x,j]$ as follows. Use the normal
  embedding algorithm for blue trees to embed $B[x,j]$ onto $[j,x]$
  (placing the root at $j$), but embed $B'$ as the closest subtree,
  i.e., embed $B'$ at $[y-2,j-1]$. This embeds the center of $B'$ at
  $j-1$.

  Embed $r$ onto $i$, $q$ onto $y$, and the children of $q$ onto
  $[y-1,i+1]$. This works so far: $y$ is adjacent only to $j-1$ in the
  blue embedding. If $S$ is a star (it is not rooted at a center) then
  embed $s$ onto $y+1$, the child $s'$ of $s$ onto $j$, and the children
  of $s'$ onto $[j-1,y+2]$. See
  \figurename~\ref{fig:large_red_ij_not_used_dc_not_star_yy1_star_1}.
  This works because $j$ is isolated on $B[y-1,j]$. If $S$ is not a
  star, embed $S$ recursively onto $[j,y+1]$. See
  \figurename~\ref{fig:large_red_ij_not_used_dc_not_star_yy1_star_2}.
  Since $j$ is isolated in $B[y+1,j]$, $B[y+1,j]$ is not a star and
  there is no conflict for embedding $S$ onto $[j,y+1]$.

  \begin{figure}
    \centering\hfil%
    \subfloat[Case~2.2.2.2]{\includegraphics{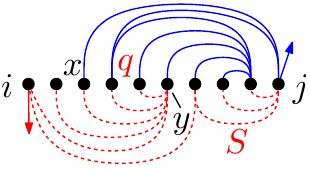}\label{fig:large_red_ij_not_used_dc_not_star_yy1_star_1}}\hfil%
    \subfloat[Case~2.2.2.2]{\includegraphics{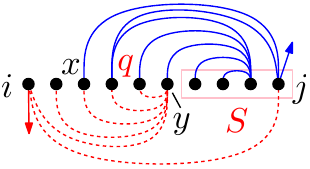}\label{fig:large_red_ij_not_used_dc_not_star_yy1_star_2}}\hfil%
    \subfloat[Case~2.3.1.1]{\includegraphics{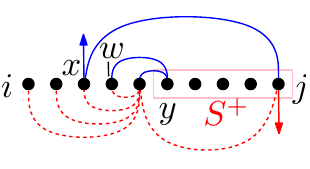}\label{fig:large_red_ij_not_used_dc_not_isolated_1}}\hfil%
    \caption{The case analysis in the proof of
      Proposition~\ref{prop:rec_large_red_star_ij_not_used_sp_not_star}~(Part~4/6).}
  \end{figure}

  \case{2.3} $B[y,j]$ is not a star and $y-1$ is not isolated in
  $B[i,y-1]$. We distinguish three cases.

  \case{2.3.1} $B[i,x]$ is not a star and $y$ is not isolated in
  $B[i,y]$. By 1SR at $y-1$, the edge $\{y-1,y\}$ is not used. Let $w$
  be the rightmost neighbor of $y$ on $[i,y-1]$. Then $x\leq w$ and
  $B[w,y]$ is a tree on at least three vertices.

  If $B[w,y]$ is a central-star, then its root and center is at $w$. Use
  the leaf-isolation-shuffle on $B[w,y]$ to place a leaf at $y-1$ and
  its parent at $y$. By Proposition~\ref{prop:leafshuffle}, if $B[w,y]$
  is not a central-star, this places the root of $B[w,y]$ at $w$ and
  preserves the 1SR at $y$. If $B[w,y]$ is a central-star, then let $z$
  be the largest index such that $B[w,z]$ is a central-star. Since
  $B[i,x]$ is not a star, we have $x<w$ and since $\{i,x\}\in\EB$ we
  have $z\leq j-1$. The leaf-isolation-shuffle places the root of
  $B[w,z]$ at $y$. Note that all vertices in $[y+1,z]$ are now also
  adjacent to $y$.

  \case{2.3.1.1} $B[w,y]$ is not a central-star or $B[w,y]$ is a
  central-star but $z<j-1$. In the latter case, the edge $\{w,j\}$ is
  not used, as this would imply that $\{w,z+1\}$ is used by LSFR at $w$,
  contradicting the choice of $z$. Embed $q$ onto $y-1$ and the children
  of $q$ onto $[y-2,i]$. This works because $y-1$ is adjacent only to
  $y$ in $B$. Embed $S^+$ recursively onto $[j,y]$. See
  \figurename~\ref{fig:large_red_ij_not_used_dc_not_isolated_1}. Since
  the root of $B[x,j]$ is at $x$, any edge-conflicts must be caused by
  edges to $y-1$ (which is where we embedded $q$). However, only $y$ is
  adjacent to $y-1$ in $B$ and $\treeatt{[y,j]}{j}\neq B[y,j]$ by 1SR on
  $y$ or by $z<j-1$. Hence, there is no conflict for embedding $S^+$
  onto $[j,y]$.

  \case{2.3.1.2} $B[w,y]$ is a central-star with $z=j-1$. Then $B[w,j]$ is
  a dangling star centered at $w$. Since $w<y$, we can proceed as in
  Case~2.2.2.2 (the argument still works for the larger star we have in
  this case).

  \case{2.3.2} $B[i,x]$ is not a star and $y$ is isolated in $B[i,y]$.
  Since $\{x,j\}\in\EB$, it follows that $y$ is not isolated in
  $B[y,j]$. We first try the following. Embed $r$ onto $i$, $q$ onto
  $y$, the children of $q$ onto $[y-1,i+1]$, and $S$ recursively onto
  $[j,y+1]$. See
  \figurename~\ref{fig:large_red_ij_not_used_dc_not_isolated_2}. The
  embedding of $Q$ works because $y$ is isolated in $B[i,y]$. The
  embedding of $S$ fails if (1) $S$ is a star. In addition, the
  embedding could fail if $B[y+1,j]$ is a star or if there is a conflict
  for embedding $S$ onto $[j,y+1]$, in which case $\treeatt{[y+1,j]}{j}$
  is a central-star. We cover these cases with (2) $B[y+1,j]$ is a
  dangling star and (3) $[j,y+1]$ is in conflict for embedding $S$.

  \case{2.3.2.1} $S$ is a star. Since $S^+$ is not a star, $S$ is a
  dangling star centered at $s'$. Let $z$ be such that
  $\treeatt{[y,j]}{y}=B[y,z]$. Suppose first that $B[y,z]$ is not a
  central-star. Use a leaf-isolation shuffle on $B[y,z]$ to put a leaf
  at $y+1$, its parent at $y$, and the root at $z$. This works by
  Proposition~\ref{prop:leafshuffle}. Embed $r$ onto $i$, $q$ onto $y$,
  and the children of $q$ onto $[y-1,i+1]$. This works so far, since the
  leaf-isolation shuffle preserves the 1SR at $y$. Embed $s$ onto $j$,
  $s'$ onto $y+1$, and the children of $s'$ onto $[y+2,j-1]$. This works
  because $\treeat{i}\neq\treeat{j}$ and because $y+1$ is isolated in
  $B[y+1,j]$. See
  \figurename~\ref{fig:large_red_ij_not_used_dc_not_isolated_3}.

  \begin{figure}
    \centering\hfil%
    \subfloat[Case~2.3.2]{\includegraphics{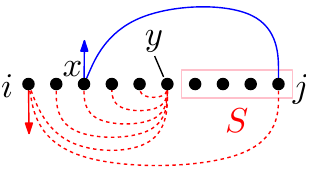}\label{fig:large_red_ij_not_used_dc_not_isolated_2}}\hfil%
    \subfloat[Case~2.3.2.1]{\includegraphics{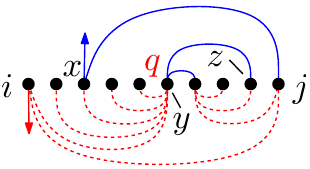}\label{fig:large_red_ij_not_used_dc_not_isolated_3}}\hfil%
    \subfloat[Case~2.3.2.1]{\includegraphics{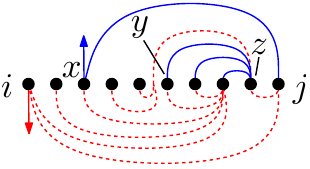}\label{fig:large_red_ij_not_used_dc_not_isolated_4}}\hfil%
    \subfloat[Case~2.3.2.2]{\includegraphics{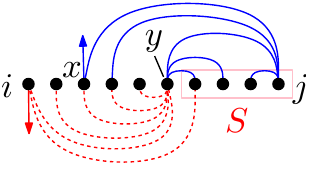}\label{fig:large_red_ij_not_used_dc_not_isolated_5}}\hfil%
    \caption{The case analysis in the proof of
      Proposition~\ref{prop:rec_large_red_star_ij_not_used_sp_not_star}~(Part~5/6).}
  \end{figure}

  Otherwise, $B[y,z]$ is a central-star. Then it must be rooted and
  centered at $z$. By the assumption of Case~2.3, we must have $z<j$.
  Embed $r$ onto $i$, $s$ onto $j$, and $s'$ onto $z$. This works so far
  since $\treeat{i}\neq\treeat{j}$ and
  $\treeatt{[y,j]}{y}\neq\treeatt{[y,j]}{j}$. Embed a child of $s'$ on
  every vertex in $[z+1,j]$. Exactly $|[y,z-1]|$ children of $s'$ remain
  to be embedded. Since $y-1$ is not isolated in $B[i,y-1]$ (assumption
  of Case~2.3), $y-1$ and $z$ must have some common parent $p$ with
  $x\leq p\leq y-2$. By LSFR at $p$, we have
  $|\tr(y-1)|\geq|\tr(z)|>|[y,z-1]|$, and so $\tr(y-1)$ is large enough to
  accomodate all remaining children of $s'$. This is true even if $x=p$
  (for which we modified the order of the subtrees), since $z$ is not
  the last subtree. Thus, embed the remaining children of $s'$ onto
  $[y-1,y-|[y,z-1]|]$. Embed $q$ onto the leaf $z-1$ of $z$ and the
  children of $q$ onto the remainder. See
  \figurename~\ref{fig:large_red_ij_not_used_dc_not_isolated_4}.

  \case{2.3.2.2} $B[y+1,j]$ is a dangling star. Since $y$ is not
  isolated in $B[y,j]$, we must have $\{y,j\}\in\EB$. Simultaneously
  shift $B[y+1,j-1]$ to $[y,j-2]$ and $y$ to $j-1$. Embed $r$ onto $i$,
  $q$ onto $y$, the children of $q$ onto $[y-1,i+1]$, and $S$
  recursively onto $[y+1,j]$. Since $y+1$ is isolated in $B[y+1,j]$, the
  recursive embedding of $S$ always works. See
  \figurename~\ref{fig:large_red_ij_not_used_dc_not_isolated_5}.

  \case{2.3.2.3} $[j,y+1]$ is in conflict for embedding $S$. Let
  $w$ be such that $\treeatt{[y+1,j]}{j}=B[w,j]$. Then $B[w,j]$ is a
  central-star rooted at $j$. If $w=y+1$, then $y$ must be connected to
  $j$ and so $B[y,j]$ is a star. This contradicts our assumption of
  Case~2.3 and hence $w\geq y+2$. The root $j$ of $B[w,j]$ cannot be in
  edge-conflict with $s$ because $\treeat{i}\neq\treeat{j}$. Thus, it is
  in degree-conflict and we have $\deg_B(j)+\deg_S(s)\geq|[y,j]|$.
  Recall from the start of Case~2 that the root of $B[x,j]$ has degree
  at least $|S^+|$. Since $B[w,j]$ is not an isolated vertex, all other
  subtrees of $x$ must have size at least two. Thus,
  $|[x+1,w-1]|\geq2(|S^+|-1)\geq|S|$. Embed $r$ onto $i$. Use a
  blue-star embedding to embed $S$ onto $[j,x+1]$. Note that
  $\treeatt{[x+1,j]}{j}=B[w,j]$ and that \ref{gg:dc} is satisfied by the
  discussion above. Embed $q$ onto $y$ and the children of $q$ onto
  $[y-1,i+1]$ to complete the embedding.

  \case{2.3.3} $B[x,j]$ is a star. Recall that $B[x,j]$ is rooted at
  $x$. Since $\deg_B(j)\geq|S^+|\geq4>1$, $B[x,j]$ is a central-star.
  Then the rearrangement of $B[x,j]$ at the start of Case~2 did not
  change anything, and hence $B$ satisfies the invariants. We replay the
  case analysis, starting from the very start of this proof, but now we
  embed on $[i',j']:=[j,i]$ (i.e. we embed from the other side). Note
  that $[i',j']$ may not satisfy the peace invariant, but it satisfies
  the other invariants. Consider the initial embedding in the proof,
  which performs a red-star embedding of $Q$ from $j'$ and then
  embed $S^+$ on the left of $[i',j']$. The embedding of $S^+$ always
  works: $B[x,j]$ is a star of size larger than $|S^+|$ which appears on
  the left of the interval $[i',j']$, and hence the first $|S^+|$
  elements of $B[i',j']$ form an independent set. Thus, if the initial
  embedding fails, we must land in Case~2. Since we have not yet used
  the peace invariant in Case~2 so far, we can simply execute the case
  analysis of Case~2 until we get an embedding or we arrive at this case
  (Case~2.3.3).

  It remains to consider the event that the embedding procedure also
  reaches this case (Case~2.3.3) for embedding $R$ onto $[j',i']$. Refer
  to Figure~\ref{fig:large_red_ij_not_used_dc_two_stars}. Then
  $\treeat{i}$ and $\treeat{j}$ are both central-stars of size larger
  than $|S^+|$. Flip $\treeat{i}$ if necessary to put its root at $i$
  and flip $\treeat{j}$ if necessary to put its root at $j$. Let $x'$ be
  such that $B[i,x']=\treeat{i}$. By the peace invariant for embedding
  $R$ on $[i,j]$, the root of $\treeat{i}$ at $i$ is not in
  edge-conflict with $r$. Embed $r$ at $i$ and $q$ at $x$, drawing the
  edge $\{r,q\}$ as a biarc that is in the upper halfplane near $r$ and
  crosses the spine between $x'$ and $x'+1$. Embed a child of $q$ on
  every vertex in $[x'+1,x-1]\cup[x+1,j-1]$. Using that
  $|\treeat{i}|\geq|S^+|+1$, this works because
  $\deg_Q(q)=|I|-|S^+|-1\geq|I|-|\treeat{i}|=|[x'+1,j]|>|[x'+1,x-1]\cup[x+1,j-1]|$.
  Embed $s$ onto $j$. The remaining blue vertices in $[i+1,x']$ form an
  independent set on which we can easily embed the remaining children of
  s $q$ and $S$ explicitly.

  \begin{figure}
    \centering%
    \subfloat[Case~2.3.3]{\includegraphics{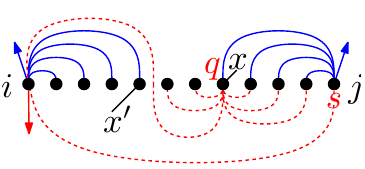}\label{fig:large_red_ij_not_used_dc_two_stars}}%
    \caption{The case analysis in the proof of
      Proposition~\ref{prop:rec_large_red_star_ij_not_used_sp_not_star}~(Part~6/6).}
  \end{figure}
\end{proof}

\begin{proposition}\label{prop:rec_large_red_star_ij_not_used_sp_star}
  If $R^-$ and $S^+$ are both stars and $\{i,j\}\not\in\EB$,
  then $R$ and $B$ admit an ordered plane packing onto $[i,j]$.
\end{proposition}
\begin{proof}
  Let $q$ be the child of $r$ in $R^-$ and let $Q=\tr(q)$. Then $Q$ is a
  star centered at $q$ and $S$ is a star centered at $s$. The case
  $|S|=1$ is handled by Lemma~\ref{lem:rec_singleton}. In the remainder
  we assume $|S|\geq 2$.  We deal with two red stars here, so we
  frequently use the red-star embedding. Since all embeddings in this
  proof are explicit (we cannot recursively embed stars, after all), we
  only perform Step~1 (Embed) of the red-star embedding for ease of
  explanation.

  Let $h$ such that $B[h,j]=\treeat{j}$. Re-embed $B[h,j]$ by putting
  its root at $j$ and embedding its subtrees according to the
  smaller-subtree-first rule (SSFR) and the 1SR. By assumption,
  $\{i,j\}\not\in\EB$ and hence these modifications do not touch
  $\treeat{i}$. Our general plan is the following. Embed $r$ at $i$.
  This works by the placement invariant. Perform a red-star
  embedding to embed $s$ onto $j$ and the children of $s$ onto the
  rightmost $\deg_S(s)$ non-neighbors of $j$ in $[i+1,j-1]$. Since
  $\{i,j\}\not\in\EB$, $j$ is not in edge-conflict with $s$ and
  hence~\ref{sgg:ec} holds. Hence, this works unless~\ref{sgg:dc} fails,
  i.e., unless (1) $\deg_S(s)+\deg_B(j)\geq |[i+1,j-1]|+1=|I|-1$. We
  embed $q$ onto the rightmost child $h'$ of $j$. This works unless $j$
  has no children, i.e., unless (2) $\deg_B(j)=0$. We finally embed the
  children of $q$ onto the remaining vertices. See
  \figurename~\ref{fig:two_red_no_ij_default}. Since the
  red-star embedding ensures that all remaining vertices are
  visible from below, this is possible unless $h'$ has an edge to a
  remaining vertex. Note that all edges of $h'$ are in $B[h',j]$, and we
  embedded $s$ onto $j$. Hence, it suffices to handle the case where the
  red-star embedding did not embed a child of $s$ onto every
  vertex of $B[h'+1,j-1]$, i.e., the case that (3)
  $\deg_S(s)\leq|[h'+1,j-1]|-1$. We deal with these remaining cases
  below. We first state a useful observation.

  \begin{figure}[b]
    \centering\hfil%
    \subfloat[Default]{\includegraphics{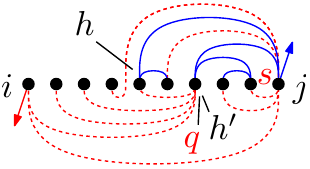}\label{fig:two_red_no_ij_default}}\hfil%
    \subfloat[Case~1]{\includegraphics{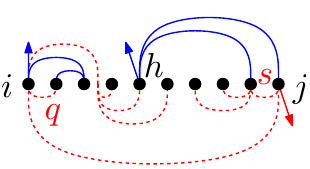}\label{fig:two_red_no_ij_dc}}\hfil%
    \subfloat[Case~2]{\includegraphics{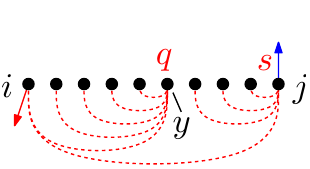}\label{fig:two_red_no_ij_deg0_default}}\\
    \subfloat[Case~2.1]{\includegraphics{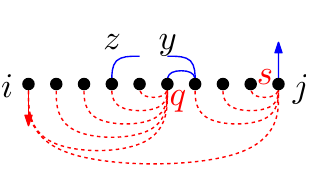}\label{fig:two_red_no_ij_deg0_y1_not_iso_1}}\hfil%
    \subfloat[Case~2.1]{\includegraphics{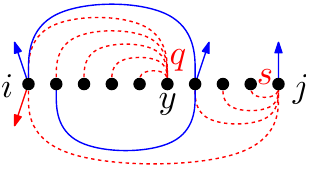}\label{fig:two_red_no_ij_deg0_y1_not_iso_2}}\hfil%
    \label{fig:two_red_no_ij_1}
    \caption{The case analysis in the proof of
      Proposition~\ref{prop:rec_large_red_star_ij_not_used_sp_star}
      (Part~1/4).}
  \end{figure}

  \begin{observation}\label{obs:two_stars_sdc}
    Let $b,b'\in B$ be the roots of two different trees of the
    forest $B$. Suppose that $\deg_B(b)+\deg_S(s)\geq |I|-1$. Then
    \begin{enumerate}[label={(P\arabic*)}]\setlength{\itemindent}{3\labelsep}
    \item\label{obs:two_stars_sdc_large} $\deg_B(b)\geq(|I|+1)/2$;
    \item\label{obs:two_stars_sdc_leaves} at least three children of $b$
      are leaves; and
    \item\label{obs:two_stars_sdc_nodc} $\deg_B(b')+\deg_S(s)\leq
      \deg_B(b')+\deg_Q(q)\leq |I|-3$
    \end{enumerate}
  \end{observation}
  \begin{proof}
    Since $r$ has two subtrees and $S$ is the smaller one, we have
    $|S|\leq (|I|-1)/2$ and hence $\deg_S(s)=|S|-1\leq (|I|-3)/2$. By
    the assumption, we have $\deg_B(b)\geq|I|-1-\deg_S(s)\geq
    |I|-1-(|I|-3)/2=(|I|+1)/2$, as claimed
    in~\ref{obs:two_stars_sdc_large}. Let $\lambda$ be the number of
    leaf subtrees of $b$. The other subtrees of $b$ have size at least
    two and the total size of $\treeat{b}$ is at most $|I|-1$, since $b$
    and $b'$ are the roots of different trees. Hence,
    $1+\lambda+2(\deg_B(b)-\lambda)\leq|\treeat{b}|\leq |I|-1$, and so
    $\lambda\geq 1+2\deg_B(b)-(|I|-1)=2-|I|+2\deg_B(b)$. Then,
    by~\ref{obs:two_stars_sdc_large}, $\lambda\geq 2-|I|+(|I|+1)=3$,
    which proves~\ref{obs:two_stars_sdc_leaves}.

    Since $r$ has two subtrees and $Q$ is the larger one, we have
    $|Q|\geq (|I|-1)/2$ and hence $\deg_Q(q)=|Q|-1\geq (|I|-3)/2$. The
    first inequality of~\ref{obs:two_stars_sdc_nodc} follows from
    $|S|\leq|Q|$. Suppose towards a contradiction that the second
    inequality of~\ref{obs:two_stars_sdc_nodc} is false, that is,
    $\deg_B(b')+\deg_Q(q)\geq |I|-2$. Adding this equation to the
    assumption, we obtain $\deg_B(b)+\deg_B(b')+\deg_S(s)+\deg_Q(q)\geq
    2|I|-3$. Since $\deg_S(s)+\deg_Q(q)=|I|-3$, it follows that
    $\deg_B(b)+\deg_B(b')\geq|I|$, which contradicts $b\neq b'$.
    Claim~\ref{obs:two_stars_sdc_nodc} follows.
  \end{proof}

  \case{1} $\deg_S(s)+\deg_B(j)\geq |I|-1$. Then
  Observation~\ref{obs:two_stars_sdc} applies with $b:=j$. Re-embed
  $B[h,j]$ by placing its root at $h$ and embedding its subtrees with
  LSFR and 1SR. Embed $r$ onto $j$ and $s$ onto $j-1$. This works
  because $j$ and $j-1$ are leaves by~\ref{obs:two_stars_sdc_leaves} and
  LSFR. If necessary, flip $\treeat{i}$ to put its root at $i$. Use
  the red-star embedding to embed $q$ onto $i$ and the children
  of $q$ onto the leftmost $\deg_Q(q)$ non-neighbors of $i$ in
  $[i+1,j-2]$. \ref{sgg:ec} holds since $\{i,j\}\not\in\EB$.
  \ref{sgg:dc} holds since $|[i+1,j-2]|=|I|-3$ and
  by~\ref{obs:two_stars_sdc_nodc} with $b:=j$ and $b':=i$. Let $x$ be
  the largest index on which a child of $q$ was embedded. Then
  $|[i,x]|\geq 1+\deg_Q(q)$. Since $\deg_B(j)\geq(|I|+1)/2$
  by~\ref{obs:two_stars_sdc_large} we have $|[h,j]|\geq(|I|+3)/2$. Then
  $|[i,x]|+|[h,j]|\geq 1+(|I|-3)/2+(|I|+3)/2=|I|+1$. It follows that
  $x\geq h$, and so the red-star embedding embedded a child of
  $q$ onto $h$. Since this is the only vertex in $B$ adjacent to
  $j-1$ (which is where we embedded $s$), we can embed the children of
  $s$ on the remainder. See \figurename~\ref{fig:two_red_no_ij_dc}.

  \case{2} $\deg_B(j)=0$. Let $y$ such that $|[i,y]|=1+|Q|$. If $y$ is
  isolated in $B[i,y]$ then embed $r$ onto $i$, $q$ onto $y$, the
  children of $q$ onto $[y-1,i+1]$, $s$ onto $j$, and the children of
  $s$ onto $[j-1,y+1]$. See
  \figurename~\ref{fig:two_red_no_ij_deg0_default}. This works due to
  the placement invariant and the fact that $y$ is isolated in $B[i,y]$
  and $j$ is isolated in $B$. Otherwise, $y$ is not isolated in
  $B[i,y]$. We distinguish two cases.

  \case{2.1} $y+1$ is not isolated in $B[i,y+1]$. Let $z$ be the
  rightmost neighbor of $y+1$ in $B[i,y+1]$. We have $i\leq z\leq y$. If
  $i<z$, then perform a leaf-isolation-shuffle on $B[z,y+1]$ to put a
  leaf at $y$ and its parent at $y+1$. Embed $r$ onto $i$. Since $i<z$,
  the blue vertex at $i$ was not changed and hence this works by the
  placement invariant. Embed $q$ onto $y$ and the children of $q$ onto
  $[i+1,y-1]$. This works since $y$ is adjacent only to $y+1$ in
  $B$. Finally, embed $s$ onto $j$ and the children of $s$ onto
  $[j-1,y+1]$. This works because $j$ is isolated in $B$. See
  \figurename~\ref{fig:two_red_no_ij_deg0_y1_not_iso_1}.

  Otherwise, $i=z$. Since $z$ was chosen as the rightmost vertex of
  $y+1$ in $B[i,y+1]$, we have $\deg_{B[i,y+1]}(y+1)=1$ and hence
  $\{i,y\}\in\EB$. Flip $B[i,y+1]$. After flipping,
  $\{i,i+1\}\not\in\EB$. Embed $r$ onto $i+1$ and and $q$ onto $i$. Flip
  $B[i+1,y+1]$ into the lower halfplane and embed the children of $q$
  onto $[i+2,y]$. This works because after flipping, $i$ is adjacent
  only to $y+1$ in $B[i,y+1]$. Finally, embed $s$ onto $j$ and the
  children of $s$ onto $[j-1,y+1]$. This works because $j$ is isolated
  in $B$. See
  \figurename~\ref{fig:two_red_no_ij_deg0_y1_not_iso_2}.

  \case{2.2} $y+1$ is isolated in $B[i,y+1]$. In other words, all
  (possibly zero) edges incident to $y+1$ leave $y+1$ to the right. We
  distinguish two cases.

  \case{2.2.1} $\treeat{i}$ is a central-star. If
  $\treeat{i}=\treeat{y}$ then use Lemma~\ref{lem:rec_large_blue_star}
  to compute an ordered plane packing. Otherwise
  $\treeat{i}\neq\treeat{y}$. Flip $\treeat{i}$ if necessary to put
  its root at $i$.

  If $|\treeat{i}|=1$ then flip $\treeat{y}$ if necessary to put the
  root away from $y$ (recall that $y$ is not isolated in $B[i,y]$).
  Embed $r$ onto $y$, $q$ onto $i$, the children of $q$ onto
  $[i+1,y-1]$, $s$ onto $j$ and the children of $s$ onto $[j-1,y+1]$.
  See \figurename~\ref{fig:two_red_no_ij_deg0_y1_iso_star_1}. This works
  because $i$ and $j$ are both isolated in $B$.

  \begin{figure}
    \centering\hfil%
    \subfloat[Case~2.2.1]{\includegraphics{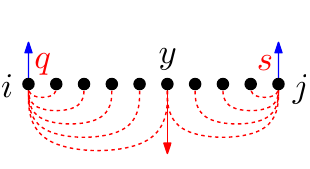}\label{fig:two_red_no_ij_deg0_y1_iso_star_1}}\hfil%
    \subfloat[Case~2.2.1]{\includegraphics{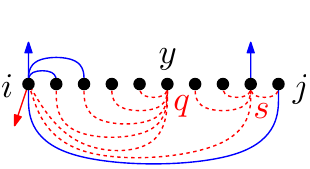}\label{fig:two_red_no_ij_deg0_y1_iso_star_2}}\hfil%
    \label{fig:two_red_no_ij_2}
    \caption{The case analysis in the proof of
      Proposition~\ref{prop:rec_large_red_star_ij_not_used_sp_star}
      (Part~2/4).}
  \end{figure}

  If $|\treeat{i}|\geq 2$, then we change the blue embedding as follows.
  Simultaneously shift $B[i+2,j]$ to $[i+1,j-1]$ and $i+1$ to $j$. The
  new edge $\{i,j\}$ is drawn in the lower halfplane. Afterwards, $y$ is
  isolated in $B[i,y]$ and $j-1$ is isolated in $B$. Embed $r$ onto
  $i$. By the peace invariant, $i$ is not in edge-conflict with
  $r$. Embed $q$ onto $y$ and the children of $q$ onto $[y-1,i+1]$.
  Embed $s$ onto $j-1$ and the children of $s$ onto $j$ and $[j-2,y+1]$.
  See \figurename~\ref{fig:two_red_no_ij_deg0_y1_iso_star_2}.

  \begin{figure}
    \centering\hfil%
    \subfloat[Case~2.2.2]{\includegraphics{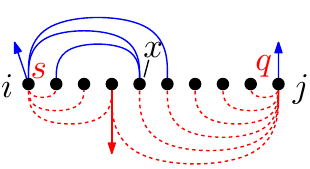}\label{fig:two_red_no_ij_deg0_y1_iso_no_star_1}}\hfil%
    \subfloat[Case~2.2.2]{\includegraphics{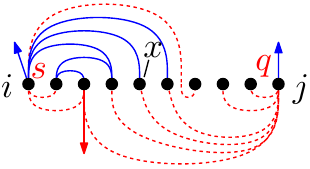}\label{fig:two_red_no_ij_deg0_y1_iso_no_star_2}}\hfil%
    \subfloat[Case~2.2.2]{\includegraphics{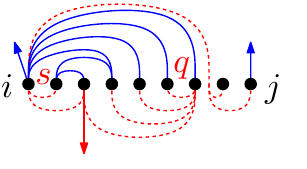}\label{fig:two_red_no_ij_deg0_y1_iso_no_star_3}}\hfil%
    \subfloat[Case~2.2.2]{\includegraphics{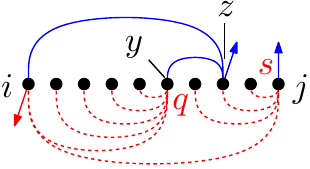}\label{fig:two_red_no_ij_deg0_y1_iso_no_star_4}}\hfil%
    \label{fig:two_red_no_ij_3}
    \caption{The case analysis in the proof of
      Proposition~\ref{prop:rec_large_red_star_ij_not_used_sp_star}
      (Part~3/4).}
  \end{figure}

  \case{2.2.2} $\treeat{i}$ is not a central-star. Then
  $|\treeat{i}|\geq 3$. Flip $\treeat{i}$ if necessary to put its root
  at $i$. Let $z$ such that $B[i,z]=\treeat{i}$ and let $x$ be the
  leftmost neighbor of $i$. Then $i<i+2\leq x\leq z$.

  If $\deg_S(s)\leq |[i+1,x-2]$ then embed $r$ onto $i+1+\deg_S(s)$, $s$
  onto $i$, the children of $s$ onto $[i+1,i+\deg_S(s)]$, $q$ onto $j$,
  and the children of $q$ onto $[j-1,i+2+\deg_S(s)]$. See
  \figurename~\ref{fig:two_red_no_ij_deg0_y1_iso_no_star_1}.

  Otherwise, $\deg_S(s)\geq |[i+1,x-2]|+1$. Embed $r$ onto $x-1$ and $s$
  onto $i$. Embed children of $s$ onto $[i+1,x-2]$. Use the
  red-star embedding to embed the remaining children of $s$ onto
  the $\deg_S(s)-|[i+1,x-2]|$ leftmost non-neighbors of $i$ in
  $[x+1,j-1]$. If~\ref{sgg:dc} is not violated, we complete the
  embedding by placing $q$ at $j$ and embedding the children of $q$ on
  the remainder. See
  \figurename~\ref{fig:two_red_no_ij_deg0_y1_iso_no_star_2}.
  If~\ref{sgg:dc} is violated, then
  $\deg_S(s)-|[i+1,x-2]|+\deg_{B[x+1,j]}(i)>|[x+1,j-1]|$. Equivalently,
  $\deg_S(s)+\deg_B(i)\geq |I|-2$. It follows that
  $\deg_B(i)\geq|I|-2-\deg_S(s)\geq|I|-2-(|I|-3)/2=(|I|-1)/2$. Since
  $|Q|\geq|S|\geq2$ we have $|I|\geq 5$. Hence $\deg_B(i)\geq
  (5-1)/2=2$. Instead of performing the red-star embedding on
  $[x+1,j-1]$, we now perform it on $[x+1,j]$. If~\ref{sgg:dc} is not
  violated, then since our first red-star embedding failed, the
  remaining vertices are exactly the neighbors of $i$ in $B$, which
  form an independent set. Complete the embedding by placing $q$ at the
  rightmost neighbor of $i$ (which is not adjacent to $r$) and the
  children of $q$ on the remainder. See
  \figurename~\ref{fig:two_red_no_ij_deg0_y1_iso_no_star_3}.

  It remains to consider the case where~\ref{sgg:dc} is again violated.
  In this case we have $\deg_S(s)+\deg_B(i)\geq |I|-1$. Due to the
  degree-conflict and the fact that $\deg_S(s)=|[y+2,j]|$ we have $z\geq
  y+2$. Flip $B[i,z]$ to put its root at $z$.
  Observation~\ref{obs:two_stars_sdc} applies with $b:=z$. By LSFR
  and~\ref{obs:two_stars_sdc_leaves}, $i$ is a leaf of $B[i,z]$. We want
  to apply Observation~\ref{obs:unary_2dc_partition} on $B[i+1,z]$. We
  first argue that the preconditions are satisfied. Let
  $n:=|B[i+1,z]|-1$ and $t:=\deg_B(z)-1$. Then $n\leq |[i+1,j-1]|-1\leq
  |I|-3$ and by~\ref{obs:two_stars_sdc_large} $t\geq
  (|I|+1)/2-1=(|I|-3)/2+1\geq n/2+1$, as required. Apply
  Observation~\ref{obs:unary_2dc_partition} with $k:=|[i+1,y-1]|$ and
  rearrange $B[i+1,z]$ to put the corresponding subtrees at $[i+1,y-1]$.
  Afterwards, all edges adjacent to $y$ leave $y$ to the right. Embed
  $r$ onto $i$, $q$ onto $y$, the children of $q$ onto $[y-1,i+1]$, $s$
  onto $j$, and the children of $s$ onto $[j-1,y+1]$. See
  \figurename~\ref{fig:two_red_no_ij_deg0_y1_iso_no_star_4}. This works
  since $y$ is isolated in $B[i,y]$ and $j$ is isolated in $B$.

  \case{3} $\deg_S(s)\leq|[h'+1,j-1]|-1$. Recall that we re-embedded
  $B[h,j]$ by placing the root at $j$ and embedding the subtrees of $j$
  according to the SSFR and 1SR. We defined $h'$ as the rightmost child
  of $j$. We have $|\tr(h')|=|[h',j-1]|\geq 2+\deg_S(s)\geq 3$. Hence, all
  subtrees of $j$ have size at least $1+|S|\geq 3$. We distinguish two
  cases.

  \case{3.1} $\deg_B(j)\geq 2$.

  \case{3.1.1} All subtrees of $j$ are central-stars. Then we flip
  $B[h,j]$, placing its root at $h$. Embed $r$ onto $i$, $s$ onto $j$,
  and the children of $s$ onto the rightmost $\deg_S(s)$ non-neighbors
  of $j$ in $[h+1,j-1]$. Each edge is drawn with a biarc that is in the
  upper halfplane close to $j$. This works because $\deg_B(h)\geq2$ (by
  our assumption and after flipping $B[h,j]$) and since all subtrees of
  $h$ have size at least $1+|S|$. Since $j-1$ is adjacent only to $j$
  (which is where we embedded $s$), we can safely place $q$ on $j-1$ and
  the children of $q$ on the remainder. See
  \figurename~\ref{fig:two_red_no_ij_smalls_largej_1}.

  \case{3.1.2} Some subtree $Z$ of $j$ is not a central-star. Re-embed
  $B[h,j]$, putting the root at $j$ and embedding the subtrees of $j$ in
  any order that places $Z$ leftmost. Let $x$ and $y$ with $x<y$ such
  that $Z=B[x,y]$. By Proposition~\ref{prop:leafshuffle} and since $Z$
  is not a central-star, we can use the leaf-isolation shuffle to place
  a leaf of $B[x,y]$ at $y-1$, its parent at $y$, and the root at $x$.
  Embed $r$ onto $i$. This works because of the placement invariant.
  Embed $q$ onto $y-1$ and $s$ onto $j$. This works since $y-1$ is
  incident only to $y$ in $B$ and $\{i,j\}\not\in\EB$. Embed a
  child of $s$ onto $y$, drawing the edge in the upper halfplane. This
  works because $\deg_B(j)\geq 2$ and $Z$ is the leftmost subtree of
  $j$. Embed the remaining children of $s$ onto the rightmost vertices
  of $[i,j-1]$. This works because all subtrees of $j$ have size at
  least $1+|S|$. Finally, note that we already embedded a vertex on the
  only blue vertex incident to $y-1$, and hence we can embed the
  children of $q$ onto the remainder. See
  \figurename~\ref{fig:two_red_no_ij_smalls_largej_2}.

  \begin{figure}
    \centering\hfil%
    \subfloat[Case~3.1.1]{\includegraphics{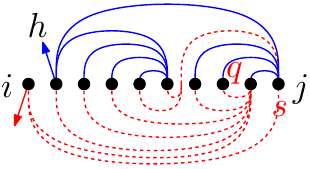}\label{fig:two_red_no_ij_smalls_largej_1}}\hfil%
    \subfloat[Case~3.1.2]{\includegraphics{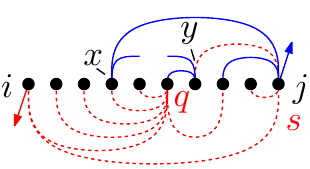}\label{fig:two_red_no_ij_smalls_largej_2}}\hfil%
    \subfloat[Case~3.2]{\includegraphics{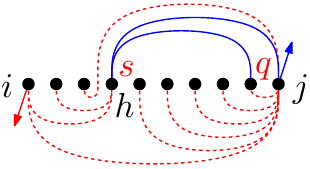}\label{fig:two_red_no_ij_smalls_smallj_1}}\hfil%
    \subfloat[Case~3.2]{\includegraphics{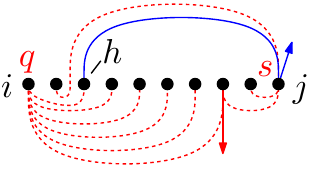}\label{fig:two_red_no_ij_smalls_smallj_2}}\hfil%
    \label{fig:two_red_no_ij_4}
    \caption{The case analysis in the proof of
      Proposition~\ref{prop:rec_large_red_star_ij_not_used_sp_star}
      (Part~4/4).}
  \end{figure}

  \case{3.2} $\deg_B(j)=1$. We first try the following. Embed $r$ onto
  $i$. Use the red-star embedding to embed $q$ onto $j$ and the
  children of $q$ onto the rightmost $\deg_Q(q)$ non-neighbors of $j$ in
  $[i+1,j-1]$. \ref{sgg:ec} holds since $\{i,j\}\not\in\EB$. Since
  $|S|\geq 2$ we have $\deg_Q(q)\leq|I|-4$ and hence
  $\deg_Q(q)+\deg_B(j)\leq |I|-3<|[i+1,j-1]|$. This
  establishes~\ref{sgg:dc}. Embed $s$ onto $h$ and the children of $s$
  onto the remainder. See
  \figurename~\ref{fig:two_red_no_ij_smalls_smallj_1}. This works unless
  some remaining vertex is adjacent to $h$.

  Since all neighbors of $h$ are in $[h+1,j]$, this implies $|Q|\leq
  |[h+1,j]|-1$, or equivalently, $|[h,j]|\geq |Q|+2$. Embed $r$ onto
  $h+\deg_Q(q)$. By $|[h,j]|\geq |Q|+2$, this is not $j$ and hence there
  is no edge-conflict. Embed $q$ onto $i$ and $s$ onto $j$. This works
  because all neighbors of $h+\deg_Q(q)$ (which is where we placed $r$)
  in the blue embedding are in $[h,j-1]$ since $\deg_B(j)=1$. Embed the
  children of $q$ onto $[h,h+\deg_Q(q)-1]$. This works because
  $\treeat{i}\neq\treeat{j}$. Embed the children of $s$ onto
  $[j-1,h+\deg_Q(q)+1]$ and $[h-1,i+1]$ (with biarcs). This works
  because the only neighbor of $j$ is at $h$. See
  \figurename~\ref{fig:two_red_no_ij_smalls_smallj_2}.
\end{proof}

\begin{proposition}\label{prop:rec_large_red_star_ij_used}
  If $R^-$ is a star and $\{i,j\}\in\EB$, then $R$ and $B$
  admit an ordered plane packing onto $[i,j]$.
\end{proposition}
\begin{proof}
  The presence of edge $\{i,j\}\in E(B)$ implies that $B$ is a tree. In
  this case, we discard the initial embedding of $B$. Instead, we embed
  $R$ using Algorithm~\ref{alg:embed_t1}, and then \emph{re-embed} $B$
  using the additional information that $R^-$ is a star.

  To simplify notation, we exchange the roles of $R$ and $B$. Refer to
  \figurename~\ref{fig:color_exchange_1}--\ref{fig:color_exchange_2}.
  That is, we assume that $B$ has been embedded using
  Algorithm~\ref{alg:embed_t1}, its root is at $j$, and it is composed
  of two trees $S_B$ and $B^-$ (corresponding to $S$ and $R^-$,
  respectively): $S_B=B[i,x]$ is a tree of size $|S_B|\geq 2$ rooted at
  $i$, and $B^-=B[x+1,j]$ is a star of size $|B^-|\geq (|I|+1)/2$
  centered at $x+1$ and rooted at $j$. We do not make any assumption
  about $S_B$, apart from that it fulfills
  invariants~\ref{inv:starconflict} and \ref{inv:bluelocal}. It remains
  to embed $R$ onto $[i,j]$. Let $S$ be a smallest subtree of $R$ and
  let $R^-=R\setminus S$.

\begin{figure}[htbp]
  \centering\hfil%
  \subfloat[]{\includegraphics{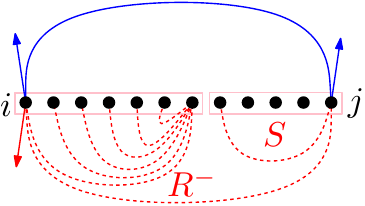}\label{fig:color_exchange_1}}\hfil%
  \subfloat[]{\includegraphics{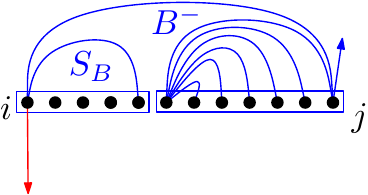}\label{fig:color_exchange_2}}\hfil%
  \subfloat[]{\includegraphics{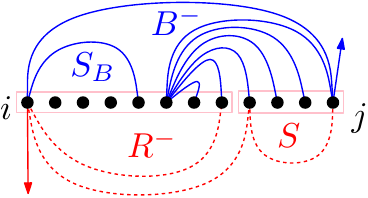}\label{fig:color_exchange_3}}\hfil%
  \caption{When the default embedding of both $R$ and $B$ contains edge
    $\{i,j\}$, we exchange the roles of $R$ and $B$ to simplify
    notation.\label{fig:color_exchange}}
\end{figure}

If $\deg_R(r)=1$, then Lemma~\ref{lem:rec_unary} completes the proof. If
$|S|=1$, then Lemma~\ref{lem:rec_singleton} completes the proof. Hence
we may assume $\deg_R(r)\geq 2$ and $|S|\geq 2$. Since $S$ is a smallest
of two or more subtrees of $r$, we have $|S|\leq (|I|-1)/2$ and
$\deg_{R^-}(r)\leq (|R^-|-1)/2$. Let $z$ be such that $|R^-|=|B[i,z]|$.
Since $|B^-|\geq(|I|-1)/2$, it follows that $z+1$ is a leaf of the
star $B^-$, and $B[z+1,j]$ consists of isolated vertices.

Our first option to embed $R$ is the following. Embed $S$ onto $[z+1,j]$
using Algorithm~\ref{alg:embed_t1}, and then embed $R^-$ recursively
onto $[i,z]$; see \figurename~\ref{fig:color_exchange_3}. This works
unless $[i,z]$ is in degree-conflict with $R^-$, or $R^-$ is a star. We
consider these two possibilities separately.

\case{1} $[i,z]$ is in degree-conflict with $R^-$, but $R^-$ is not a
star. In this case, $S_B=\treeatt{[i,z]}{i}$ is a central-star
rooted at $i$ and $\deg_{B[i,z]}(i)+\deg_{R^-}(r)\geq |R^-|$. Note that
$B[i,j-1]$ consists of two central-stars, $S_B=B[i,x]$ (rooted at $i$)
and $B[x+1,j-1]$ (rooted at $x+1$). Since $|R^-|\geq(|I|+1)/2$ and
$|S_B|\leq(|I|-1)/2$ we have $z\geq x+1$.

\case{1.1} $z\geq x+2$. Embed $S$ explicitly onto $[z+1,j]$ and $R^-$
recursively onto $[z,i]$. Since $z\geq x+2$, the blue vertex at $z$ is a
leaf that is adjacent only to $x+1$. Hence, there is no edge-conflict
for the recursive embedding of $R^-$. Since $\treeatt{[i,z]}{z}$ is a
central-star, there could be a degree-conflict. In this case we have
$\deg_{B[i,z]}(x+1)+\deg_{R^-}(r)\geq|R^-|$. Adding this equation to the
equation for the degree-conflict at $[i,z]$, we obtain
$2|R^-|-2\deg_{R^-}(r)\leq \deg_{B[i,z]}(i)+\deg_{B[i,z]}(x+1)=|R^-|-2$.
It follows that $\deg_{R^-}(r)\geq |R^-|/2+1$. Since each subtree of $r$
in $R^-$ has size at least $|S|\geq2$, we get
$\deg_{R^-}(r)\geq(1+2\deg_{R^-}(r))/2+1>\deg_{R^-}(r)+1$, a
contradiction. Hence, there is no degree-conflict and the recursive
embedding of $R^-$ always works.

\case{1.2} $z=x+1$. In this case $|R^-|=(|I|+1)/2$ and $|S|=(|I|-1)/2$.
Hence, $r$ is binary in $R$. Let $Q=\tr(q)$ be the subtree of $r$ in
$R^-$. Embed $r$ onto $i$, $Q$ explicitly onto the independent set at
$[z,i+1]$, and $S$ explicitly onto the independent set at $[z+1,j]$.
Since the blue embedding uses neither $\{i,z\}$ nor $\{i,z+1\}$, this
always works.

\case{2} $R^-$ is a star. Since $|S|\geq2$ and $S$ is a smallest subtree
of $r$, $R^-$ is a dangling star, that is, it is centered at the unique
child $q$ of $r$ in $R^-$. In this case, $R$ and $B$ have even more similarities:
their roots each have two children, and both $B^-$ and $R-$ are dangling stars.
See \figurename~\ref{fig:color_exchange_4}.

\begin{figure}[htbp]
  \centering
  \subfloat[]{\includegraphics{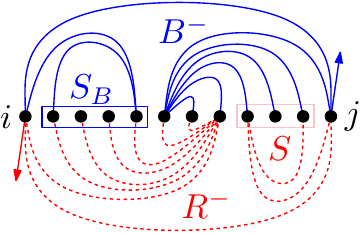}\label{fig:color_exchange_4}}\hfil%
  \subfloat[]{\includegraphics{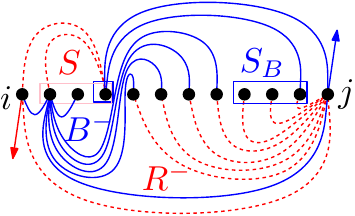}\label{fig:color_exchange_6}}\hfil%
  \caption{(a) In Case~2, the default embeddings of $B$ and $R$ share
    several edges. (b) We embed $R$ and $B$ explicitly
    (right).\label{fig:color_exchange_again} }
\end{figure}

We embed $B$ and $R$ simultaneously such that the root of $S_B$ and $S$
are mapped to the same point, and all other vertices of $S_B$ are $S$
are disjoint. This is possible since $S_B$ and $S$ each have size at
most $(|I|-1)/2$. Refer to \figurename~\ref{fig:color_exchange_6}. Embed
the star $B^-$ on $\{j\}\cup [i,j-|S_B|]\setminus \{i+|S|\}$ such
that its root (which is the root of $B$) is embedded at $j$ and its center at $i+1$.
The edges between the center $i+1$ and other vertices of $[i,i+|S|-1]$ are
semicircles \emph{below} the spine; the edge $\{i+1,j\}$ is also a semicircle below the spine;
and the edges between $i+1$ and $[i+|S|+1,j-|S_B|]$ are biarcs that start from $i+1$ below the
spine and cross the spine right after $i+|S|$.
Embed the subtree $S_B$ onto $\{i+|S|\}\cup [j-|S_B|+1,j-1]$
using Algorithm~\ref{alg:embed_t1}, with semicircles above the spine.
Embed the tree $R^-$ on $\{i\}\cup [i+|S|,j]$ such that its root (which is the root of $R$)
is at $i$, and its center is at $j$, using semicircles below the spine.
If $S$ is not a star, then finish by embedding $S$ onto $[i+|S|,i+1]$
recursively, such that the edge $\{r,s\}$ and all edges of $S$ are
semicircles \emph{above} the spine.
If $S$ is a central-star, embed $S$ explicitly onto $[i+|S|,i+1]$ above
the spine. If $S$ is a dangling star, then $|S|\geq3$. Flip the blue
embedding at $[i+1,i+|S|-1]$, placing the star-center of $B^-$ at
$i+|S|-1$. Since $|S|\geq3$, $\{i,i+1\}\not\in\EB$. Embed $s$ onto
$i+1$, the child $s'$ of $s$ onto $i+|S|$, and the children of $s'$ onto
$[i+|S|-1,i+2]$ (all above the spine).

It remains to show that we can recursively embed $S$ as described above
when $S$ is not a star.
Note that $B[i+1,i+|S|]$ consists of an isolated vertex at $i+|S|$ (the
root of $S_B$), and a star $B[i+1,i+|S|-1]$ centered and rooted at
$i+1$. Hence \ref{inv:starconflict} and \ref{inv:bluelocal} follow.
\end{proof}

Proposition~\ref{prop:rec_large_red_star_ij_not_used_sp_not_star},
Proposition~\ref{prop:rec_large_red_star_ij_not_used_sp_star}, and
Proposition~\ref{prop:rec_large_red_star_ij_used} together prove the
following.

\begin{lemma}
  \label{lem:rec_large_red_star}
  If $R^-$ is a star, then $R$ and $B$ admit an ordered plane packing
  onto $[i,j]$.
\end{lemma}

\section{Embedding the red tree: a small blue star}
\label{subsec:rec_small_blue_star}
In this section, we consider the case that $B[j-|S|+1,j]$ is a star, but
$B[i,j-|S|]$ is not a star. The size of the star is $|S|$. Due to
Lemma~\ref{lem:rec_singleton}, we may assume $|S|\geq2$. Note, however,
that $B[j-|S|+1,j]$ may be part of a larger star within $B$. Let
$B^{**}$ be the maximal star in $B$ that contains $B[j-|S|+1,j]]$. Note
that the tree $\treeat{j}$ may be larger than $B^{**}$. Clearly, we have
$|S|=|B[j-|S|+1,j]|\leq |B^{**}|$. Due to 1SR, the center and the root
of $B^{**}$ are each located at either $j$ or the leftmost vertex of
$B^{**}$, which may be outside of the interval $[j-|S|+1,i]$. We
distinguish two cases: either $|S|<|B^{**}|$
(Proposition~\ref{prop:rec_small_blue_star_larger}) or $|S|=|B^{**}|$
(Proposition~\ref{prop:rec_small_blue_star_equal_ij_not_used},
Proposition~\ref{prop:rec_small_blue_star_equal_ij_used_no_red_star},
and Proposition~\ref{prop:rec_small_red_star_ij_used}). These cases are
tackled below.

\begin{proposition}\label{prop:rec_small_blue_star_larger}
  If $B[j-|S|+1,j]$ is a star and $2\leq |S|<|B^{**}|$, then $R$ and $B$
  admit an ordered plane packing onto $[i,j]$.
\end{proposition}
\begin{proof}
  By Lemma~\ref{lem:rec_large_blue_star}, we may assume that
  $B[i,j-|S|]$ is not a star. By Lemma~\ref{lem:rec_large_red_star}, we
  may assume that $R^-$ is not a star. Recall that the center of
  $B^{**}$ is $j$, and its root is either $j$ or the leftmost vertex of
  $B^{**}$. We start by rearranging the tree $B^{**}$ such that its
  center moves to $j-|S|$; see \figurename~\ref{fig:move_center}. If
  $B^{**}$ is rooted at its center, then the root automatically moves to
  $j=|S|$, as well. Otherwise $B^{**}$ is rooted at a leaf, which is the
  leftmost vertex of $B^{**}$ and the root of the entire tree
  $\treeat{j}$ due to 1SR, and then we move the root of $B^{**}$ to $j$.
  In both cases, $B[j-|S|+1,j]$ consists of $|S|$ isolated vertices, and
  $B[i,j-|S|]$ continues to fulfill invariant~\ref{inv:bluelocal}.

\begin{figure}
  \centering\hfil%
  \subfloat[]{\includegraphics{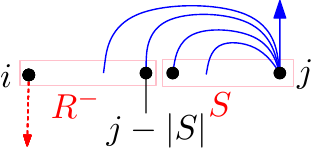}\label{fig:move_center_1}}\hfil%
  \subfloat[]{\includegraphics{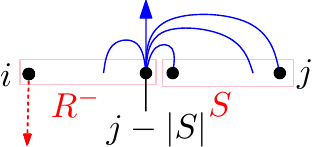}\label{fig:move_center_2}}\hfil%
  \subfloat[]{\includegraphics{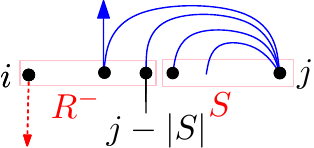}\label{fig:move_center_3}}\hfil%
  \subfloat[]{\includegraphics{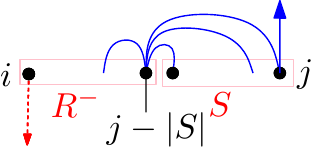}\label{fig:move_center_4}}\hfil%
  \caption{Moving the center of star $B^{**}$ from $j$ to $j-|S|$:
  when $B^{**}$ is rooted at its center (a--b), and when it is rooted at a leaf (c--d).
  \label{fig:move_center}
}
\end{figure}

\case{1} $\treeatt{[i,j-|S|]}{i}$ is not a central-star of size at least
$|R^-|-\deg_{R^-}(r)+1$. In this case, we embed $S$ explicitly onto
$[j-|S|+1,j]$ and then $R^-$ recursively onto $[i,j-|S|]$. Since
$B[j-|S|+1,j]$ consists of isolated vertices and $j-|S|+1$ is not
adjacent to $i$, the embedding of $S$ always works. We can embed $R^-$
on $[i,j-|S|]$ because it fulfills invariants \ref{inv:starconflict} and
\ref{inv:bluelocal}. Invariant~\ref{inv:bluelocal} holds by
construction. If $\treeatt{[i,j-|S|]}{i}$ is not a central-star, then
\ref{inv:starconflict} is immediate; otherwise $\treeatt{[i,j-|S|]}{i}$
is a central-star of size at most $|R^-|-\deg_{R^-}(r)$. Hence, $r$ has
no degree-conflict with $[i,j-|S|]$, and \ref{inv:starconflict} follows.

\case{2} $\treeatt{[i,j-|S|]}{i}$ is a central-star of size at least
$|R^-|-\deg_{R^-}(r)+1$. Let $B^*:=\treeatt{[i,j-|S|]}{i}$. We claim
that $\treeat{i}\neq\treeat{j}$. Suppose that $\treeat{i}=\treeat{j}$
for the sake of contradiction. Then before rearranging $B^{**}$, we had
$\{i,j\}\in\EB$. By LSFR and since $B$ is not a star, the root of $B$
was not at $i$. Again by LSFR, the root could have been at $j$ only if
$B^{**}$ is a dangling star. But then, since $|B^{**}|>|S|$,
$B[j-|S|+1,j]$ was not a star to begin with: a contradiction. The claim
follows. Since $|\treeat{j}|\geq|B^{**}|>|S|$, we have
$\treeatt{[i,j-|S|]}{i}=\treeat{i}$ and so $B^*=\treeat{i}$.

Since $R^-$ has a degree-conflict with $[i,j-|S|]$, we follow a
different strategy. We first blue-star embed $R^-$ from $\sigma=i$ with
$\varphi=(i+|B^*|+1,\ldots,i+|B^*|+d)$, and then embed $S$ on
$[j-|S|+1,j]$. The conditions for the blue-star embedding are met:
\ref{gg:ec} holds by \ref{inv:placement} for embedding $R$ onto $[i,j]$;
for \ref{gg:dc} on the one hand
$|R^-|\leq |B^*|+\deg_{R^-}(r)-1\leq |B^*|+\deg_{R^-}(r)$ and on the
other hand, by \ref{inv:starconflict}, we have $|B^*|\leq |R|-\deg_R(r)$
and so $|B^*|+\deg_{R^-}(r)\le |R|-1$. As $B^*=\treeat{i}$, the vertices
in $B\setminus(B^*\cup\varphi)$ form an interval and both \ref{gg:int}
and \ref{gg:cs} hold.

By Proposition~\ref{p:greedygrab} we are left with an interval
$[j-|S|+1,j]$ that satisfies \ref{inv:bluelocal}. Note that $\varphi$
includes the center $j-|S|$ of the star $B^{**}$, but does not include
$j$. Consequently, $B[j-|S|,j]$ consists of isolated vertices after the
blue-star embedding, and $j$ is not in edge-conflict with $s$. Hence, we
can embed $S$ explicitly onto $[j,j-|S|+1]$.
\end{proof}

\begin{proposition}\label{prop:rec_small_blue_star_equal_ij_not_used}
  If $B[j-|S|+1,j]$ is a star, $2\leq |S|=|B^{**}|$, and
  $\{i,j\}\not\in\EB$, then $R$ and $B$ admit an ordered plane
  packing onto $[i,j]$.
\end{proposition}
\begin{proof}
  By Lemma~\ref{lem:rec_large_blue_star}, we may assume that
  $B[i,j-|S|]$ is not a star. By Lemma~\ref{lem:rec_unary} and
  Lemma~\ref{lem:rec_large_red_star}, we may assume that
  $\deg_{R^-}(r)\geq1$ and $R^-$ is not a star. By
  Lemma~\ref{lem:rec_singleton}, we may assume that $|S|\geq2$. Due to
  LSFR, the center and the root of $B^{**}$ are each located at either
  $j-|S|+1$ or $j$, but $B^{**}$ may be either a central-star or a
  dangling star. We distinguish two cases.

  \case{1} $S$ is not a central-star or $B^{**}$ is a central-star. If
  necessary, flip $B^{**}$ to put its center at $j$. We will later perform a
  blue-star embedding of $S$ from $j$ with $\varphi=(j-|S|,\dots,z)$.

  Let us first check the conditions for the blue-star embedding.
  \ref{gg:ec} follows from the condition that $\{i,j\}\not\in E(B)$.
  For the other conditions, consider first the case that $B^{**}$ is a
  central-star. If the parent $p$ of $j$ is in $B$ then since
  $B^{**}$ is maximal, $p$ must have a subtree other than $B^{**}$. By
  LSFR and $\deg_S(s)\leq|S|-1$, we have that $p<z$. Hence, regardless
  of whether $p$ is in $B$, we know that
  $B\setminus(B^{**}\cup\varphi)$ forms an interval. \ref{gg:int} and
  \ref{gg:cs} follow. For \ref{gg:dc}, on the one hand we have
  $|S|=|B^{**}|<|B^{**}|+\deg_S(s)$. On the other hand we have
  $|B^{**}|+1+\deg_S(s)\leq 2|S|\leq |I|-1$.

  Otherwise, $B^{**}$ is a dangling star. Then \ref{gg:cs} is satisfied
  (with $B^+:=B^{**}$) and \ref{gg:int} is satisfied by the assumption
  of Case~1 and the choice of $\varphi$. For \ref{gg:dc}, on the one hand
  we have $|S|=|B^{**}|\leq|B^{**}|-1+\deg_S(s)$ since $\deg_S(s)\geq1$.
  On the other hand we have $|B^{**}|+\deg_S(s)< 2|S|\leq |I|-1$.

  Before performing this blue-star embedding, we modify the embedding
  of $\treeat{i}$. Since $\treeat{i}\neq\treeat{j}$ this does not affect
  the validity of the preconditions of the blue-star embedding. We
  want to ensure the following: if $[i,j-|S|]$ is in degree-conflict with
  $R^-$ after the blue-star embedding, then $\treeat{i}$ is a star
  before the blue-star embedding. We proceed as follows.

  The interval that the blue-star embedding will leave for $R^-$
  consists of $B[i,z-1]$, followed by $\deg_S(s)$ isolated vertices (all
  in edge-conflict). Suppose that this interval would be in
  degree-conflict for embedding $R^-$. Then
  $B[i,x]:=\treeatt{[i,z-1]}{i}$ is a central-star. If
  $B[i,x]=\treeat{i}$ then we do nothing. Otherwise, $B[i,x]$ is rooted
  at $i$. Let $p$ be the parent of $i$ in $B$. By 1SR we have
  $B[i,p]=\treeat{i}$. If $\deg_B(p)=1$ then $\treeat{i}$ is a dangling
  star and we do nothing. Otherwise, $\deg_B(p)\geq2$. We claim that
  then $z\geq x+2$. To prove the claim, suppose to the contrary that
  $z\leq x+1$. Since $\treeat{i}\neq\treeat{j}$ and
  $|\treeat{j}|\geq|S|$ we know that $\treeat{i}\subseteq[i,j-|S|]$. By LSFR
  and 1SR and $\deg_B(p)\geq2$, $p$ has at least one subtree $B'$
  besides $B[i,x]$ in $[i,j-|S|]$ with size
  $|B'|\geq|B[i,x]|=1+\deg_B(i)$. Since $z\leq x+1$, we know that the
  blue-star embedding consumes $p$ and all except at most one vertex of
  $B'$. Hence, $\deg_S(s)\geq 1+|B'|-1\geq 1+\deg_B(i)$. By the
  degree-conflict, we know that $\deg_{R^-}(r)+\deg_B(i)\geq|R^-|$.
  Since every subtree of $r$ in $R^-$ has size at least $|S|\geq
  1+\deg_S(s)$, we get
  \begin{align*}
    \deg_{R^-}(r)%
    &\geq |R^-|-\deg_B(i)\\
    &\geq (\deg_{R^-}(r))(1+\deg_S(s))-\deg_B(i)\\
    &\geq (\deg_{R^-}(r))(2+\deg_B(i))-\deg_B(i)\\
    &>2\deg_{R^-}(r),
  \end{align*}
  a contradiction. This proves our claim that $z\geq x+2$. Now let
  $G_1,\dots,G_k$ be the subtrees of $p$ from left to right. Note that
  $G_1=B[i,x]$. We select a parameter $t\in\{1,\ldots,k\}$ as follows.
  If $z=p$ then let $t=k$. If $z$ coincides with the root of a subtree
  of $p$, then let $t$ be such that $z$ coincides with the root of
  $G_{t+1}$. Otherwise, let $t$ be such that $z$ is contained in
  $G_t$. Then $t\geq 2$ since $z\geq x+2$. Modify the embedding of
  $\treeat{i}$ as follows. Flip the embedding of each subtree
  $G_t,\dots,G_k$ individually. Simultaneously shift each subtree
  $G_t,\dots,G_k$ one position to the right and shift $p$ to the
  position before $G_t$. In this modified embedding, $B[i,z-1]$
  satisfies LSFR and 1SR and $\treeatt{[i,z-1]}{i}$ is not a
  central-star, as intended.

  Now perform the blue-star embedding of $S$ with the parameters listed
  above. Recursively embed $R^-$ onto $[i,j-|S|]$. This works unless
  there is a conflict. There can be no edge-conflict since
  $\treeat{i}\neq\treeat{j}$ and by \ref{inv:starconflict}. If there is
  a degree-conflict, then $\treeatt{[i,z-1]}{i}$ is a central-star
  centered at $c$ and $\deg_{R^-}(r)+\deg_{B[i,z-1]}(c)\geq|R^-|$. By
  the modification of the embedding described above, we know that
  $\treeat{i}$ was a (possibly larger) star before the blue-star
  embedding. Undo the blue-star embedding. We have
  $\deg_{R^-}(r)+|\treeat{i}|-1\geq|R^-|$. We distinguish two subcases.

  \case{1.1} $|\treeat{i}|+\deg_{R^-}(r)\leq|I|-1$. Flip $B^{**}$ to put
  its center at $j-|S|+1$. If necessary, flip $\treeat{i}$ to put its
  center at $i$. First blue-star embed $R^-$ from $i$ with $\varphi$
  as the $\deg_{R^-}(r)$ leftmost vertices following $\treeat{i}$; and
  then embed $S$ onto $[j,j-|S|+1]$ using Algorithm~\ref{alg:embed_t1}.
  The conditions for the blue-star embedding are met: \ref{gg:ec}
  follows from \ref{inv:starconflict} and \ref{gg:dc} follows from
  $\deg_{R^-}(r)+|\treeat{i}|-1\geq|R^-|$ and the assumption of
  Case~1.1. \ref{gg:int} and \ref{gg:cs} hold by choice of $\varphi$. The
  blue-star embedding replaces the center of $B^{**}$ at $j-|S|+1$
  with an isolated vertex, but it does not affect $j$. Consequently,
  after the blue-star embedding $B[j-|S|+1,j]$ consists of $|S|\geq 2$
  isolated vertices, where $j$ is not in edge-conflict with $s$. Thus,
  we can embed $S$ onto $[j,j-|S|+1]$.

  \case{1.2} $|\treeat{i}|+\deg_{R^-}(r)\geq|I|$. By
  \ref{inv:starconflict}, $\treeat{i}$ is not a central-star and must
  hence be a dangling star $B[i,y]$. Since $\treeatt{[i,z-1]}{i}$ is a
  central-star, $\treeat{i}$ is centered at $i$ and $z\leq y$. Flip
  $\treeat{i}$ to place the center at $y$. Perform the original
  blue-star embedding of $S$ from $j$ again. Let us consider the
  interval $B[i,j-|S|]$ that remains for $R^-$. Since $z\leq y$,
  $B[i,j-|S|]$ is an independent set. At $i$ we have the original root
  of $\treeat{i}$, which may be in edge-conflict with $r$. Each of the
  $\deg_S(s)$ rightmost vertices of $B[i,j-|S|]$ is in edge-conflict
  with $r$. Since $|\treeat{j}|\geq|B^{**}|\geq2$ and
  $\treeat{i}\neq\treeat{j}$ we have
  $\deg_{R^-}(r)\geq|I|-|\treeat{i}|\geq2$. Every subtree of $r$ in
  $R^-$ has size at least $|S|$, and hence we can embed one subtree
  explicitly on a prefix of $[i,j-|S|]$ (which takes care of the vertex
  $i$ which is potentially in edge-conflict) and one subtree explicitly on
  a suffix of $[i,j-|S|]$ (which takes care of all $\deg_S(s)$ vertices
  which are in edge-conflict). The remaining vertices are not in
  edge-conflict, and so we can explicitly complete this partial
  embedding of $R^-$ to a complete embedding of $R^-$.

  \case{2} $S$ is a central-star and $B^{**}$ is a dangling star.
  Flip $B^{**}$ if necessary to put its root at $j$. This preserves
  1SR on $B$. We distinguish two cases.

  \case{2.1} Every vertex in $B[i+1,j-|S|]$ is a neighbor of $j$.
  Since $\treeat{i}\neq\treeat{j}$ the blue embedding is completely
  determined. Flip the blue embedding at $[j-|S|,j]$. Embed $s$ onto
  $j$ and the children of $s$ onto $[j-|S|+1,j-1]$. $B[i,j-|S|]$
  consists of an isolated vertex at $i$ (which is not in edge-conflict
  with $r$ by \ref{inv:placement}) and a central-star $B[i+1,j-|S|]$.
  Hence, we can embed $R^-$ recursively onto $[i,j-|S|]$.

  \case{2.2} Some vertex in $B[i+1,j-|S|+1]$ is not a neighbor of $j$.
  We first try the following. Use the red-star embedding to embed $s$ to
  $j$ and the children of $s$ onto the rightmost $\deg_S(s)$
  non-neighbors of $j$ in $[i+1,j]$. \ref{sgg:ec} holds due to
  $\{i,j\}\not\in\EB$. For \ref{sgg:dc} we have to show that there are
  at least $\deg_S(s)$ non-neighbors of $j$ in $[i,j-1]$. This is the
  case because $B^{**}$ already contains $|B^{**}|-2=|S|-2=\deg_S(s)-1$
  non-neighbors of $j$, and the last vertex is provided by the
  assumption of this case. Embed $R^-$ recursively onto $[i,j-|S|]$.

  This works unless there is a conflict, in which case
  $\treeatt{[i,j-|S|]}{i}$ is a central-star. As usual, this
  central-star cannot be in edge-conflict for $r$ and is hence in
  degree-conflict. Since $\treeat{i}=\treeatt{[i,j-|S|]}{i}$ both before
  and after the red-star embedding and since the
  red-star embedding either leaves $\treeat{i}$ untouched or
  replaces \emph{only} its rightmost vertex by a vertex that is isolated
  in $B[i,j-|S|]$, we know that $\treeat{i}$ was a (dangling or
  central-)star before the blue-star embedding. Undo the
  red-star embedding. By the degree-conflict, we have
  $\deg_{R^-}(r)+|\treeat{i}|-1\geq|R^-|$. We proceed analogously to
  Case~1.1 and Case~1.2.

  \case{2.2.1} $|\treeat{i}|+\deg_{R^-}(r)\leq|I|-1$. Recall that the
  center of $B^{**}$ is at $j-|S|+1$. If necessary, flip $\treeat{i}$ to
  put its center at $i$. First blue-star embed $R^-$ from $i$ with
  $\varphi$ as the $\deg_{R^-}(r)$ leftmost vertices following
  $\treeat{i}$; and then embed $S$ onto $[j,j-|S|+1]$ using
  Algorithm~\ref{alg:embed_t1}. The conditions for the blue-star
  embedding are met: \ref{gg:ec} follows from \ref{inv:starconflict} and
  \ref{gg:dc} follows from $\deg_{R^-}(r)+|\treeat{i}|-1\geq|R^-|$ and
  the assumption of Case~2.2.1. \ref{gg:int} and \ref{gg:cs} hold by
  choice of $\varphi$ and since $R^-$ is not a star. The blue-star
  embedding replaces the center of $B^{**}$ at $j-|S|+1$ with an
  isolated vertex, but it does not affect $j$. Consequently, after the
  blue-star embedding $B[j-|S|+1,j]$ consists of $|S|\geq 2$ isolated
  vertices, where $j$ is not in edge-conflict with $s$. Thus, we can
  embed $S$ onto $[j,j-|S|+1]$.

  \case{2.2.2} $|\treeat{i}|+\deg_{R^-}(r)\geq|I|$. By
  \ref{inv:starconflict}, $\treeat{i}$ is not a central-star and must
  hence be a dangling star $B[i,y]$. Since the red-star
  embedding of $S$ used only one vertex of $B[i,j-|S|]$ and since
  $\treeat{i}$ was a central-star after the red-star embedding,
  $\treeat{i}$ must be rooted at $y$ and centered at $i$. Flip
  $\treeat{i}$ to place the center at $y$. Perform the original
  red-star embedding of $S$ from $j$ again. Let us consider the
  interval $B[i,j-|S|]$ that remains for $R^-$. $B[i,j-|S|]$ is an
  independent set. At $i$ we have the original root of $\treeat{i}$,
  which may be in edge-conflict with $r$. The rightmost vertex of
  $B[i,j-|S|]$ is in edge-conflict with $r$. Since
  $|\treeat{j}|\geq|B^{**}|\geq2$ and $\treeat{i}\neq\treeat{j}$ we have
  $\deg_{R^-}(r)\geq|I|-|\treeat{i}|\geq2$. Every subtree of $r$ in
  $R^-$ has size at least $|S|$, and hence we can embed one subtree
  explicitly on a prefix of $[i,j-|S|]$ (which takes care of the vertex
  $i$ which is potentially in edge-conflict) and one subtree explicitly
  on a suffix of $[i,j-|S|]$ (which takes cares of the vertex $j-|S|$
  which is in edge-conflict). The remaining vertices are not
  in edge-conflict, and so we can explicitly complete this partial
  embedding of $R^-$ to a complete embedding of $R^-$.
\end{proof}

\begin{proposition}\label{prop:rec_small_blue_star_equal_ij_used_no_red_star}
  If $B[j-|S|+1,j]$ is a star, $2\leq |S|=|B^{**}|$,
  $\{i,j\}\in\EB$, and $S$ is not a star, then $R$ and $B$
  admit an ordered plane packing onto $[i,j]$.
\end{proposition}
\begin{proof}
  The presence of edge $\{i,j\}\in \EB$ means that $B$ is a tree, rooted
  at $i$ or $j$. We distinguish two cases based on the root of $B$. By
  Lemma~\ref{lem:rec_large_red_star}, we may assume that $R^-$ is not a star.

  \case{1} $B$ is a tree rooted at $i$. We shall flip $B$, and show that
  $B[j-|S|+1,j]$ is no longer a star after the flip. By LSFR, $B^{**}$
  is a smallest subtree of $i$. The largest subtree of $i$ has size at
  least $|S|=|B^{**}|$, and so its root is outside of $[i,i+|S|-1]$.
  Therefore, $\treeatt{[i,i+|S|-1]}{i}$ is an isolated vertex.
  Consequently, after flipping $B$, $\treeatt{[j-|S|+1,j]}{j}$ is
  an isolated vertex, and $B[j-|S|+1,j]$ cannot be a star. If
  $B[i,j-|S|]$ is a star now, use Lemma~\ref{lem:rec_large_blue_star} to
  find an ordered plane packing. Otherwise, none of $S$, $R^-$,
  $B[i,j-|S|]$ and $B[j-|S|+1,j]$ are stars, and we can use
  Lemma~\ref{lem:rec_general} to find an ordered plane packing.

  \begin{figure}[htbp]
    \centering\hfil%
    \subfloat[]{\includegraphics{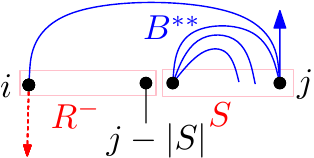}\label{fig:shuffle_subtrees_1}}\hfil%
    \subfloat[]{\includegraphics{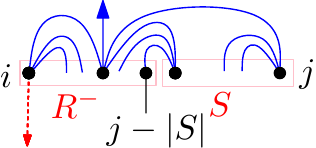}\label{fig:shuffle_subtrees_2}}\hfil%
    \subfloat[]{\includegraphics{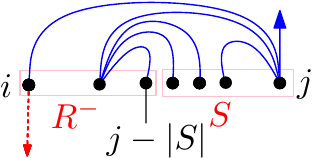}\label{fig:shuffle_subtrees_3}}\hfil%
    \caption{(a) $B^{**}=S$, $\{i,j\}\in E(B)$, and $B$ is rooted at $j$.
      (b) When two subtrees of $j$ are central-stars each with at least 2 vertices.
      (c) When $j$ has a unique maximal subtree, and all other subtrees are singletons or not central-stars.
    }
    \label{fig:shuffle_subtrees}
  \end{figure}

  \case{2} $B$ is a tree rooted at $j$. See
  \figurename~\ref{fig:shuffle_subtrees_1}. Since $B$ is not a star,
  LSFR implies that $B^{**}$ is a dangling star rooted at $j$. That is,
  $B[j-|S|+1,j-1]$ is a central-star, and by LSFR it is a largest
  subtree of $j$. Because $S$ is a smallest subtree of $r$, we have
  $|S|\leq (|I|-1)/2$, and so every subtree of $j$ has size at most
  $|S|-1\leq (|I|-3)/2$. Consequently, $j$ has at least 3 subtrees in
  $B$. We distinguish subcases based on the subtrees of $j$.

  Recall that $j$ has a maximal subtree that is a central-star
  ($B[j-|S|+1,j-1]$). If $j$ has another maximal subtree, then either
  this is a central-star (Case~2.2) or not (Case~2.1). Otherwise,
  $B[j-|S|+1,j-1]$ is the unique maximal subtree of $j$ and either there
  exists another subtree of $j$ that is a central-star on $\ge 2$
  vertices (Case~2.2) or every other subtree of $j$ is a singleton or
  not a central-star (Case~2.3).

  \case{2.1} $j$ has two or more maximal subtrees, but not all of them
  are central-stars. Re-embed $B$ using Algorithm~\ref{alg:embed_t1}
  such that the subtree closest to $j$ is \emph{not} a central-star (we
  only change a tie-breaking rule in Algorithm~\ref{alg:embed_t1}). Then
  $B[j-|S|+1,j]$ is no longer a star, and $B[i,j-|S|]$ does not become a
  star. Use Lemma~\ref{lem:rec_general} to find an ordered plane
  packing.

  \case{2.2} Two or more subtrees of $j$ are central-stars each with at
  least 2 vertices. Let $C_1:=B[j-|S|+1,j-1]$, which is central-star
  subtree of $j$ with at least 2 vertices. Let $C_2$ be another subtree
  of $j$ that is a central-star and has minimal size (possibly 1). We
  re-embed $B$ as follows (see
  \figurename~\ref{fig:shuffle_subtrees_2}). Embed the root of $B$ at
  $j-|C_1|-|C_2|$. Embed $C_1$ onto $[j,j-|C_1|+1]$ and $C_2$ onto
  $[j-|C_1|,j-|C_1|-|C_2|+1]$ each respecting 1SR. Embed all
  remaining subtrees onto $[i,j-|C_1|-|C_2|-1]$ each respecting 1SR.
  Note that $B$ does not obey 1SR because its root has subtrees on both
  sides. However, $B[i,j-|S|]$ and $B[j-|S|+1,j]$ each satisfy both LSFR
  and 1SR. Furthermore, neither $B[i,j-|S|]$ nor $B[j-|S|+1,j]$ is a
  star (since $j$ has at least 3 subtrees); and
  $\treeatt{[j-|S|+1,j]}{j-|S|+1}$ is an isolated vertex.

  Provisionally place $r$ at $i$. We embed $S$ recursively onto
  $[j-|S|+1,j]$. There is no conflict for this embedding since $j-|S|+1$
  is isolated in $B[j-|S|+1,j]$ and not adjacent to $i$. Embed $R^-$
  recursively onto $[i,j-|S|]$. This works because
  $\treeatt{[i,j-|S|]}{i}$ is a singleton or not a central-star. Indeed,
  suppose to the contrary that $\treeatt{[i,j-|S|]}{i}$ is a
  central-star. By construction,
  $\treeatt{[i,j-|S|]}{i}=B[i,j-|C_1|-|C_2|]$ and contains the root of
  $B$. Hence, apart from $C_1$ and $C_2$, all subtrees of the root of
  $B$ are singletons. By the choice of $C_2$, however, $C_2$ is also a
  singleton. Therefore the root of $B$ has only one subtree with at
  least 2 vertices, contradicting our assumption.

  \case{2.3} $j$ has a unique maximal subtree, which is a central-star,
  and every other subtree is either a singleton or not a central-star.
  Recall that $j$ has at least 3 subtrees. Re-embed $B$ such that its
  root is at $j$, an arbitrary smallest subtree is embedded closest to
  $j$, and all other subtrees are embedded according to LSFR (all
  subtrees are embedded recursively by Algorithm~\ref{alg:embed_t1}). In
  particular, $B^{**}$ is now the second subtree of $j$, counting from
  the right. See \figurename~\ref{fig:shuffle_subtrees_3}. As a result,
  $B[j-|S|+1,j]$ is no longer a star, and $B[i,j-|S|]$ does not become a
  star. Note also that $\treeatt{[j-|S|+1,j]}{j-|S|+1}$ becomes an
  isolated vertex (it is a leaf of the dangling star $B^{**}$); and
  $\treeatt{[i,j-|S|]}{i}$ is either an isolated vertex or not a
  central-star.

  Provisionally place $r$ at $i$. Embed $S$ recursively onto
  $[j-|S|+1,j]$. This works because $\treeatt{[j-|S|+1,j]}{j-|S|+1}$ is
  locally isolated and not adjacent to $i$. Embed $R^-$ recursively onto
  $[i,j-|S|]$. The recursive embedding of $R^-$ works because
  $\treeatt{[i,j-|S|]}{i}$ is either an isolated vertex (which is not
  adjacent to the blue vertex on which $s$ was embedded) or not a
  central-star.
\end{proof}

It remains to consider the case where $B[j-|S|+1,j]$ is a star, $2\leq
|S|=|B^{**}|$, $\{i,j\}\in\EB$, and $S$ is a star. We deal with this
case by handling the case where $S$ is a star and $\{i,j\}\in\EB$ in
full generality.

\begin{proposition}\label{prop:rec_small_red_star_ij_used}
  If $S$ is a star and $\{i,j\}\in\EB$, then $R$ and $B$ admit
  an ordered plane packing onto $[i,j]$.
\end{proposition}
\begin{proof}
  Since $\{i,j\}\in\EB$, $B$ is a tree, rooted at $i$ or $j$,
  and we can use symmetry by exchanging the roles of $B$ and $R$
  (\figurename~\ref{fig:smallred_ij_1}). Remove the embedding of $B$.
  Embed $R$ using Algorithm~\ref{alg:embed_t1}, placing its root at $j$.
  Rename $R$ to $B$ and $B$ to $R$. Define $S$ to be a smallest subtree
  of $R$. Since $B$ is rooted at $j$ and $B$ is not a star, there is no
  conflict for embedding $R$ onto $[i,j]$.

  \xxx{MK: This is very concise, maybe elaborate on what covers what?}
  Embedding $R$ onto $[i,j]$ is handled by Lemma~\ref{lem:rec_general},
  Lemma~\ref{lem:rec_unary}, Lemma~\ref{lem:rec_singleton},
  Lemma~\ref{lem:rec_large_blue_star},
  Lemma~\ref{lem:rec_large_red_star}, or
  Proposition~\ref{prop:rec_small_blue_star_larger} unless the situation
  after the color exchange is as follows: $\deg_R(r)\ge 2$, $|S|\geq2$,
  $B[i,j-|S|]$ is not a star, $R^-$ is not a star, and (i) $S$ is a star
  with $|S|\geq 2$ or (ii) $B[j-|S|+1,j]$ is a star and the maximal star
  that contains $B[j-|S|+1,j]$ has size exactly $|S|$. If $S$ is not a
  star then (ii) holds and we can use
  Proposition~\ref{prop:rec_small_blue_star_equal_ij_used_no_red_star}
  to find an ordered plane packing.

  Otherwise, we are in Case~(i) and $S$ is a star. This means that the
  smallest subtree of both $r$ and $b$ is a star on at least two
  vertices and both $R$ and $B$ have at least two subtrees each. Denote
  by $S_B$ a smallest subtree of $B$. By symmetry (possibly exchanging
  roles again), we may assume $|S_B|\ge|S|$.

  \begin{figure}[htbp]
    \centering\hfil%
    \subfloat[]{\includegraphics{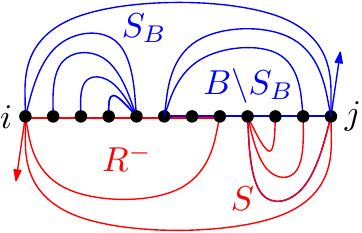}\label{fig:smallred_ij_1}}\hfil%
    \subfloat[]{\includegraphics{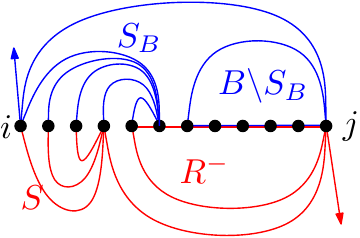}\label{fig:smallred_ij_2}}\hfil%
    \caption{When $R$ and $B$ play symmetric
      roles.\label{fig:smallred_ij}}
  \end{figure}

  We proceed as follows (\figurename~\ref{fig:smallred_ij_2}). Re-embed
  $B$ in the upper halfplane, placing $b$ at $i$, $S_B$ as the closest
  subtree, and the remaining subtrees according to LSFR.

  We first explain how to embed $S$. We will do this in such a way that
  $s$ is embedded on a vertex of $S_B$ at $i+|S|-1$.
  If $S_B$ is a central-star, this re-embedding places its root and
  center at $i+|S_B|$. Since $|S_B|\geq|S|$, now $B[i,i+|S|-1]$ is an
  independent set. Embed $S$ explicitly onto $[i+|S|-1,i]$.
  If $S_B$ is a dangling star, the re-embedding places its root at
  $i+|S_B|$ and its center at $i+1$. If $|S|=2$, then embed $s$ onto
  $i+1$ and its child onto $i$.
  If $|S|\geq3$ and $S$ is a central-star, flip $B[i+1,i+|S_B|]$ to
  put the root of $S_B$ at $i+1$ and the center at $i+|S_B|$ and embed
  $s$ onto $i+|S|-1$ and the children of $s$ onto $[i-|S|-2,i+1]$.
  If $|S|\geq3$ and $S$ is a dangling star, flip $B[i+1,i+|S_B|-1]$ to
  put the center of $S_B$ at $i+|S_B|-1$ and embed $s$ onto $i+|S|-1$,
  its child $s'$ onto $i$, and the children of $s'$ onto
  $[i+1,i+|S|-2]$.

  Next, embed $R^-$ recursively onto $[j,j-|R^-|+1]$. Since $s$ was not
  embedded at $i$, the only obstacle for this recursive embedding is a
  possible conflict, in which case $B^*=\treeatt{[i+|S|,j]}{j}$ is a
  central-star. Since $i+|S|-1$ (which is where we embedded $s$) is
  adjacent only to vertices of $S_B$ and possibly $b$, and since none of
  these vertices are part of $B^*$, the conflict must be a
  degree-conflict. Then $|B^*|\geq3$. As the root $b$ of $B$ is not in
  $[j,j-|R^-|+1]$, we can reorder the subtrees of $B\setminus S_B$
  arbitrarily without having to worry about LSFR on $[j,j-|R^-|+1]$.
  Therefore, we may suppose that \emph{all} subtrees of $b$ are
  central-stars on $\ge 3$ vertices and each of them leads to a
  degree-conflict when taking the role of $B^*=\treeat{j}$ above. Given
  that there are at least two such substars, we may as well choose a
  smallest one, $S_B$ to have its center at $j$. Any other substar can
  take the role intended for $S_B$ in
  \figurename~\ref{fig:smallred_ij_2} originally, its leaves being
  paired up with $S$.

  We claim that then there is no degree-conflict for embedding $R^-$
  onto $[j,j-|R^-|+1]$ recursively. For such a degree-conflict to occur
  we need $\deg_{R^-}(r)+|S_B|-1\ge|R^-|$. So let us argue that this
  does not happen.

  By the choice of $S$ as a minimal size subtree of $r$, we have
  $\deg_R(r)\le(|R|-1)/|S|$. As $S_B$ is a smallest of at least two
  subtrees of $b$, we have $|S_B|\le(|B|-1)/2=(|R|-1)/2$. Together this
  yields
  \begin{align*}
    \deg_{R^-}(r)+|S_B| &= \deg_R(r)-1+|S_B|\\
    &\leq \frac{|R|-1}{|S|}-1+\frac{|R|-1}{2}\\
    &= \frac{|R||S|+2|R|-|S|-2}{2|S|}-1\\
    &= \frac{|R||S|+2|R|-3|S|-2}{2|S|}.
  \end{align*}
  We want to show $\deg_{R^-}(r)+|S_B|\le|R^-|$. So consider the
  expression
  \begin{align*}
    |R^-|-(\deg_{R^-}(r)+|S_B|) &= |R|-|S|-(\deg_{R^-}(r)+|S_B|)\\
    &\geq |R|-|S|-\frac{|R||S|+2|R|-3|S|-2}{2|S|}\\
    &= \frac{|R||S|-2|R|-2|S|^2+3|S|+2}{2|S|}\\
    &= \frac{(|S|-2)(|R|-2|S|-1)}{2|S|},
  \end{align*}
  which is non-negative because $2\le|S|\le(|R|-1)/2$. This proves our
  claim and shows that there is no degree-conflict for embedding $R^-$
  onto $[j,j-|R^-|+1]$ recursively. Therefore at least one of the two
  options provides an ordered plane packing as claimed.
\end{proof}

Lemma~\ref{lem:rec_singleton},
Proposition~\ref{prop:rec_small_blue_star_larger},
Proposition~\ref{prop:rec_small_blue_star_equal_ij_not_used},
Proposition~\ref{prop:rec_small_blue_star_equal_ij_used_no_red_star},
and Proposition~\ref{prop:rec_small_red_star_ij_used} together prove the
following.

\begin{lemma}
  \label{lem:rec_small_blue_star}
  If $B[j-|S|+1,j]$ is a star, then $R$ and $B$ admit an ordered plane
  packing onto $[i,j]$.
\end{lemma}

\section{Embedding the red tree: a small red star}
\label{subsec:small_red_star}
Next, we handle the case where $S$ is a star. We may assume that
$B[i,j-|S|]$ and $B[j-|S|+1,j]$ are not stars. The graph $R^-$ is also
not a star and $|S|\geq 2$.

\begin{proposition}\label{prop:rec_small_red_star_ij_not_used}
  If $S$ is a star and $\{i,j\}\not\in\EB$, then $R$ and $B$
  admit an ordered plane packing onto $[i,j]$.
\end{proposition}
\begin{proof}
  We may assume $\deg_R(r)\geq 2$ by Lemma~\ref{lem:rec_unary}. $S$ can
  be a central-star or a dangling star. We handle these cases
  separately. By Lemma~\ref{lem:rec_singleton}, we may assume that
  $|S|\geq 2$. Let $x$ be such that $|R^-|=|[i,x]|$. Flip $\treeat{j}$ if
  necessary to put the root at $j$. We use the following observation
  several times.

  \begin{observation}
    \label{obs:rec_small_red_star_sc}
    Suppose that we embedded $s$ on a vertex of $\treeat{j}$ and that at
    most $|[i,x]|-1$ rightmost vertices of $B[i,x]$ have been replaced
    by locally isolated vertices. Then $[i,x]$ is not in
    edge-conflict for embedding $R^-$ onto $[i,x]$.
  \end{observation}
  \begin{proof}
    Suppose to the contrary that $[i,x]$ is in edge-conflict for
    embedding $R^-$. Let $y\leq x$ such that
    $B[i,y]=\treeatt{[i,x]}{i}$. Then the root of $B[i,y]$ is in
    edge-conflict with $r$. It cannot be due to an edge to $s$ since
    $\treeat{i}\neq\treeat{j}$. Hence, it must have an edge to the
    outside of $[i,j]$. By 1SR and LSFR, $\treeat{i}$ must then also be
    a (possibly larger) central-star whose root is in edge-conflict with
    $r$. This contradicts the peace invariant for embedding $R$ onto $I$
    and thus concludes the proof.
  \end{proof}

  \case{1} $S$ is a central-star. Since $\{i,j\}\not\in\EB$ we have
  $\treeat{i}\neq\treeat{j}$ and hence this does not change the blue
  vertex at $i$. Use the red-star embedding to embed $s$ onto
  $j$ and the children of $s$ onto the rightmost $\deg_S(s)$
  non-neighbors of $j$ in $[i+1,j-1]$. If $B[i,x]$ is a star now, then
  it was also a star before the red-star embedding (which may
  have modified $B[i,x]$), and we can find an ordered plane packing with
  Lemma~\ref{lem:rec_large_blue_star}. Otherwise, recursively embed
  $R^-$ onto $[i,x]$. By the placement invariant and since
  $\{i,j\}\not\in\EB$, the placement invariant for the recursive
  embedding of $R^-$ holds. Hence, the embedding of $R^-$ fails only if
  (1) there is a conflict for embedding $R^-$ onto $[i,x]$. For the
  embedding of $S$, \ref{sgg:ec} holds and so the embedding works unless
  \ref{sgg:dc} fails, i.e. unless (2) $\deg_S(s)+\deg_B(j)\geq |I|-1$.
  We deal with (1)-(2) next.

  \case{1.1} There is a conflict for embedding $R^-$ onto $[i,x]$.  Let
  $y\leq x$ such that $B[i,y]=\treeatt{[i,x]}{i}$. Then $B[i,y]$ is a
  central-star rooted at a vertex $b^*$. By
  Observation~\ref{obs:rec_small_red_star_sc}, the conflict for
  embedding $R^-$ onto $[i,x]$ is a degree-conflict. In other words,
  $\deg_{[i,x]}(b^*)+\deg_{R^-}(r)\geq|R^-|$. Consequently,
  $|B[i,y]|\geq|R^-|-\deg_{R^-}(r)$. Additionally, $\deg_{R^-}(r)\geq 2$
  since $B[i,x]$ is not a star and $|B[i,y]|\geq3$ by
  Lemma~\ref{lem:degcon3}. Revert to the original blue embedding.  See
  \figurename~\ref{fig:small_red_no_ij_cstar_sc_1}. Note that $B[i,y]$
  is still a central-star.  We distinguish two cases.

  \begin{figure}[b]
    \centering\hfil%
    \subfloat[Case~1.1]{\includegraphics{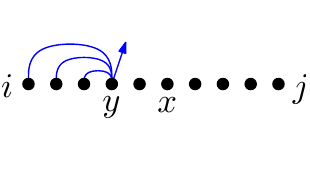}\label{fig:small_red_no_ij_cstar_sc_1}}\hfil%
    \subfloat[Case~1.1.2.1]{\includegraphics{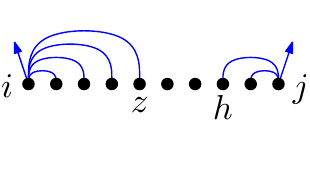}\label{fig:small_red_no_ij_cstar_sc_2}}\hfil%
    \subfloat[Case~1.1.2.1]{\includegraphics{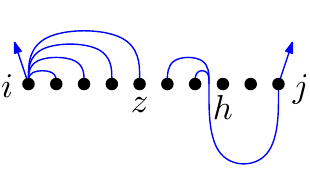}\label{fig:small_red_no_ij_cstar_sc_3}}\hfil%
    \label{fig:small_red_no_ij_1}
    \caption{The case analysis in the proof of
      Proposition~\ref{prop:rec_small_red_star_ij_not_used} (Part~1/3).}
  \end{figure}

  \case{1.1.1} $\treeat{j}$ is not a central-star. Then in particular
  $|\treeat{j}|\geq 3$. Since $\deg_R(r)\geq2$ and
  $|B[i,y]|\geq|R^-|-\deg_{R^-}(r)$ we get by Lemma~\ref{lem:degr} that
  $|B[i,y]|\geq|S|$. If $B[i,y]$ is rooted at $y$ then
  $\treeat{i}=B[i,y]$ by 1SR and we can flip $\treeat{i}$ to put the
  root (and center) at $i$. $B[i,i+|S|-1]$ is now a small blue star.
  Flip $\treeat{j}$ to put its root at the left and embed $R$ onto
  $[j,i]$ with Lemma~\ref{lem:rec_small_blue_star}. This works because
  $j$ is not in edge-conflict with $r$ and $\treeat{j}$ is not a
  central-star.

  \case{1.1.2} $\treeat{j}$ is a central-star. Flip $\treeat{j}$ if
  necessary to put its root (and center) at $j$. If
  $|\treeat{j}|\geq|S|$ then use Lemma~\ref{lem:rec_small_blue_star} to
  find an ordered plane packing. Otherwise $|S|\geq|\treeat{j}|+1$. We
  distinguish two cases.

  \case{1.1.2.1} $\treeat{i}$ is a central-star. Let $z$ such that
  $B[i,z]=\treeat{i}$ and note that $z\geq y$. If necessary, flip
  $B[i,z]$ to put its root at $i$. By the peace invariant, $i$ is not in
  edge-conflict with $r$. Since $i$ is in degree-conflict with $r$ for
  embedding $R^-$ onto $[i,x]$ we have $\deg_B(i)+\deg_{R^-}(r)\geq
  |R^-|$.

  In our first attempt at embedding $R$, we embedded $S$ from $j$ using
  a red-star embedding and tried to embed $R^-$ onto $[i,x]$.  Since
  $|S|\geq2$, the red-star embedding moved all (possibly zero) children
  of $j$ in $\treeat{j}$ to a suffix of $[i,x]$. Since there was a
  degree-conflict for the embedding of $R^-$ onto $[i,x]$, it follows
  that $\deg_{R^-}(r)>|\treeat{j}|-1$. Let $h$ such that
  $B[h,j]=\treeat{j}$. See
  \figurename~\ref{fig:small_red_no_ij_cstar_sc_2}.

  We know $\deg_B(i)+\deg_R(r)\leq |I|-1$ by the peace invariant. It
  follows that $|\treeat{i}|+\deg_{R^-}(r)\leq |I|-1$.  Combining this
  with the degree-conflict at $i$, we obtain
  $|R^-|\leq|B[i,y]|+\deg_{R^-}(r)\leq|\treeat{i}|+\deg_{R^-}(r)\leq|I|-1$.
  Hence, \ref{gg:dc} is satisfied and we can perform a blue-star
  embedding to embed $R^-$ onto $[i,j]$ (which will not embed any vertex
  onto $j$). Before doing so, modify the blue embedding by
  simultaneously shifting $B[h,j-1]$ to $[z+1,z+j-h]$ (redrawing the
  edges to $j$ with biarcs) and $B[z+1,h-1]$ to $[z+j-h+1,j-1]$. See
  \figurename~\ref{fig:small_red_no_ij_cstar_sc_3}. Since
  $\deg_{R^-}(r)>|\treeat{j}|-1$, the blue-star embedding will embed a
  vertex on every child of $\treeat{j}$. Complete the embedding by
  placing $s$ at $j$ and the children of $s$ onto the remainder.

  \case{1.1.2.2} $\treeat{i}$ is not a central-star. Let $w$ such that
  $\treeat{i}=B[i,w]$. Since $B[i,y]$ is a central-star, by 1SR $B[i,y]$
  must be rooted at $i$ and $B[i,w]$ must be rooted at $w$. See
  \figurename~\ref{fig:small_red_no_ij_cstar_sc_4}.

  We claim that $\deg_B(w)=1$. Towards a contradiction, suppose that
  $\deg_B(w)\geq 2$. Recall that $\deg_B(i)+\deg_{R^-}(r)\geq|R^-|$.
  Since $S$ is a smallest subtree of $r$ in $R$, we have
  $\deg_{R^-}(r)\leq (|R^-|-1)/|S|\leq|R^-|/|S|$. Hence, $\deg_B(i)\geq
  |R^-|-\deg_{R^-}(r)\geq(1-1/|S|)|R^-|$. By LSFR, $|\treeat{w}|\geq
  1+2(1+\deg_B(i))= 3+2\deg_B(i)$ and hence
  $|\treeat{w}|\geq3+(2-2/|S|)|R^-|$. Since
  $|R^-|+|S|=|I|\geq|\treeat{w}|$ and $|R^-|\geq|S|$, we obtain
  $|S|\geq3+(2-2/|S|)|R^-|-|R^-|=3+(1-2/|S|)|R^-|\geq3+|S|-2=|S|+1$, a
  contradiction. The claim follows.

  Since $B[i,y]$ is a central-star rooted at $i$, by LSFR
  $\treeat{i}=B[i,w]$ is a star centered at $i$ and rooted at $w$. If
  $w\geq x$ then we can use Lemma~\ref{lem:rec_large_blue_star} to find
  an ordered plane packing. Otherwise, $w\leq x-1$.

  \begin{figure}[thbp]
    \centering%
    \subfloat[Case~1.1.2.2\label{fig:small_red_no_ij_cstar_sc_4}]{\includegraphics{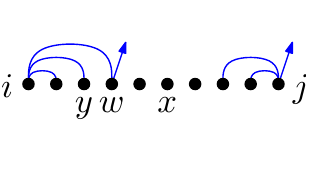}}\hfil%
    \subfloat[Case~1.1.2.2\label{fig:small_red_no_ij_cstar_sc_5}]{\includegraphics{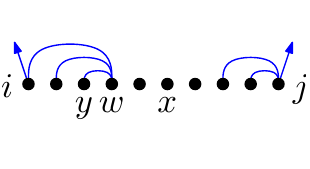}}\hfil%
    \subfloat[Case~1.2]{\includegraphics{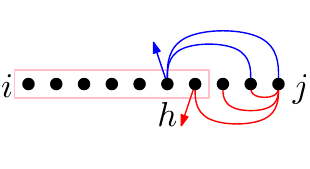}\label{fig:small_red_no_ij_cstar_dc}}\hfil%
    \caption{The case analysis in the proof of
      Proposition~\ref{prop:rec_small_red_star_ij_not_used} (Part~2/3).}
  \end{figure}

  Since there was a conflict for the original embedding, the
  red-star embedding of $S$ from $j$ embeds a child of $s$ onto
  all blue vertices originally at $[w,x]$. Flip $B[i,w]$ to put its root
  at $i$ and center at $w$. See
  \figurename~\ref{fig:small_red_no_ij_cstar_sc_5}. Execute the
  red-star embedding of $S$ from $j$ again. This embeds a child
  of $s$ onto the center of $\treeat{i}$ at $w$ and hence the remaining
  vertices of $\treeat{i}$ form an independent set. Consider the
  now-modified blue embedding at $[i,x]$. The leftmost vertex of
  $B[i,x]$ is the original root of $\treeat{i}$ and may be in
  edge-conflict with $r$. The suffix of $[i,x]$ of size
  $\deg_B(j)\leq|S|-2$ is formed by blue vertices adjacent to $j$ (which
  is where we embedded $s$) that were placed there by the
  red-star embedding of $S$ from $j$. All of these blue vertices
  are in edge-conflict with $r$. However, by the original
  degree-conflict, we know that $\deg_{R^-}(r)\geq 2$ and hence we can
  find an explicit embedding of $R^-$ onto $[i,x]$ that avoids placing
  the root at $i$ or at the suffix of size $|S|-2$. This uses that all
  subtrees of $r$ in $R^-$ have size at least $|S|$.

  \case{1.2} $\deg_S(s)+\deg_B(j)\geq |I|-1$. Then
  $\deg_B(j)\geq|I|-1-\deg_S(s)=|I|-|S|=|R^-|\geq(|I|+1)/2$ and hence
  $\treeat{j}$ has a leaf. Let $h$ such that $B[h,j]=\treeat{j}$. Then
  $|B[h,j]|>|S|$ and so $h\leq x$. If $B[h,j]$ is a star, then we flip
  $B[h,j]$ if necessary to put its center at $j$ and use
  Lemma~\ref{lem:rec_small_blue_star} to find an ordered plane packing.
  Otherwise, $B[h,j]$ is not a star. We claim that then $h<x$. Indeed,
  if $h=x$, then $\deg_B(j)\leq |B[x,j]|-2=|S|+1-2=|S|-1$ and so
  $\deg_S(s)\geq|I|-1-\deg_B(j)\geq|I|-1-|S|+1=|R^-|$, a contradiction.
  The claim follows. Flip $B[h,j]$ to put the root on the left. This
  places a leaf at $j$. Embed $s$ onto $j$ and the children of $s$ onto
  $[j-1,x+1]$. Embed $R^-$ recursively onto $[x,i]$. See
  \figurename~\ref{fig:small_red_no_ij_cstar_dc}. The placement
  invariant holds since $h<x$ and $h$ is the only vertex incident to $j$
  (which is where we embedded $s$).  By LSFR and since $B[h,j]$ is not a
  star, $\treeatt{[i,x]}{x}$ is not a central-star. Hence the peace
  invariant holds and we can complete the packing.

  \case{2} $S$ is a dangling star. Then it is rooted at the child $q$ of
  $s$. Let $Q=\tr(q)$. We will embed $R$ similarly to Case~1. Let $h$
  such that $B[h,j]=\treeat{j}$. We distinguish two cases.

  \case{2.1} Suppose that $B[h,j]$ is not a central-star. Then in
  particular $|B[h,j]|\geq 3$. Let $h'$ be the rightmost neighbor of $j$
  in $[i,j-1]$. If $h'\leq x$, then embed $s$ onto $j-|S|+1$, $q$ onto
  $j$, and the children of $q$ onto $[j-1,j-|S|+2]$. See
  \figurename~\ref{fig:small_red_no_ij_dstar_bj_no_star}. Otherwise,
  embed $s$ onto $h'+1$, $q$ onto $j$, and embed a child of $q$ onto
  every blue vertex of $[h'+2,j-1]$. Use the red-star embedding
  to embed the remaining vertices onto the rightmost
  $\deg_Q(q)-|[h'+2,j-1]|$ non-neighbors of $j$ of $[i+1,h']$. In either
  case, embed $R^-$ recursively onto $[i,x]$. The embedding of $R^-$
  works unless (1) there is a conflict for embedding $R^-$ onto
  $[i,x]$. The embedding of $S$ works unless~\ref{sgg:dc} fails,
  i.e. unless (2) $\deg_Q(q)+\deg_B(j)\geq|I|-2$.

  \begin{figure}[b]
    \centering\hfil%
    \subfloat[Case~2.1]{\includegraphics{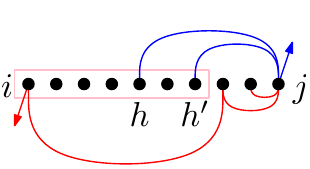}\label{fig:small_red_no_ij_dstar_bj_no_star}}\hfil%
    \subfloat[Case~2.2]{\includegraphics{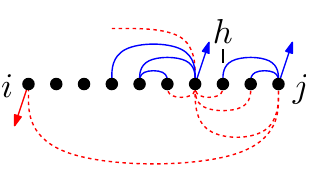}\label{fig:small_red_no_ij_dstar_bj_star_1}}\hfil%
    \subfloat[Case~2.2.2]{\includegraphics{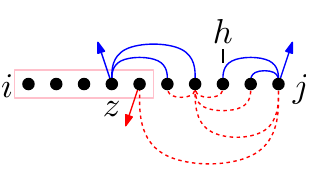}\label{fig:small_red_no_ij_dstar_bj_star_2}}\hfil%
    \subfloat[Case~2.2.2]{\includegraphics{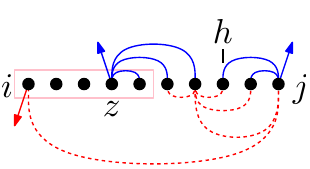}\label{fig:small_red_no_ij_dstar_bj_star_3}}\hfil%
    \label{fig:small_red_no_ij_2}
    \caption{The case analysis in the proof of
      Proposition~\ref{prop:rec_small_red_star_ij_not_used} (Part~3/3).}
  \end{figure}

  \case{2.1.1} There is a conflict for embedding $R^-$ onto $[i,x]$. Let
  $y\leq x$ such that $B[i,y]=\treeatt{[i,x]}{i}$. Then $B[i,y]$ is a
  central-star. By Observation~\ref{obs:rec_small_red_star_sc}, the
  conflict for embedding $R^-$ onto $[i,x]$ is a degree-conflict, and
  hence $|B[i,y]|\geq|R^-|-\deg_{R^-}(r)$. Following the reasoning in
  Case~1.1.1, we see that $|B[i,y]|\geq|S|$ and hence $B[i,i+|S|-1]$ is
  a small blue star after flipping $B[i,y]$ if necessary. Flip $B[h,j]$
  to put its root at $h$ and use Lemma~\ref{lem:rec_small_blue_star} to
  embed $R$ onto $[j,i]$. This works because $j$ is not in edge-conflict
  with $r$ and $\treeat{j}$ is not a central-star.

  \case{2.1.2} $\deg_Q(q)+\deg_B(j)\geq|I|-2$. This case is similar to
  Case~1.2. Since $|S|\leq(|I|-1)/2$ we have
  $\deg_Q(q)=|S|-2\leq(I-5)/2$. Then
  $\deg_B(j)\geq|I|-2-\deg_Q(q)\geq|I|-2-(|I|-5)/2=(|I|+1)/2$. Since
  $B[h,j]$ is not a central-star, we get $h<x$ as in Case~1.2. Let
  $\lambda$ be the number of leaf children of $j$. Then
  $1+\lambda+2((\deg_B(j)-\lambda))\leq|\treeat{j}|\leq|I|-1$. Since
  $\deg_B(j)\geq(|I|+1)/2$ it follows that $1-\lambda+|I|+1\leq|I|-1$
  and hence $\lambda\geq3$. Flip $B[h,j]$ to put its root at $h$.  Since
  $h$ now has $\lambda$ leaf children in $B[h,j]$, in particular $j-1$
  and $j$ are leaves. Embed $s$ onto $j$, $q$ onto $j-1$, and the
  children of $q$ onto $[j-2,x+1]$. Embed $R^-$ recursively onto
  $[x,i]$. Since $B[h,j]$ is not a star by assumption and by LSFR,
  $\treeatt{[i,x]}{x}$ is not a central-star on at least two vertices.
  Hence the peace invariant holds.

  \case{2.2} Suppose that $B[h,j]$ is a central-star. If $h\leq x+1$
  then $B[x+1,j]$ is a star and we can find an ordered plane packing by
  Lemma~\ref{lem:rec_small_blue_star}. Otherwise $h\geq x+2$. Flip
  $\treeat{h-1}$ if necessary to put its root at $h-1$. Embed $s$ onto
  $j$, $q$ onto $h-1$, and a child of $q$ on every vertex in $[h,j-1]$.
  Use the red-star embedding to embed the remaining children of
  $q$ onto the rightmost $\deg_Q(q)-|[h,j-1]|$ non-neighbors of $h-1$ in
  $[i+1,h-2]$. Embed $R^-$ recursively onto $[i,x]$. See
  \figurename~\ref{fig:small_red_no_ij_dstar_bj_star_1} for the
  situation before the cleanup step of the red-star
  embedding. The embedding of $R^-$ works unless (1) there is a conflict
  for embedding $R^-$ onto $[i,x]$. The embedding of $S$ works
  unless~\ref{sgg:dc} fails, i.e. unless (2)
  $\deg_Q(q)+\deg_B(h-1)\geq|I|-2$.

  \case{2.2.1} There is a conflict for embedding $R^-$ onto $[i,x]$. Let
  $y\leq x$ such that $B[i,y]=\treeatt{[i,x]}{i}$. Then $B[i,y]$ is a
  central-star. By Observation~\ref{obs:rec_small_red_star_sc}, the
  conflict is a degree-conflict. Revert to the original blue embedding
  (before the red-star embedding in Case~2.2) and note that
  $B[i,y]$ is still a central-star. We proceed similarly to Case~1.1.2.

  \case{2.2.1.1} $\treeat{i}$ is a central-star. Let $z$ such that
  $B[i,z]=\treeat{i}$ and note that $z\geq y$. If necessary, flip
  $B[i,z]$ to put its root at $i$. By the peace invariant, $i$ is not in
  edge-conflict with $r$. Since $i$ is in degree-conflict with $r$ for
  embedding $R^-$ onto $[i,x]$ we have $\deg_B(i)+\deg_{R^-}(r)\geq
  |R^-|$.

  We blue-star embed $R^-$ starting from $i$ with
  $\varphi=(z+1,\ldots)$. Let us argue that the conditions for the
  blue-star embedding hold. The peace invariant guarantees \ref{gg:ec}
  and $\deg_B(i)+\deg_R(r)\leq |I|-1$. It follows that
  $|\treeat{i}|+\deg_{R^-}(r)\leq |I|-1$, which is the second inequality
  of \ref{gg:dc}. The first inequality of \ref{gg:dc} holds by the
  degree-conflict condition. \ref{gg:int} holds by construction, making
  \ref{gg:cs} trivial.  Hence, the conditions are satisfied and we can
  perform the blue-star embedding as described.

  Since we attain the first inequality in~\ref{gg:dc} strictly, the
  blue-star embedding does not exhaust all vertices in $B[i,z]$. Indeed,
  $\deg_B(i)\geq |R^-|-\deg_{R^-}(r)$, while the blue-star embedding
  embeds only $|R^-|-\deg_{R^-}(r)-1$ vertices on the neighbors of $i$.
  Perform the blue-star embedding of $R^-$ onto $[i,j]$. This leaves an
  interval containing $j$ (since the blue-star-embedding always leaves
  at least one vertex) and at least one locally isolated vertex
  (originating from $B[i+1,z]$). Embed $s$ onto $j$, $q$ onto this
  locally isolated vertex, and the children of $q$ onto the remainder to
  complete the embedding.

  \case{2.2.1.2} $\treeat{i}$ is not a central-star. We proceed
  similarly to Case~1.1.2.2. Let $w$ such that $\treeat{i}=B[i,w]$. The
  exact same argument as in Case~1.1.2.2 shows that $B[i,w]$ is a star
  rooted at $w$ and centered at $i$. If $w\geq x$ then we can use
  Lemma~\ref{lem:rec_large_blue_star} to find an ordered plane packing.
  Otherwise, $w\leq x-1$.

  Since there was a conflict for the original embedding, the red-star
  embedding of (the remainder of) $Q$ from $h-1$ embeds a child of $q$
  onto all blue vertices originally at $[w,x]$. Flip $B[i,w]$ to put its
  root at $i$ and center at $w$. Embed $s$ onto $j$, $q$ onto $h-1$, and
  a child of $q$ onto all vertices in $[h,j-1]$.  Execute the red-star
  embedding of the remainder of $Q$ from $h-1$ onto $[h-2,i+1]$
  again. This embeds a child of $s$ onto the center of $\treeat{i}$ and
  hence the remaining vertices form an independent set. Consider the
  now-modified blue embedding at $[i,x]$.  The leftmost vertex of
  $B[i,x]$ is the original root of $\treeat{i}$ and may be in
  edge-conflict with $r$. We embedded a child of $q$ onto all neighbors
  of $j$ (which is where we embedded $s$), and hence there are no
  further edge-conflicts. Hence, we can embed $R^-$ explicitly onto
  $[x,i]$.

  \case{2.2.2} $\deg_Q(q)+\deg_B(h-1)\geq|I|-2$. Let $z$ such that
  $B[z,h-1]=\treeat{h-1}$. It is possible that $z=i$ and
  $\treeat{i}=\treeat{h-1}$. Analogously to Case~2.1.2 we get
  $\deg_B(h-1)\geq(|I|+1)/2$ and that $h-1$ has at least $3$ leaf
  children. It follows that $z<x$. Recall that $h\geq x+2$. Flip
  $B[z,h-1]$ to put its root at $z$. If $z=i$ and $B[i,x]$ is now a
  star, use Lemma~\ref{lem:rec_large_blue_star} to find an ordered plane
  packing. Otherwise, flipping $B[z,h-1]$ placed a leaf child of $z$ at
  $h-1$. Embed $s$ onto $j$, $q$ onto $h-1$, and the children of $q$
  onto $[j-1,h]$ and $[h-2,x+1]$. This works because $z<x$ and $h\geq
  x+2$.

  We first try to embed $R^-$ recursively onto $[x,i]$. See
  \figurename~\ref{fig:small_red_no_ij_dstar_bj_star_2}. Since $z<x$,
  this works unless $\treeatt{[i,x]}{x}$ is a central-star, which
  implies that $B[z,h-1]$ is a central-star by LSFR. In this scenario we
  already handled the case $z=i$ and so we may assume
  $\treeat{i}\neq\treeat{z}$. Embed $R^-$ recursively onto $[i,x]$. See
  \figurename~\ref{fig:small_red_no_ij_dstar_bj_star_3}. By the
  placement invariant, this works unless there is a conflict for
  embedding $R^-$ onto $[i,x]$.

  So suppose there is a conflict for embedding $R^-$ onto $[i,x]$.
  Since $z<x$ and $\treeat{i}\neq\treeat{z}$, we have
  $\treeatt{[i,x]}{i}=\treeat{i}$ and hence $\treeat{i}$ is a
  central-star. By the peace invariant, the root of $\treeat{i}$ is not
  in edge-conflict with $r$. Flip $\treeat{i}$ if necessary to put its
  root at $i$. Then $i$ is in degree-conflict with $r$ and hence
  $\deg_B(i)+\deg_{R^-}(r)\geq|R^-|$. Adding this inequality to the
  inequality in the assumption (replacing $h-1$ by $z$ due to our
  flipping), we get
  $\deg_B(i)+\deg_{R^-}(r)+\deg_Q(q)+\deg_B(z)\geq|I|-2+|R^-|$. Since
  $\treeat{i}$, $\treeat{z}$, and $\treeat{j}$ are all different we have
  $\deg_B(i)+\deg_B(z)\leq|I|-3$. Hence
  $\deg_{R^-}(r)+\deg_Q(q)\geq|I|-2+|R^-|-|I|+3=|R^-|+1$. Since
  $|S|=\deg_Q(q)+2$ we get $|S|+\deg_{R^-}(r)\geq|R^-|+3$. Since $S$ is
  a smallest subtree of $r$, we have $\deg_{R^-}(r)\leq (|R^-|-1)/|S|$
  and hence $|R^-|\geq|S|\deg_{R^-}(r)$. It follows that
  $|S|+\deg_{R^-}(r)\geq|S|\deg_{R^-}(r)+3$, which has no solution for
  $|S|\geq1$ and $\deg_{R^-}(r)\geq1$. We conclude that there is no
  conflict for embedding $R^-$ onto $[i,x]$, as desired.
\end{proof}

\noindent
Propositions~\ref{prop:rec_small_red_star_ij_not_used} and
\ref{prop:rec_small_red_star_ij_used} together prove the following.
\begin{lemma}
  \label{lem:rec_small_red_star}
  If $S$ is a star, then $R$ and $B$ admit an ordered plane packing
  onto $[i,j]$.
\end{lemma}

\noindent
Finally, Lemmata~\ref{lem:rec_general}, \ref{lem:rec_large_blue_star},
\ref{lem:rec_large_red_star}, \ref{lem:rec_small_blue_star}, and
\ref{lem:rec_small_red_star} together prove
Theorem~\ref{thm:main}. \label{proofend}

\bibliographystyle{mh-url}\bibliography{bibliography}

\end{document}